\documentclass[aps, pra, reprint, superscriptaddress]{revtex4-2}
\def\preprint{0}
\def\draft{0}
\def\onlymaintext{0}
\input{style/preamble.sty}
\let\originalsection\section
\let\originalsubsection\subsection
\let\originalsubsubsection\subsubsection

\newif\ifinappendix
\inappendixtrue %

\renewcommand{\subsection}[1]{%
  \ifinappendix
    \originalsubsection*{#1}%
  \else
    \originalsubsection{#1}%
  \fi
}

\renewcommand{\section}[1]{%
  \ifinappendix
    \originalsection*{#1}%
  \else
    \originalsection{#1}%
  \fi
}

\renewcommand{\subsubsection}[1]{%
  \ifinappendix
    \originalsubsubsection*{#1}%
  \else
    \originalsubsubsection{#1}%
  \fi
}

\let\originalappendix\appendix
\renewcommand{\appendix}{%
  \originalappendix
  \inappendixfalse %
}

\begin{document}

\newcommand{\sect}[1]{~\\\noindent{\large{\bf{#1}}~\\~\\ \noindent}}
\newcommand{\subsect}[1]{\noindent{\bf{#1}}\\ \noindent}
\newcommand{\subsubsect}[1]{{\emph{#1}.}}

\newcommand{\mytitle}{Exponentially Decaying Quantum Simulation Error with Noisy Devices}

\title{\mytitle}
\author{Jue Xu}
\email{These authors contributed equally to this work.}
\affiliation{\HKU}
\author{Chu Zhao}
\email{These authors contributed equally to this work.}
\affiliation{\IIIS}
\author{Junyu Fan}
\affiliation{International Quantum Academy, Shenzhen, Guangdong, China}
\author{Qi Zhao}
\email{zhaoqi@cs.hku.hk}
\affiliation{\HKU}
\date{\today}
\begin{abstract}
    Quantum simulation is a promising way toward practical quantum advantage,
but noise in current quantum hardware poses a significant obstacle.
We prove that not only the physical error but also the algorithmic error in a single Trotter step decreases exponentially with the circuit depth.
This theoretical finding is validated by our numerical results over various Hamiltonians, initial states, and noise channels.
Furthermore,
we derive the optimal number of 
Trotter steps and the noise requirement to guarantee total simulation precision.
To explicitly show the requirements for robust quantum simulation, we plot a phase diagram of the accumulated error in terms of circuit depth and noise rate.
At last, we demonstrate that our improved error analysis leads to significant resource-saving for fault-tolerant Trotter circuits.
By addressing these aspects, 
this work provides fresh and systematic insight on the practical quantum advantage through quantum simulation.

\end{abstract}
\maketitle

\section{Introduction}
Quantum dynamics simulation is a cornerstone application of quantum computation \cite{feynmanSimulatingPhysicsComputers1982, feynmanQuantumMechanicalComputers1985, georgescuQuantumSimulation2014} 
and is widely recognized as a leading candidate for achieving practical quantum advantage \cite{ciracGoalsOpportunitiesQuantum2012, daleyPracticalQuantumAdvantage2022, bernienProbingManybodyDynamics2017, monroeProgrammableQuantumSimulations2021, altmanQuantumSimulatorsArchitectures2021, kimEvidenceUtilityQuantum2023, wangRealizationFractionalQuantum2024, guoSiteresolvedTwodimensionalQuantum2024}.
Efficient quantum dynamics simulation would render general studies of quantum physics tractable \cite{lloydUniversalQuantumSimulators1996, jordanQuantumAlgorithmsQuantum2012, mcardleQuantumComputationalChemistry2020}, which are notoriously difficult for classical computers.
Furthermore, quantum dynamics simulation also serves as an indispensable subroutine for a plethora of quantum algorithms,
such as 
probing ground state energies \cite{leeEvaluatingEvidenceExponential2023},  %
solving linear algebra problems \cite{harrowQuantumAlgorithmSolving2009, liuEfficientQuantumAlgorithm2021}, 
and optimization problems \cite{farhiQuantumApproximateOptimization2014, albashAdiabaticQuantumComputation2018}.

Since the seminal digital quantum dynamics simulation algorithm based on the Trotter formula for local Hamiltonians proposed by Lloyd \cite{lloydUniversalQuantumSimulators1996}, 
novel post-Trotter techniques \cite{childsHamiltonianSimulationUsing2012, berrySimulatingHamiltonianDynamics2015,  berryHamiltonianSimulationNearly2015,lowOptimalHamiltonianSimulation2017, haahQuantumAlgorithmSimulating2018} %
have been developed and are asymptotically optimal in key parameters.
Nonetheless, the variants of the Trotter formula \cite{berryEfficientQuantumAlgorithms2007,campbellRandomCompilerFast2019, childsFasterQuantumSimulation2019, anTimedependentUnboundedHamiltonian2021,bosseEfficientPracticalHamiltonian2025,fangTrotterErrorManybody2025} 
remain primary schemes for near-term quantum devices to achieve practical quantum advantage, owing to their mild hardware requirements and decent theoretical performance.
Despite the extensive analyses of Trotter error for showing quantum advantage \cite{childsFirstQuantumSimulation2018,childsNearlyOptimalLattice2019,childsTheoryTrotterError2021,zhaoHamiltonianSimulationRandom2021,sahinogluHamiltonianSimulationLowenergy2021,gongComplexityDigitalQuantum2024,zlokapaHamiltonianSimulationLowenergy2024,zhaoEntanglementAcceleratesQuantum2025,mizutaTrotterizationSubstantiallyEfficient2025}, there still lacks systematic analysis of noisy Trotter circuits.

Although quantum dynamics simulation is as powerful as universal quantum computation that implies quantum advantage \cite{feynmanSimulatingPhysicsComputers1982}, its full potential necessitates fault-tolerant quantum computing.
Notably, ubiquitous noise detrimentally affects quantum computations, making them less accurate as systems scale up and increasingly susceptible to classical simulation.
While a fully fault-tolerant quantum computer remains on the horizon, 
many quantum simulation experiments have been recently implemented \cite{zhangObservationManyBodyDynamical2017,bernienProbingManybodyDynamics2017,guoSiteresolvedTwodimensionalQuantum2024} on noisy intermediate-scale quantum (NISQ) devices \cite{preskillQuantumComputingNISQ2018, bhartiNoisyIntermediatescaleQuantum2022}.
Significantly, Refs.~\cite{kimEvidenceUtilityQuantum2023} simulated the dynamics of the Ising model by Trotterized circuits,
but this advantage was challenged by classical simulation techniques, including tensor-network methods for shallow circuits \cite{begusicFastConvergedClassical2024, tindallEfficientTensorNetwork2024} and the Pauli path integral method \cite{shaoSimulatingNoisyVariational2024, fontanaClassicalSimulationsNoisy2025, martinezEfficientSimulationParametrized2025} %
as shown in \cref{fig:advantage}. 
Therefore, 
it is imperative to determine the noise requirement for an accurate Trotter simulation with potential quantum advantage. 
Even in the fault-tolerant quantum computing era, with the help of quantum error correction techniques, it is also essential to know the required logical noise rate in a Trotter simulation task. 

\begin{figure}[!t]
    \centering
    \includegraphics[width=0.99\linewidth]{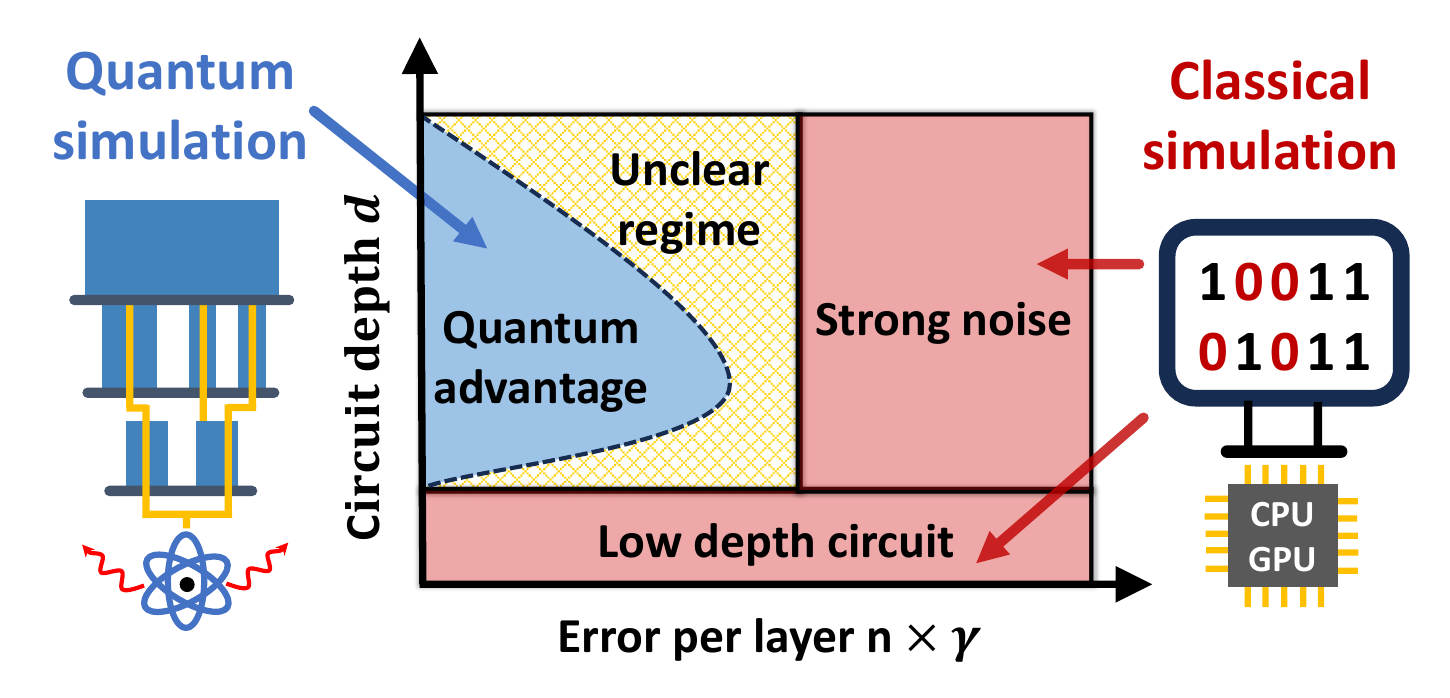}
    \caption{
    The regime of potential quantum advantage by quantum simulation.
    For low-depth quantum circuits or strong noise circuits, 
    the quantum advantage would be diminished by classical algorithms.
    }
    \label{fig:advantage}
\end{figure}

The accuracy of a noisy Trotter simulation task is affected by two competing factors, i.e., circuit depth and noise rate. 
Specifically, increasing the Trotter steps (circuit depth) suppresses the Trotter error (algorithmic error) but introduces more physical error induced by circuit noise. 
Although Refs.~\cite {kneeOptimalTrotterizationUniversal2015,endoMitigatingAlgorithmicErrors2019, hakkakuDataEfficientErrorMitigation2025} have examined both physical and algorithmic errors in noisy Trotter circuits, 
their worst-case error analyses quantified by operator distance between ideal evolution and Trotterized unitaries yield pessimistic estimation due to overlooking crucial evolved state information.
Essentially, circuit noise that drives evolved states towards the maximally mixed state causes empirical errors deviating from the worst-case error analyses \cite{ben-orQuantumRefrigerator2013,muller-hermesRelativeEntropyConvergence2016,francaLimitationsOptimizationAlgorithms2021}, resulting in an overestimation of the resources required for an accurate Trotter simulation task.  
Consequently, to pave the way for demonstrating quantum advantage through Trotter simulation, a refined and state-dependent error analysis is indispensable yet remains unexplored.

Our work demonstrates the robustness of Trotter simulation with noisy circuits through several novel contributions. 
First, we prove that both algorithmic and physical errors in noisy Trotter simulation decay exponentially with Trotter steps (circuit depth) when evolved state information is appropriately considered. 
We validate this unexpected noise-error relationship through various numerical experiments. 
As practical applications of our findings,
we determine the optimal Trotter number of noisy circuits and establish the noise requirements to guarantee simulation precision, 
providing quantitative insights into the trade-off between algorithmic and physical errors. 
As a result, in medium-scale fault-tolerant Trotter simulations, 
our analysis can reduce the resource requirement 
(i.e., the circuit depth times the number of physical qubits) by up to 60\%,
compared with the worst-case analyses. 
Finally, we present a phase diagram of the total error in a noisy Trotter circuit, 
delineating the regime of potential quantum advantage.

\section{Results}
\subsection{Setting: noisy Trotter circuit}

We embark on studying the robustness of noisy Trotter simulation with the following setting.
Given an $n$-qubit local Hamiltonian $H = \sum_{l=1}^L H_l$ with $L=\poly(n)$ Pauli terms acting on at most $k$ qubits and evolution time $t$, 
quantum dynamics simulation aims to approximate the real-time evolution operator $U(t):=e^{-\ii Ht}$ by a quantum circuit.
For this goal, the first-order Lie-Trotter formula $\pf_1$
is a product of the dynamics of each term, 
that is
\begin{equation}\label{eq:pf1}
    \pf_1(\deltat):= e^{-\ii H_1 \deltat} e^{-\ii H_2 \deltat} \cdots\, e^{-\ii H_L \deltat} 
    = \prod^{L}_{l=1} e^{-\ii H_l \deltat},
\end{equation}
where each $e^{-\ii H_l \deltat}$ is assumed to be implemented efficiently by elementary quantum gates and $\deltat=t/r$ is a short time sliced by $r$ segments \cite{lloydUniversalQuantumSimulators1996}.
This type of formula is also called product formula (denoted as PF) in a more general sense.  
As the Hamiltonian terms are not commutative to each other in general,
the algorithmic error introduced by the Trotter formula is called Trotter error. 
And it has the upper bound $\opnorm{\pf_1(\deltat) - U(\deltat)}=\bigO(\deltat^2)$ measured by the spectral norm 
$\opnorm{\cdot}$.
Consequently, the accumulated error between ideal evolution $U(t)$ and $r$-step Trotter circuit $\pf_1^r(t/r)$ scales as $\bigO(t^2/r)$,
which can be suppressed by increasing $r$.

Besides, the Trotter error can also be improved by patterns of products.
For instance, the second-order Suzuki-Trotter formula (PF2), i.e.,
$\pf_2(\deltat):= \pf_1(\deltat/2) \pf_1^{\leftarrow}(\deltat/2)$,
has the error scales $\bigO(\deltat^3)$,
where $\pf_1^{\leftarrow}$ denotes the product in the reverse order.
More generally,
the $p$th-order product formulas (PF$p$) are defined recursively as
\begin{equation}\label{eq:high_pf}
    \pf_{p}(\deltat): = \pf_{p-2} (u_p \deltat)^2 \pf_{p-2}((1-4u_p)\deltat) \pf_{p-2} (u_p \deltat)^2,
\end{equation}
where $u_p:=1/(4-4^{1/(p-1)})$ and $p$ is a positive even integer \cite{suzukiGeneralTheoryFractal1991}.
In this way, the Trotter steps of the $p$th-order product formula for achieving the worst-case precision $\varepsilon$ is $\bigO(\acmm_p^{1/p}t^{1+1/p}/\varepsilon^{1/p})$ 
where $\acmm_p$ is the $p$th-order nested commutator norm \cite{childsTheoryTrotterError2021}.
Though the higher-order formula has significantly smaller errors,
it leads to overhead in the circuit depth. 
Specifically, the number of circuit layers in one Trotter step is proportional to $\Upsilon_p=2\cdot 5^{p/2-1}$.

\begin{figure}[!t]
    \centering
    \includegraphics[width=0.95\linewidth]{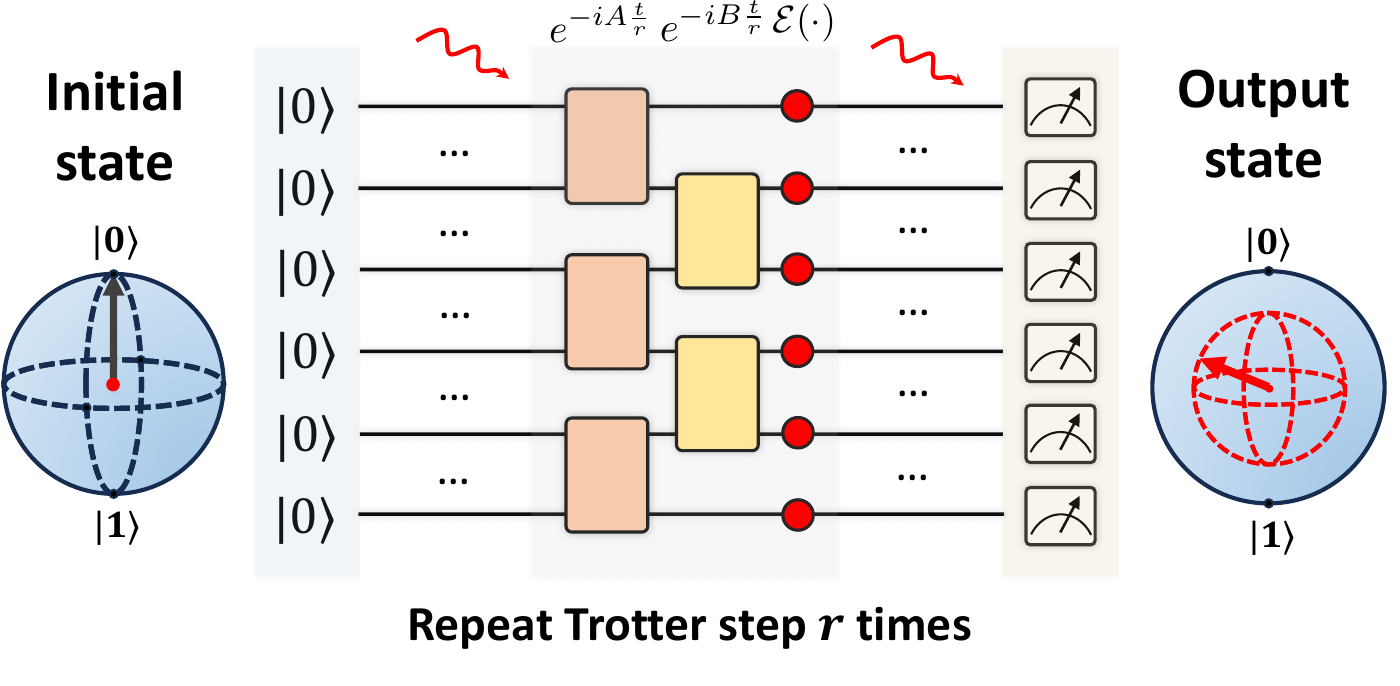}
    \caption{
    Our model of noisy Trotter simulation circuits.
    The noisy Trotter circuit is the repeated Trotter steps interspersed by a one-qubit depolarizing noise channel marked by solid red dots. 
    The gray box covers a Trotter step repeated $r$ times.
    The depolarizing noise renders states close to the maximally mixed state, 
    which results in error decay.
    }
    \label{fig:noisy_circuit}
\end{figure}

On the other hand, the other source of error, physical error, comes from the imperfection of quantum circuits.
In the conventional noisy quantum circuit model, every qubit undergoes a one-qubit depolarizing channel $\E^{\depo}_{\gamma}$ with noise rate $\gamma \in [0,1]$ after applying a gate
\begin{align}\label{eq:single_depolarizing}
    \E^{\depo}_{\gamma}(\rho) 
    &:= (1-\gamma)\rho + \frac{\gamma}{3} \qty(X\rho X + Y \rho Y + Z\rho Z) ,
\end{align}
where $\rho$ is the one-qubit density matrix 
\footnote{
The state $\rho$ is left alone with probability $1-\gamma$, and the operators $X$, $Y$ and $Z$ applied each with probability $\gamma/3$.
Equivalently, $\E_{\gamma}^{\depo}(\rho)=(1-\gamma^\prime)\rho + \frac{\gamma^\prime}{2} \I $ where $\gamma^\prime = \frac{4}{3}\gamma$
}.
For an $n$-qubit circuit, the one-qubit (local) depolarizing noise channel is $\E^{n}_{\gamma}(\rho) :=\bigotimes^n \E^{\depo}_{\gamma}(\rho)$ applied 
after a layer of gates.
We neglect the superscript `$\depo$' 
unless otherwise specified.
Whilst this error model is simplistic, 
it is an effective model of the errors seen in some current quantum hardware.
To analyze the realistic experiments more accurately, more fine-grained categories of noise models are necessary \cite{clintonHamiltonianSimulationAlgorithms2021}.
As a complement, we discuss other common noise channels, such as the dephasing and the amplitude-damping noise channel, 
in the last section and more details in
\ifnum\onlymaintext=0
    \cref{apd:noise_channels}.
\else
    the Supplementary Materials \cite{seesm}.
\fi

Taking these elements together, we formally introduce the setting of our noisy Trotter circuit model as follows.
\begin{definition}[Noisy Trotter]\label{def:noisy_circuit}
    Given the $p$th-order Trotter circuit channel 
    $\pfc_p(\rho):=\pf_p(t/r)\rho\pf_p^\dagger(t/r)$ 
    and the one-qubit (local) depolarizing noise channel $\E^n_\gamma$ on $n$ qubits with noise rate $\gamma$ after every Trotter step as \cref{fig:noisy_circuit},
    the quantum channel $\C_{p,\gamma,r}$ of the whole circuit is 
    \begin{equation}
        \C_{p,\gamma,r}:= 
        \E^n_\gamma \circ \pfc_p \circ \cdots \circ \E^n_\gamma \circ \pfc_p
        \equiv \qty(\E^n_\gamma \circ \pfc_p)^r,
    \end{equation}
    where $\circ$ denotes the composition of channels and $r$ is the number of Trotter steps or called Trotter number.
\end{definition}

\subsection{Exponentially decaying one-step errors in noisy Trotter circuit}

Roughly, the goal of a quantum simulation experiment is to achieve a certain simulation precision $\varepsilon$ given the system Hamiltonian, 
with maximal simulation time and minimal circuit depth.
Ideally, without physical noise, the algorithmic (Trotter) error can be arbitrarily suppressed by increasing the Trotter number $r$.
Despite the presence of noise, physical errors tend to accumulate as circuit depth increases. 
Therefore, it is essential to analyze the tradeoff between algorithmic and physical errors, as well as the impact of noise on algorithmic error. 
In this context, errors need to be formally defined to evaluate the robustness of a noisy Trotter circuit.

First, the one-step \emph{physical error} is defined as the trace distance between the state $\rho$ with or without undergoing a one-qubit (local) depolarizing noise channel $\E_\gamma^n$
\begin{equation}\label{eq:phy_err}
    \epsilon_{p,\gamma}^{\phy}(d) := \trnorm{ \rho_d - \E^n_\gamma(\rho_d) },
\end{equation}
where $\rho_d$ is the state after the $d$th noisy $p$th-order Trotter step and $d\in[r]$. 
In the state-independent (worst-case) analysis,
the physical error induced by the one-qubit $\gamma$ depolarizing channel $\E_{\gamma}^{n}$ on $n$ qubits has the upper bound 
$\dnorm{\E_{\gamma}^{n} - \IC} \le 2n\gamma $.
However, this upper bound by the diamond norm $\dnorm{\cdot}$ is pessimistic 
because it overlooks that the state $\rho$ in a noisy Trotter circuit is getting close to the maximally mixed state $I/2^n$.

    Conventionally, the algorithmic error of the $p$th-order product formula is defined as 
    $\opnorm{U - \pf_p}$
    where the spectral norm captures the worst-case (initial state) analysis.
    To analyze the effect of a general mixed state, we define the one-step \emph{algorithmic error} as the trace distance between the mixed state $\rho_d$ 
    after an ideal one-step evolution $\U$ for a short time $t/r$ and one-step noiseless $p$th-order Trotter formula $\pf_p(t/r)$, i.e.,
    \begin{equation}\label{eq:alg_err}
        \epsilon_{p,\gamma}^{\alg}(d) := \trnorm{ \mathcal{U}(\rho_d) - \tilde{\mathcal{U}}_p(\rho_d)},
    \end{equation}
    where we use the channel notation 
    $\mathcal{U}(\rho):= U\rho U^{ \dagger}$ 
    for consistency with the definition of physical error.
It is bounded by the worst-case error up to a 2 factor, that is,
$\epsilon_{p,\gamma}^{\alg}(d)\le 2 \opnorm{U - \pf_p}$.

Before proving theoretical bounds of the errors, we take a look at the empirical errors compared with the worst-case bounds.
As a crucial observation of this work,
\cref{fig:decay} shows that
algorithmic and physical errors in one Trotter step decay exponentially with the Trotter steps and are much smaller than the worst-case bounds.
Moreover, the stronger the noise (darker color) is, the faster both errors decay.
Additionally, contrary to the prefactors of physical error (the intercepts of the lines) increasing with stronger noises $\gamma$, 
the prefactor of the algorithmic error is independent of $\gamma$.
As a showcase of typical simulation tasks, we focus on the numerical results of the transverse field Ising (TFI) Hamiltonian in the main text,
and we refer to 
\ifnum\onlymaintext=0
    \cref{apd:more_hamiltonians} 
\else
    the Supplementary Materials \cite{seesm}.
\fi
for more numerical results of other common physical Hamiltonians
including the Heisenberg model with power-law interaction, the 1D Fermi-Hubbard model and the Hydrogen chain, as well as common observables.
Next, we are going to prove these phenomena rigorously.

\begin{figure}[!t]
    \centering
    \includegraphics[width=\autowidth\linewidth]{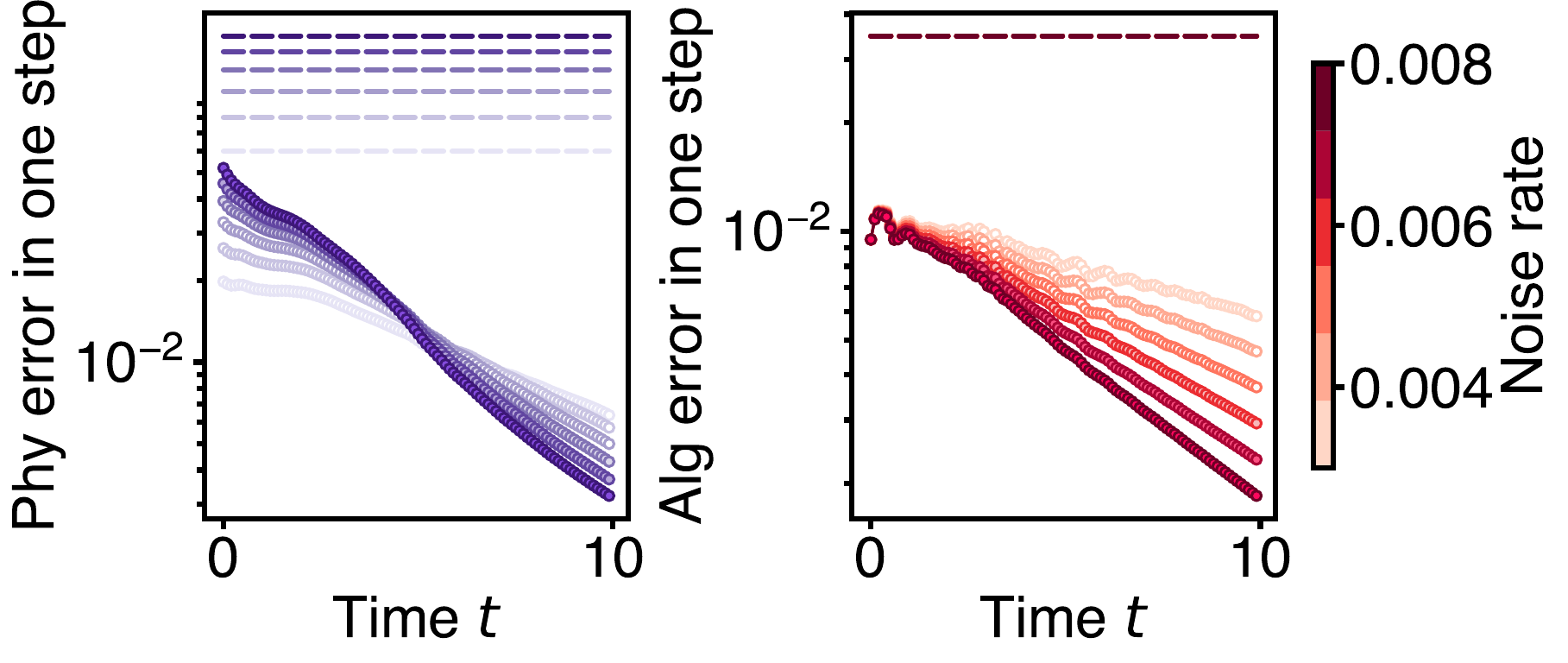}
    \caption{
    The one-step physical and algorithmic errors exponentially decay with Trotter steps.
    The physical error (\emph{left panel}), algorithmic error (\emph{right panel})
    at every Trotter step with different noise rates $\gamma\in[0.003,0.008]$. 
    The $y$-axis is in the log scale.
    The color of the line indicates the strength of the noise. 
    We take the TFI Hamiltonian $H_{\tfi}=J \sum_{j=1}^n X_jX_{j+1}+ h\sum_{j=1}^n Z_j$ with parameters ($J$=2, $h$=1, $n$=10), the periodic boundary condition. 
    The Hamiltonian can be grouped into two commuting parts as 
    $H_{\tfi}=H_{X}+H_{Z}$ where $H_X = J \sum_{j=1}^n X_jX_{j+1}$ and $H_Z = h\sum_{j=1}^n Z_j$.
    We use the second-order product formula PF2
    and take the initial state as $\ket{0}^{\otimes n}$, 
    evolution time $t=n$, and Trotter number $r=100$. 
    The dashed lines are the worst-case theoretical upper bounds with the corresponding noise rates.
    We adopt the upper bound $2n\gamma$ by the diamond norm for the worst-case physical error,
    while the commutator bound \cite{childsTheoryTrotterError2021} is used for the worst-case algorithmic (Trotter) error.
    See supplementary material for details on the worst-case analysis.
    }
    \label{fig:decay}
\end{figure}

\subsection{Upper bound of physical error}

We first sketch the proof of the exponential decay of physical error.
The intuition is that the depolarizing channel renders the evolved state close to the maximally mixed state
such that the effect of noise becomes attenuated.
Since the local depolarizing channel transforms the state into a mixture of locally depolarized states, 
we quantify the physical error by measuring the distance between the evolved state and the set of locally depolarized states 
i.e. $\sum_{F:\abs{F}=1} \calD\qty(\rho\|  \frac{1}{n}  \rho_{\overline{F}}\otimes \frac{I}{2})$. 
This distance is further bounded by the distance between the evolved state and the global maximally mixed state  $\calD\qty(\rho \| \frac{I}{2^n})$, 
which has been proved to decay exponentially with the noisy circuit depth \cite{francaLimitationsOptimizationAlgorithms2021}
i.e.,
\begin{equation}
    \calD(\C_{p,\gamma,r}(\rho) \Vert \frac{I}{2^n}) 
    \le (1-\gamma)^{r} \calD(\rho \Vert \frac{I}{2^n})
\end{equation}
where $\C_{p,\gamma,r}$ is the channel of \nameref{def:noisy_circuit}, 
and $\calD(\cdot\Vert\cdot)$ is the relative entropy.
Then, we have the following proposition.
The detailed proof can be found in 
\ifnum\onlymaintext=0
    \cref{apd:physical}.
\else
    the Supplementary Materials \cite{seesm}.
\fi
\begin{proposition}[Exponential decay of physical error]\label{thm:exp_phy_bound}
    Given a one-qubit local depolarizing channel
    $\E_{\gamma}^{n}$ with noise rate $\gamma$ and assume $\gamma=o(n^{-1/2})$,
    the $d$-th one-step physical error 
    has the upper bound 
    \begin{equation}
        \eps_{p,\gamma}^{\phy}(d)
        =\trnorm{ \E_{\gamma}^{n}(\rho_d) - \rho_d }
        = \bigO( n^{3/2} \gamma e^{-\frac{1}{2}\gamma d}),
    \end{equation}
    where $\rho_d$ is the state of the $d$-th noisy $p$th-order Trotter step.
    Further, assume the evolved state during the evolution satisfies
    $\sum_{F:\abs{F}=1} \calD\qty(\rho\|  \frac{1}{n}  \rho_{\overline{F}}\otimes \frac{I}{2})= \Theta\qty(\frac{1}{n}) \calD\qty(\rho \| \frac{I}{2^n})$,
    where $F$ takes the subsets of $\{1,2,\cdots,n\}$ and $\rho_{\overline{F}}$ denotes the reduced density matrix of $\rho$ on subsystem $\overline{F}$,
    the upper bound can be improved to 
    $\Theta(n) \gamma e^{-\frac{1}{2}\gamma d}$. 
\end{proposition}
To tighten our upper bound, we analyzed the relaxation from the distance to locally mixed states 
to the distance to the globally mixed state.
This relaxation ratio is theoretically between $[\frac{1}{n}, 1]$
and it is close to $1/n$ empirically as shown in 
\ifnum\onlymaintext=0
    \cref{apd:physical}
\else
    the Supplementary Materials \cite{seesm}
\fi
with common initial states and Hamiltonians. 
In this way, our bound is tight and matches the empirical result of physical error decay in \cref{fig:decay}.

\subsection{Upper bound of algorithmic error}

This section gives an exponentially decaying upper bound 
on the algorithmic (Trotter) error of the noisy circuit by considering state information affected by the depolarizing noise.
Here, we outline the main idea of the proof.
    Given the state $\rho_d$ of the $d$th Trotter step and $M_p$ is the multiplicative $p$th-order Trotter error operator with $\pf_p=U(I+M_p)$,
    the one-step algorithmic (Trotter) error can be written as 
    \begin{align}
        \epsilon_{p,\gamma}^{\alg}(d):=
        \trnorm{U\rho_d U^{\dag}-\pf_p\rho_d \pf_p^{\dag}} = \trnorm{ [\rho_d, M_p] },
    \end{align}
    where $U$ is the ideal one-step evolution, and $\pf_p$ is the one-step $p$th-order Trotter evolution.
In particular, the error operator can be written as a sum of local ones,
that is $M_p=(t/r)^{p+1} \sum_j E_j^{(p)}$,
where $E_j^{(p)}$ is a $p$th-order local Trotter error operator acting on a constant $w$ qubits \cite{childsTheoryTrotterError2021, zhaoEntanglementAcceleratesQuantum2025}. 
Then, we have the algorithmic error 
$\trnorm{ [\rho_d, M_p] }\le (t/r)^{p+1} \sum_j \trnorm{ [\rho_d, E_j^{(p)}] }$.
If the state $\rho_d$ is close to its locally mixed state $I/2^w\otimes\rho_{\overline{F}}$ where $\rho_{\overline{F}}$ is the reduced density matrix of $\rho_d$, 
then we can replace $\rho_d$ with $I/2^w\otimes\rho_{\overline{F}}$. 
Since the error operator $E_j^{(p)}$ only acts on few qubits with support $F$ 
and if the evolved state is the maximally mixed state in this subsystem,
then the commutator $[I/2^w\otimes\rho_{\overline{F}},E_j^{(p)}]=0$ incurs no error. 
Next, we can upper bound the error by separating the state information from the worst-case analysis via the Hölder's inequality as 
\begin{equation}
    \epsilon_{p,\gamma}^{\alg}(d) \le
    \delta t^{p+1} 2\trnorm{ \frac{I}{2^w}\otimes\rho_{\overline{F}}-\rho_d}  \cdot 
    \sum_j \opnorm{E_j^{(p)}},
\end{equation}
where $\sum_j \opnorm{E_j^{(p)}}$ bounds the worst-case Trotter error.
Now, the upper bound of algorithmic error is just to bound
the distance between $I/2^w\otimes\rho_{\overline{F}}$ and $\rho_d$, 
which has been done in the physical error part.
\begin{proposition}[Exponential decay of algorithmic error]\label{thm:exp_alg_bound}
    Given a $p$th-order Trotter circuit channel $\tilde{\mathcal{U}}_p$ with the local depolarizing noise rate $\gamma$,
    the $d$-th one-step algorithmic error $\epsilon_{p,\gamma}^{\alg}$ has the upper bound
    \begin{equation}\label{eq:alg_bound}
        \epsilon_{p,\gamma}^{\alg}(d)
        :=\trnorm{ \mathcal{U}(\rho_d) - \pfc_p(\rho_d)}
        = \bigO\qty(n^{1/2} B_p \frac{t^{p+1}}{r^{p+1}} e^{-\frac{1}{2} \gamma d}),
    \end{equation}
    where $\rho_d$ is the state of the $d$-th noisy $p$th-order Trotter step,
    $\mathcal{U}(\rho):= U\rho U^{ \dagger}$ is the one-step ideal evolution channel,
    and $B_{p}:=\sum_j \opnorm{E_j^{(p)}}$ is a factor related to the Hamiltonian and the order of the Trotter formula.
    With the same assumption in \cref{thm:exp_phy_bound},
    the overhead $n^{1/2}$ in \cref{eq:alg_bound} can be dropped.
\end{proposition}

\begin{figure}[!t]
    \centering
    \includegraphics[width=\autowidth\linewidth]{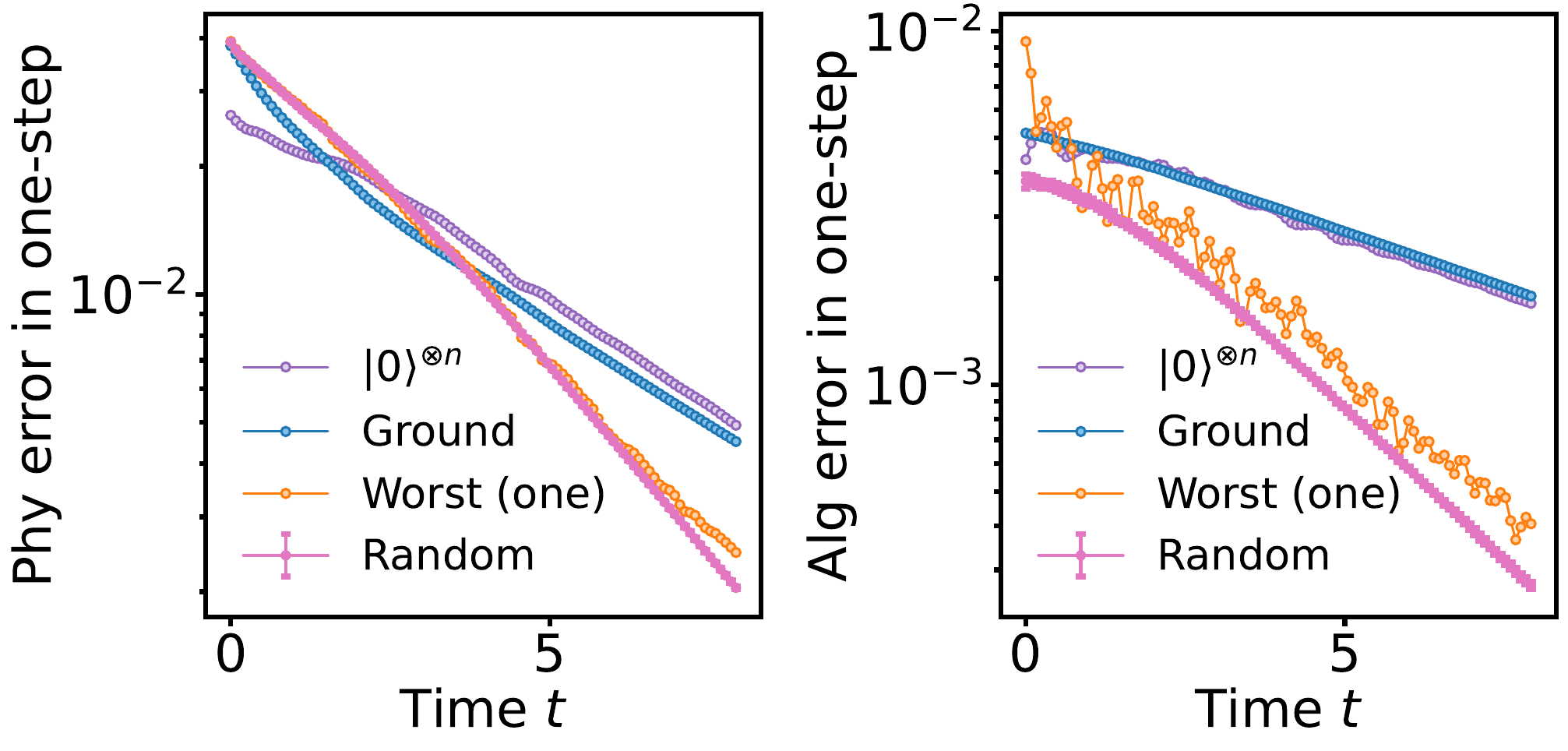}
    \caption{ 
    The impact of initial states on the error decay. 
    We use the standard setup described in \cref{fig:decay} 
    but with the fixed noise rate $\gamma=0.005$.
    The different initial states
    include the product state $\ket{0}^{\otimes n}$, 
    the worst-case state of one-step Trotter (the state that maximizes the one-step algorithmic error), 
    the ground state of the Hamiltonian,
    and an ensemble of 20 Haar random states with error bars indicating standard deviations.
    }
    \label{fig:initial_states}
\end{figure}

Similar to the physical error, the prefactor 
$B_{p} \frac{t^{p+1}}{r^{p+1}}$ 
in \cref{eq:alg_bound} is the upper bound of the worst-case noiseless Trotter error. 
For example, $B_{p}$ is $\Theta(n)$ for the nearest-neighbor interaction Hamiltonians.
Analogous to the physical error decay, the exponential decay factor $e^{-\frac{1}{2} \gamma d}$ comes from the information that the state gets close to $I/2^n$.
We put the complete proof in 
\ifnum\onlymaintext=0
\cref{apd:prf:trace_norm_alg_error_bound}.
\else
    the Supplementary Materials \cite{seesm}.
\fi
This exponential decay upper bound matches the numerical results in \cref{fig:decay} very well.
In addition, we provide numerical results on the impacts of initial states.
\cref{fig:initial_states} confirms that the different initial states all have exponentially decaying errors, but with different decaying rates.
In particular, random initial states \cite{zhaoHamiltonianSimulationRandom2021, chenAverageCaseSpeedupProduct2024} 
have smaller algorithmic errors as well as faster decaying rates, 
which should also occur with entangled states \cite{zhaoEntanglementAcceleratesQuantum2025}.
In addition, we also show the errors in common observables and their expectations \cite{heylQuantumLocalizationBounds2019,granetDilutionErrorDigital2025,yuObservableDrivenSpeedupsQuantum2025} in 
\ifnum\onlymaintext=0
\cref{apd:sec:numeric}.
\else
    the Supplementary Materials \cite{seesm}.
\fi

\subsection{Applications of improved accumulated error}

    Considering both physical and algorithmic errors,
    the \emph{one-step error} of the $d$th Trotter step is defined as
    $\epsilon_{p,\gamma}^{\tot}(d) :=\trnorm{ \mathcal{U}(\rho_d) - \E^n_\gamma\circ \tilde{\mathcal{U}}_p (\rho_d)}$,
    which has the upper bound 
    $\epsilon_{p,\gamma}^{\tot} \le \epsilon_{p,\gamma}^{\phy}+ \epsilon_{p,\gamma}^{\alg}$.
    At last, the total accumulated error of the $r$-step noisy Trotter circuit is defined as 
    the trace distance between the state undergoing the ideal evolution channel $\mathcal{U}$ and 
    the one undergoing the entire noisy Trotter channel $\C_{p,\gamma,r}$, i.e.
    \begin{equation}\label{eq:def:acc_error}
        \epsilon_{p,\gamma}^{\acc}(r) :=
        \trnorm{ \mathcal{U}(\rho_0)  - \C_{p,\gamma,r}( \rho_0 )}
        \le \sum_{d=1}^r \epsilon_{p,\gamma}^{\tot}(d)
    \end{equation}
    where $\rho_0$ is the initial state.
Therefore, with \cref{thm:exp_phy_bound} and \cref{thm:exp_alg_bound}, 
we have the following upper bound of the accumulated error.
\begin{theorem}[Upper bound of noisy Trotter error]\label{thm:accumulated_error}
    Consider a $p$th-order Trotter circuit with $r$ Trotter steps and local depolarizing noise rate $\gamma$ 
    for simulating dynamics of a Hamiltonian $H$ for time $t$.
    The accumulated error of the noisy Trotter circuit has the upper bound
    \begin{equation}\label{eq:acc_error}
        \epsilon_{p,\gamma}^{\acc} (r)
        \le
        \sum_{d=1}^r
        C \gamma e^{-c \gamma  d} + 
        B_p \frac{t^{p+1}}{r^{p+1}} e^{-b \gamma d}.
    \end{equation}
    The upper-case letters $C$ and $B$ are the prefactor coefficients
    and the lower-case letters $c$ and $b$ are the decaying coefficients for physical and algorithmic errors respectively.
\end{theorem}

\begin{figure}[!t]
    \centering
    \includegraphics[width=0.99\linewidth]{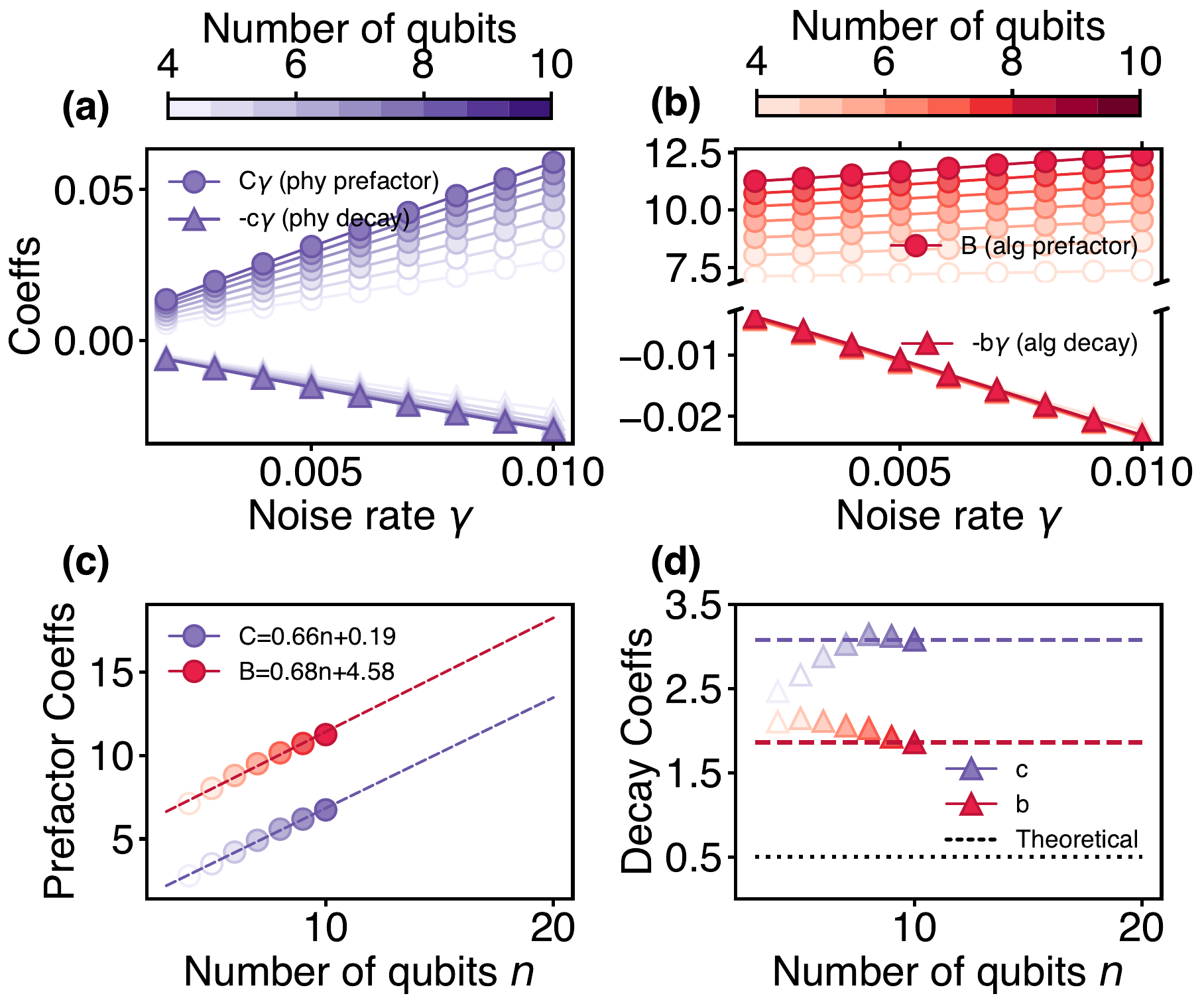}
    \caption{  
    Fitting and extrapolating the one-step 
    physical error $\eps_{2,\gamma}^{\phy}=C\gamma e^{-c \gamma  d} $ and 
    algorithmic error $\eps_{2,\gamma}^{\alg}= B \frac{t^{3}}{r^{3}} e^{-b \gamma d}$
    of the noisy PF2 with $r=100$ Trotter steps.
    We adopt the standard Trotter setup described in \cref{fig:decay} with varying system sizes (number of qubits $n$).
    (a) Both physical prefactor and decay coefficients scale linearly with noise rate $\gamma$. 
    (b) Both algorithmic prefactor and decay coefficients scale linearly with noise rate $\gamma$.
    (c) The fitted prefactor coefficients $C$ and $B$ scale linearly with number of qubits $n$.
    We extrapolate them to large $n$ cases by the dashed lines.
    (d) The fitted decaying coefficients $c$ and $b$ do not strongly depend on number of qubits $n$.
    }
    \label{fig:fitting}
\end{figure}

Without the exponential decay factors, \cref{eq:acc_error} recovers the worst-case analysis used in Refs.~\cite{kneeOptimalTrotterizationUniversal2015, endoMitigatingAlgorithmicErrors2019, avtandilyanOptimalorderTrotterSuzuki2024}.
Recall that, in \cref{fig:decay}, we numerically obtain the empirical one-step errors at every Trotter step with different noise rates.
By fitting the logarithm of the errors by a linear function in terms of Trotter steps,
we can obtain the decay coefficients for a specific typical noise rate $\gamma$, Trotter number $r$, evolution time $t$, and system size $n$.
In \cref{fig:fitting}, we plot the fitted physical coefficients in (a) and algorithmic coefficients in (b) with varying noise rates.
The color of each line indicates the system size from $n=4$ (lightest) to $n=10$ (darkest). 
With the fitted coefficients for different system sizes $n$ from (a) and (b), we plot the prefactor coefficients in (c) and the decay coefficients in (d) with varying $n$.
In \cref{fig:fitting} (c), it is apparent that the prefactors $C$ and $B$ are proportional to the system size $n$.
To predict the performance of large quantum simulators, we extrapolate the lines to $n=20$.
In our theoretical analysis, the decay coefficients $c$ and $b$ are supposed to be constants independent of system size.
Therefore, we set the theoretical value $c=b=1/2$ in our empirical formula to avoid underestimating the errors. 

\subsection{Phase diagram of error in noisy Trotter}\label{sec:phase}
In the NISQ era, it is important to identify the parameter regime for robust Hamiltonian simulation against noise.
For this objective, we depict the phase diagram of the accumulated error of a noisy Trotter circuit by a heatmap as the left panel of \cref{fig:phase}.
It shows the tradeoff between Trotter steps $r$ and noise rate $\gamma$.
With \cref{thm:accumulated_error} and the numerical results, 
the accumulated error of noisy PF2 is approximately
$\eps_{2,\gamma}^{\acc} (r)
\approx \sum_{d=1}^r  
C \gamma e^{-c \gamma d} + 
B_2 \frac{t^3}{r^3} e^{-b \gamma d}$
where $C$, $B_2$ are fitted prefactors, $c$ and $b$ are fitted decaying coefficients from \cref{fig:fitting}.
Two axes are the noise rate $\gamma$ and the Trotter number $r$.
The color indicates the accumulated error, 
i.e., the red color represents a large error while the blue color indicates a small one.
In this sense, the dark blue regime is the range of parameters for robust quantum simulation that could demonstrate practical quantum advantage.
The right panel of \cref{fig:phase} shows the reduction of the error in the phase diagram by our analysis over the worst-case analysis
$\tilde{\eps}_{2,\gamma}^{\acc} (r) \approx  2n \gamma r + B \frac{t^3}{r^2}$ 
where $B$ is obtained from fitting the worst-case empirical results in 
\ifnum\onlymaintext=0
\cref{fig:worst_algorithmic_error}.
\else
    the Supplementary Materials \cite{seesm}.
\fi

The red regime in the phase diagram not only has large accumulated errors, but also can be simulated efficiently by classical algorithms.
Specifically, low-depth circuits can be classically simulated by methods such as those based on the light cone \cite{bravyiClassicalAlgorithmsQuantum2021} or the tensor-network \cite{begusicFastConvergedClassical2024, tindallEfficientTensorNetwork2024}.
And, for the constant-per-gate (strong) noise model, the classical algorithm (i.e., the Pauli-path integral) can efficiently estimate the expectation values of evolved observables with high probability \cite{shaoSimulatingNoisyVariational2024, begusicFastConvergedClassical2024}.

\begin{figure}[!t]
    \centering
    \includegraphics[width=.99\linewidth]{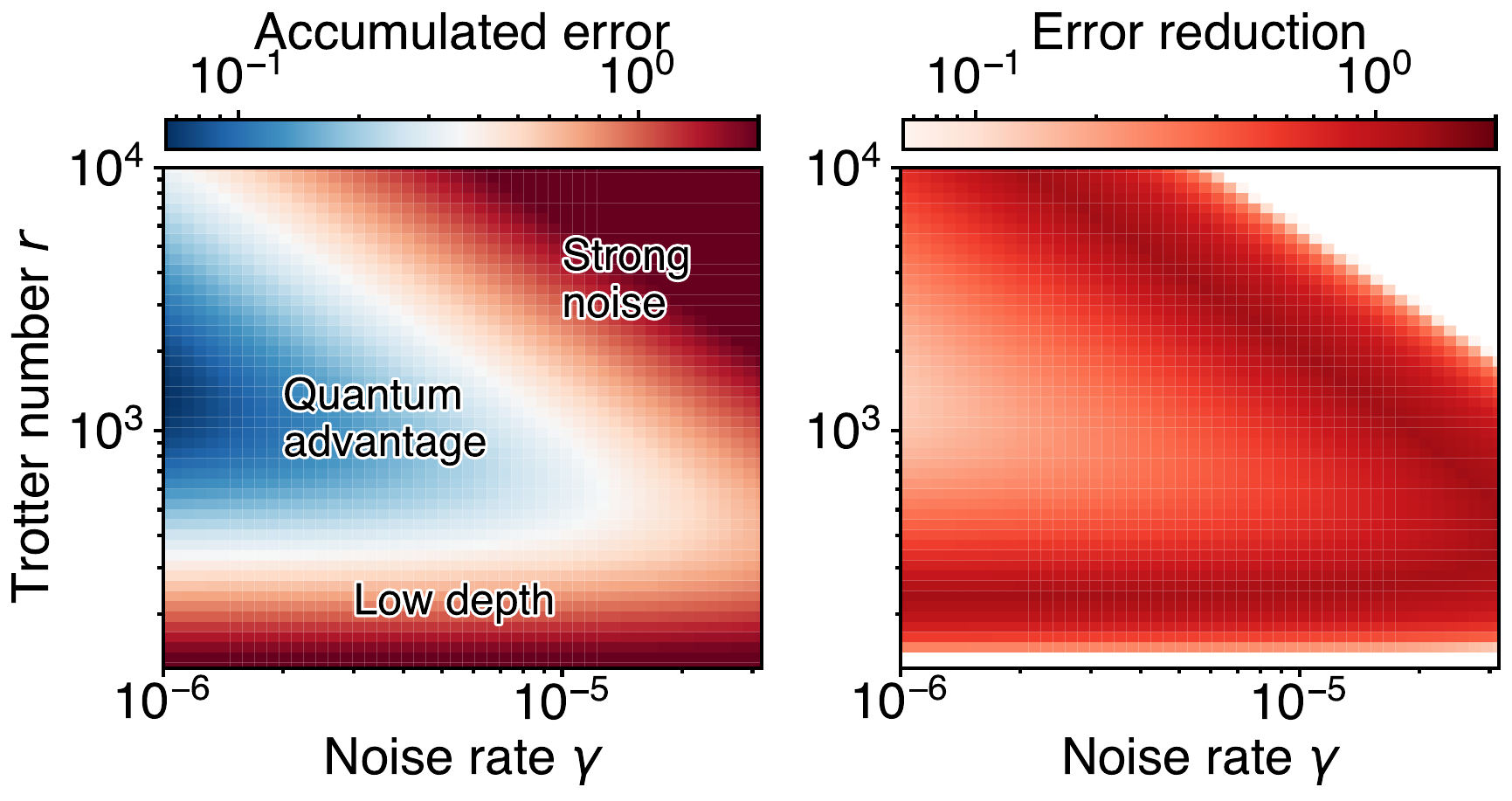}
    \caption{
    Phase diagram of the accumulated error in the noisy Trotter simulation with PF2, the TFI Hamiltonian of $n=50$, and $t=10$. 
    \emph{Left panel}. 
    The $x$ axis is the noise rate per gate $\gamma $, %
    while the $y$ axis is the Trotter number $r$. %
    The color indicates the accumulated error $\eps^{\acc}_{2,\gamma}(r)$ 
    by our analysis with fitted coefficients from \cref{fig:fitting}. 
    \emph{Right panel}. In the same setting, our state-dependent analysis yields at most 33\% error reduction over the worst-case analysis.
    }
    \label{fig:phase}
\end{figure}

\subsection{Optimal Trotter number and noise requirement}\label{sec:optimal_requirement}
Ideally, the algorithmic (Trotter) error can be infinitely suppressed by increasing the Trotter number $r$ or the Trotter order $p$.
However, this is not the case in the practical scenario.
For example, considering the overhead of circuits in high-order Trotter, Ref.~\cite{childsFirstQuantumSimulation2018} studied the optimal Trotter order to minimize the gate counts in the noiseless case %
and Ref.~\cite{avtandilyanOptimalorderTrotterSuzuki2024} studied the optimal Trotter order with noisy circuits. 

On the other hand, due to the accumulation of the physical error with increasing circuit depth,
there exists an optimal Trotter number to minimize the accumulated error as shown in \cref{fig:phase}.
Previous works \cite{kneeOptimalTrotterizationUniversal2015, endoMitigatingAlgorithmicErrors2019, childsNearlyOptimalLattice2019, hakkakuDataEfficientErrorMitigation2025} analyzed the accumulated error of the noisy first-order Trotter (PF1) circuits by simply multiplying the worst-case errors with the Trotter number $r$, that is,
$\tilde{\epsilon}_{1,\gamma}^{\acc}(r) \approx r \qty( C \gamma  + B \frac{t^{2}}{r^2})$.
Then,
the optimal Trotter number of PF1 has the form 
$\tilde{r}_{\opt}\sim t\sqrt{\frac{B}{C\gamma}}$.
However, this upper bound by the triangle inequality is not tight 
because it only considers the worst-case for each step 
\cite{childsNearlyOptimalLattice2019,avtandilyanOptimalorderTrotterSuzuki2024},
ignoring the evolved state information.

Therefore, given the immense reduction in Trotter error by higher-order product formula, it is natural to ask 
how to balance two sources of errors by the number of high-order Trotter steps.
To analytically derive the optimal high-order Trotter number that minimizes the total error,
we assume the accumulated error of the $p$th-order noisy Trotter has the following form
(supported by \cref{thm:accumulated_error} and the numerical results),
\begin{equation}\label{eq:empirical_acc_error}
    \eps_{p,\gamma}^\acc (r) = 
    \sum_{d=1}^r C\gamma\Upsilon e^{-c\gamma\Upsilon d} + 
    B_p \frac{t^{p+1}}{r^{p+1}}e^{-b\gamma\Upsilon d}, 
\end{equation}
where $\Upsilon = \bigO( 2^p)$ is the number of layers implemented in one Trotter step in high-order product formulas. 

\begin{proposition}[Optimal Trotter number for noisy Trotter]\label{thm:optimal_trotter_number}
    Consider a $p$th-order product formula circuit with local depolarizing noise rate $\gamma$ and 
    the accumulated error formula \cref{eq:empirical_acc_error}. 
    Assume $c<C$ and $b\approx c$.
    The optimal number of Trotter steps is
    \begin{equation}
        r_\opt(\gamma) = 
        \qty(\frac{pB_p}{C\gamma\Upsilon})^{\frac{1}{p+1}}t. 
    \end{equation}
    To give an explicit asymptotic scaling, 
    assume $C=\Theta(n)$ and $B_p=\Theta(n)$, 
    then $r_\opt(\gamma) = \Theta\qty( \gamma^{-\frac{1}{p+1}} t)$.
\end{proposition}

Since there exists a minimal accumulated error $\varepsilon$ achieved by the optimal Trotter number $r_{\opt}$ for specific noise rate,
we derive the following \emph{noise rate requirement} $\gammastar$ necessary for robust quantum simulation with noisy high-order Trotter circuits.
\begin{corollary}[Noise rate requirement for robust noisy Trotter simulation]\label{thm:threshold}
    Consider the $p$th-order Trotter circuit with local depolarizing noise rate $\gamma$
    and 
    assume the accumulated error formula \cref{eq:empirical_acc_error}. 
    If the noise rate is larger than 
    \begin{equation}
        \gammastar = 
        \frac{1}{C\qty(B_p)^{\frac{1}{p}}} 
        \qty(\frac{\varepsilon}{t})^{1+\frac{1}{p}} 
        \frac{p}{\Upsilon (p+1)^{1+\frac{1}{p}}},
    \end{equation}
    for system size $n$ and evolution time $t$,
    it is impossible to guarantee the accumulated error within simulation precision $\varepsilon$.
\end{corollary}
    To give a concrete asymptotic scaling, we could assume 
    $C=\Theta(n)$ and $B_p=\Theta(n)$ 
    supported by the numerical results,
    then $\gammastar(\varepsilon) = \bigO\qty(\qty(\frac{\varepsilon}{nt})^{1+\frac{1}{p}})$.
We postpone the detailed proof to 
\ifnum\onlymaintext=0
\cref{apd:sec:threshold}, 
\else
    the Supplementary Materials \cite{seesm},
\fi
where we also give examples of the requirement $\gammastar$ and optimal Trotter number $r_{\opt}$ for the first and the second-order noisy Trotter.
This asymptotic scaling (inverse polynomial in $n$) requires quantum devices at least in the weak noise ($\gamma\sim 1/n$) regime \cite{feffermanEffectNonunitalNoise2024},
which is a good approximation of the noise present in some relatively small, fixed-sized experimental systems as \cite{aruteQuantumSupremacyUsing2019,morvanPhaseTransitionsRandom2024}.
In contrast, the circuits with constant noise rate (independent of the system size $n$) are regarded in a strong noise regime. 
Since Trotter algorithm typically requires at least $\poly(n)$ circuit depth and
the state after $\Omega(\log(n))$ depth with constant depolarizing noise rate becomes meaningless (i.e., very close to the maximally mixed state) \cite{ben-orQuantumRefrigerator2013},
the constant noise circuits are not robust for demonstrating advantage by Trotter quantum simulation.

\subsection{Resource saving for fault-tolerance}

\begin{figure}[!t]
    \centering
    \includegraphics[width=0.97\linewidth]{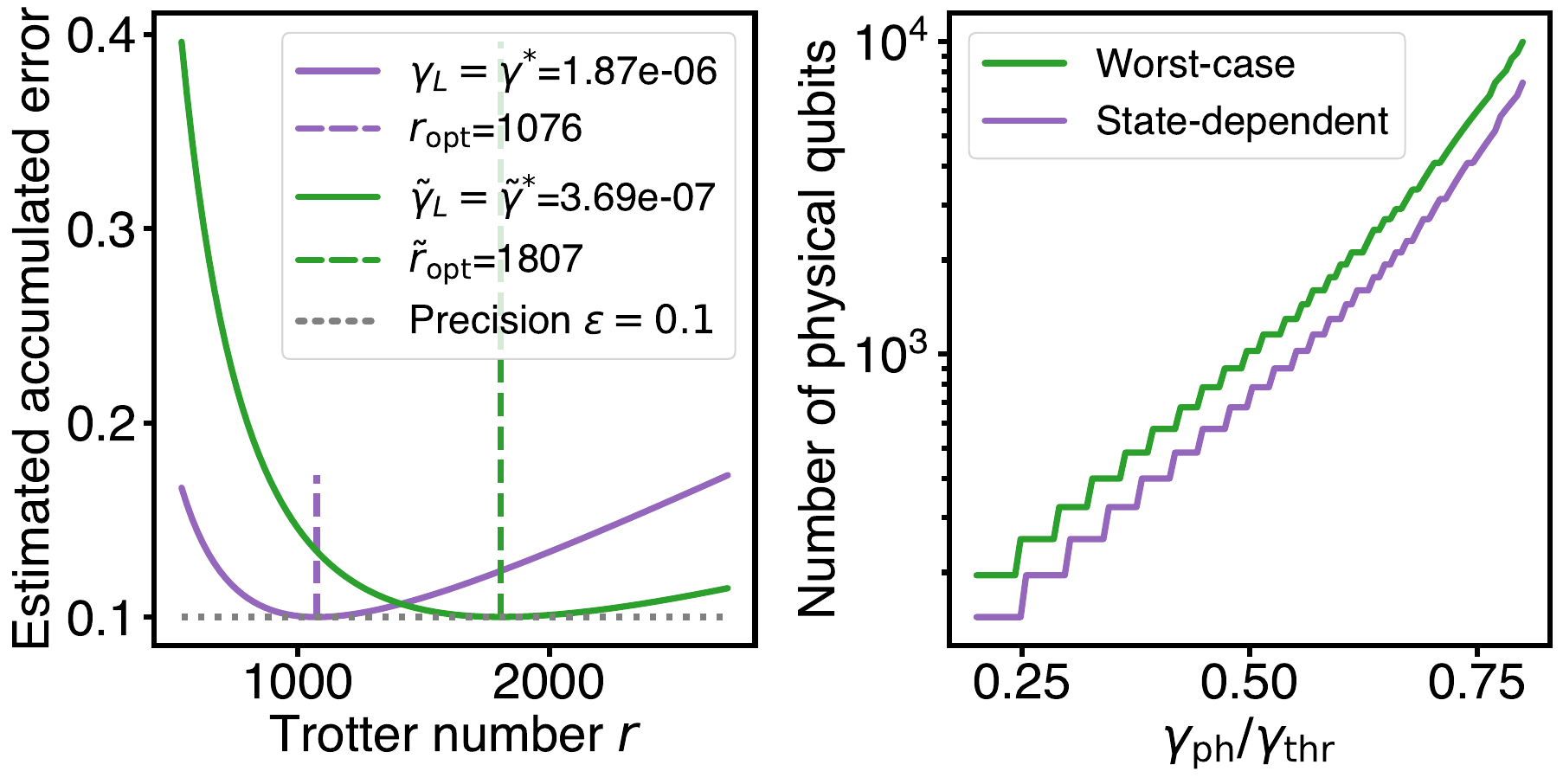}
    \caption{
    The resource-saving in fault-tolerance by our improved error analysis.
    \emph{Left panel}. Compare the optimal Trotter number and corresponding noise rates in the worst-case analysis ($\tilde{r}_{\opt},\tilde{\gamma}_\mathrm{L}=\tgammastar$) 
    and our state-dependent analysis ($r_{\opt},\gamma_\mathrm{L}=\gammastar$).
    We adopt the same Trotter setup as \cref{fig:decay} with 
    $n=50$ and $t=10$.
    We binary search the noise rate $\gammastar$ for robust simulation such that the minimum accumulated error with the optimal Trotter number $r_{\opt}$ reaches the simulation precision $\varepsilon=0.1$.
    Then, this noise rate $\gammastar$ is set as the logical error rate $\gamma_{\mathrm{L}}$ of the fault-tolerant Trotter circuits.
    \emph{Right panel}. 
    The number of physical qubits by the worst-case and our state-dependent error analyses.
    The $x$-axis $\gamma_{\ph}/\gamma_{\thr}$ represents how much the physical error rate is below the code threshold,
    while the $y$-axis indicates the number of physical qubits
    \cref{eq:physical_qubits}
    where we take $\gamma_0=0.02985$ 
    calculated from the parameters of \cite{acharyaQuantumErrorCorrection2024}.
    }
    \label{fig:qecc_reduction}
\end{figure}

Our improved error analysis enables the resource-saving for fault-tolerant quantum simulation.
Assume we realize fault tolerance with quantum error-correcting codes (QECC) and implement the small-angle rotations required for Trotterization via efficient magic state approaches
\cite{bravyiUniversalQuantumComputation2005,bravyiMagicstateDistillationLow2012,willsConstantoverheadMagicState2025}. %
The noise rate $\gamma_{\mathrm{L}}$ per logical gate can be suppressed by QECC with a code distance $d_c$  as
$\gamma_{\mathrm{L}} = \gamma_0 \qty( \frac{\gamma_{\ph}}{\gamma_{\thr}} )^{d_c/2}$,
where $\gamma_{\ph}$ is the noise rate per gate on the physical qubits,
$\gamma_0$ is a constant parameter,
and $\gamma_{\thr}$ is the threshold of the code. 
The code distance $d_c$ is typically a power of the number of physical qubits $N_c$, expressed as $d_c \sim N_c^l$, where $0\leq l \leq 1$. 
Specifically, for the widely-studied surface code \cite{fowlerSurfaceCodesPractical2012}, $d_c = \sqrt{N_c}$. 
Therefore, to satisfy the logical noise rate $\gamma_{\mathrm{L}}$, we need the number of physical qubits to be
\begin{equation}\label{eq:physical_qubits}
    N_c=(d_c)^2=\qty(2\log_{(\gamma_\ph/\gamma_\thr)} (\gamma_{\mathrm{L}}/\gamma_0))^2.
\end{equation}

Moreover, the code distance also affects the measurement cost of fault-tolerant syndrome extraction. 
For surface code, each stabilizer generator requires $\bigO(d_c)$ measurements per QECC cycle, 
while general codes may require up to $\bigO(d_c^2)$ measurements per stabilizer generator to achieve Shor-type fault-tolerance. 
Thus, reducing the code distance $d_c$ directly reduces both physical qubit number and measurement overhead in fault-tolerance.

In our state-dependent error analysis (exponentially decaying errors), 
the effective accumulation of the physical noise diminishes as circuit depth grows. 
Although operating below the threshold slows down the system's entropy growth,
the system remains entropy-increasing because the finite size of the QECC prevents the logical noise rate from reaching zero. 
Since the logical system is not completely error-free,
our exponentially decaying error formula still works. 
Consequently, to achieve the same simulation precision as in the worst-case error analysis, 
the logical noise rate per layer can be relatively relaxed by our state-dependent error analysis, leading to potential resource savings.
 
\begin{table}[!t]
    \centering
    \begin{tabular}{c|c|c|c}
         \#qubits   & $n=10$  & $n=50$ & $n=200$ \\
         \hline\hline
         ($\tilde{r}_{\opt}$, $\tgammastar$) &  (823, 4.05e-6) & (1807, 3.69e-7) & (3609, 4.60e-8) \\
         ($r_{\opt}$, $\gammastar$) & (584, 1.69e-5) & (1076, 1.87e-6) & (2053, 2.46e-7)  \\
         $1-\frac{r_{\opt} \cdot N_c}{\tilde{r}_{\opt}\cdot \tilde{N}_c}$ & $50\%$ & $60\%$ & $59\%$
         
    \end{tabular}
    \caption{
    Comparison of the optimal Trotter number and (logical) noise rate requirement estimated by the worst-case analysis ($\tilde{r}_{\opt}$, $\tgammastar$) and 
    our state-dependent analysis
    ($r_{\opt}$, $\gammastar$), for different system sizes 
    (number of qubits $n=10,50,200$). 
    For the TFI Hamiltonian, fixed $t=10$ and entire simulation precision $\varepsilon=0.1$, 
    the optimal PF2 Trotter number $\tilde{r}_{\opt}$ ($r_\opt$) and 
    the corresponding (logical) noise rates $\tgammastar$ ($\gammastar$) are determined by the binary search of the empirical worst-case (state-dependent) formulas.
    Our analysis follows 
    \cref{eq:acc_error}
    with fitted coefficients
    and the worst-case error is 
    $\tilde{\eps}^{\acc}_{2,\gamma}(r)=2nr\gamma + B_2 t^3/r^2$ 
    with $B_2(n)$ obtained from fitting empirical results 
    \ifnum\onlymaintext=0
    cf. \cref{fig:worst_algorithmic_error}.
    \else
        in the Supplementary Materials \cite{seesm}.
    \fi
    The number of physical qubits $N_c$ \cref{eq:physical_qubits}
    is determined given logical noise rates $\tgammastar(\gammastar)$ for the state-of-the-art physical noise rate $\gamma_{\ph}/\gamma_{\thr}=0.5$.
    The last row shows the total resource-saving in $r\cdot N_c$ (i.e., number of Trotter steps times number of physical qubits) for fault-tolerant Trotter circuits.
}
    \label{tab:comparison}
\end{table}

The left panel of \cref{fig:qecc_reduction} compares the optimal Trotter numbers and corresponding noise rate requirements 
in the worst-case analysis and our state-dependent analysis.
The right panel of \cref{fig:qecc_reduction} illustrates the reduction in code distance given the logical noise rate
(equivalently, reducing the number of physical qubits) 
enabled by our error analysis under different relaxation ratios between physical noise rates and logical noise rates. 
\cref{tab:comparison} summarizes the improvements in resources for different system sizes by our state-dependent analysis.
Notably, when the actual physical noise rate is only slightly below the QEC threshold $\gamma_{\thr}$, 
our exponentially decaying error formula could significantly reduce physical qubits for fault-tolerant quantum simulations.
For instance, in the case of the surface code, the threshold under depolarizing noise and two-qubit correlated noise is approximately $1.8\%$ with optimal decoder \cite{heimOptimalCircuitLevelDecoding2016}.  
At the current state-of-the-art QEC, the physical noise rates can be maintained approximately 50\% of the threshold, i.e. $\gamma_{\ph}/\gamma_{\thr}\approx 0.5$ \cite{acharyaQuantumErrorCorrection2024}. 
In this case, our analysis can reduce the resource $r\cdot N_c$ (number of Trotter steps times number of physical qubits) by $60\%$ when system size $n=50$.
Therefore, in this practical regime, our error analysis can accelerate the path toward practical quantum advantage by quantum simulation.

Though the one-qubit depolarizing noise is the most widely studied noise channel in theory,
it does not necessarily model all realistic noises on current quantum devices.
Besides the depolarizing noise, it would be valuable to analyze other types of noise \cite{ben-orQuantumRefrigerator2013,clintonHamiltonianSimulationAlgorithms2021,feffermanEffectNonunitalNoise2024}.  %

\subsection{Common incoherent noise channels}
While we use the depolarizing noise for analyses due to its simplicity,
the general Pauli noise channel $\E_{\gamma_x,\gamma_y,\gamma_z}(\rho)$ of the form 
\begin{equation}
    (1-\gamma_x-\gamma_y-\gamma_z)\rho + \qty(\gamma_x X\rho X + \gamma_y Y \rho Y + \gamma_z Z\rho Z) 
    \label{eq:pauli_noise}
\end{equation}
is more realistic.
For example, depolarizing noise is with $\gamma_x=\gamma_y=\gamma_z$,
while the dephasing noise has $\gamma_x=\gamma_y=0$.
On the other hand, for the typical non-unital noise \cite{ben-orQuantumRefrigerator2013,feffermanEffectNonunitalNoise2024,shtankoComplexityLocalQuantum2025},
the amplitude damping channel has the Kraus form,
$\mathcal{E}_\gamma^{\mathrm{amp}}(\rho)=A_0\rho A_0L^\dagger + A_1 \rho A_1^\dagger$ 
where 
$A_0=\begin{pmatrix}1&0\\ 0&\sqrt{1-\gamma}\end{pmatrix}$
and 
$A_1=\begin{pmatrix}0&\sqrt{\gamma} \\ 0&0\end{pmatrix}$.
\cref{fig:pauli_noises} shows that
the algorithmic errors always share the similar decaying behaviors with the physical errors,
while three typical noise channels have three distinct decaying behaviors.
Specifically, the one-step error with the depolarizing noise channel decays fastest, 
while the amplitude damping one decays the most slowly and oscillates.
These different error decay rates qualitatively match the maximal computation time of circuits with different noise channels as indicated in~\cite{ben-orQuantumRefrigerator2013}.
See more detailed numerical results and discussion in 
\ifnum\onlymaintext=0
\cref{apd:sec:numeric}.
\else
    the Supplementary Materials \cite{seesm}.
\fi

\begin{figure}[!t]
    \centering
    \includegraphics[width=.99\linewidth]{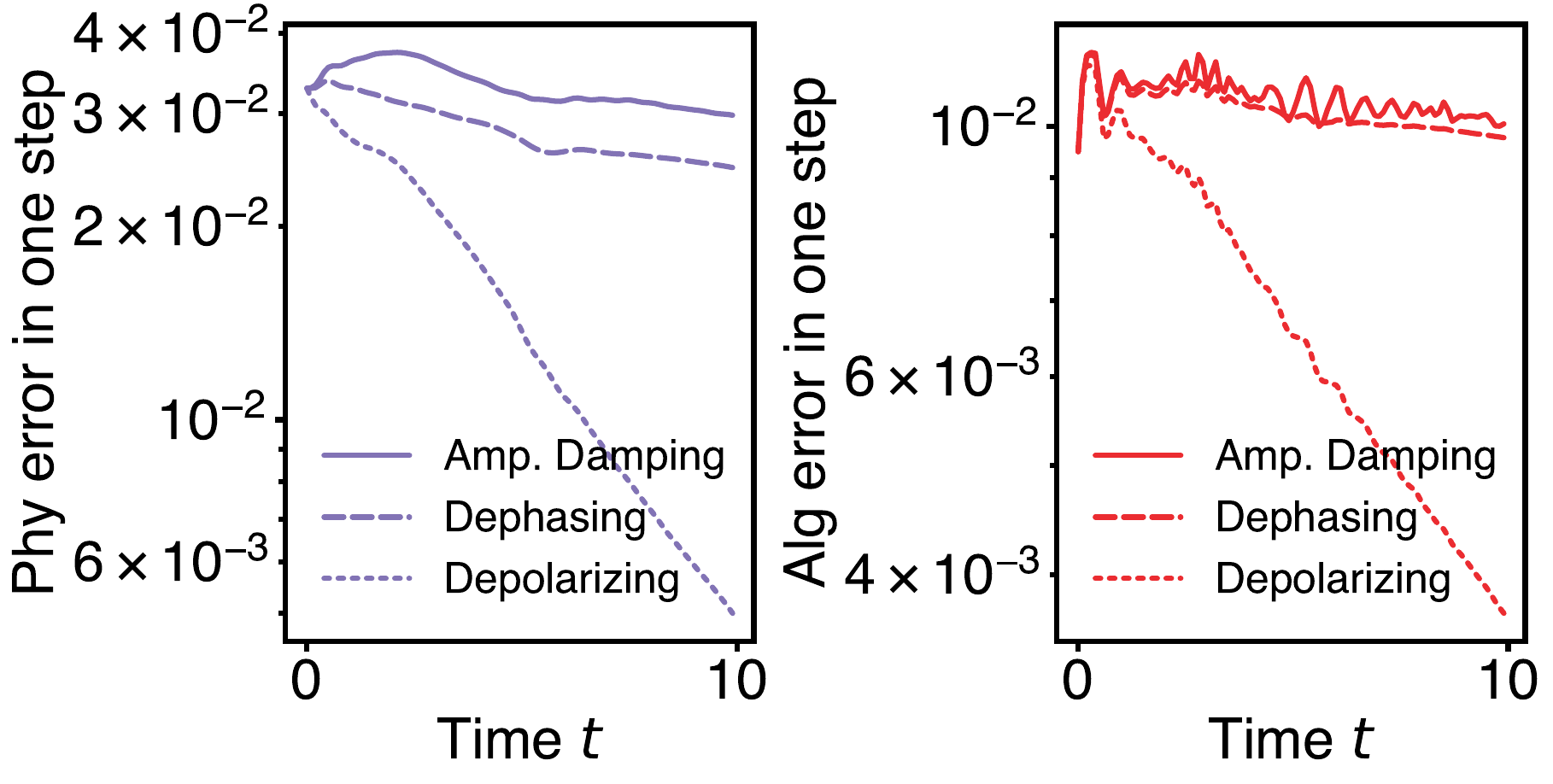}
    \caption{
    Different decay behaviors of three common noise channels.
    The physical and algorithmic error of the depolarizing, dephasing, and amplitude damping noise are plotted in two panels respectively,
    with the same noise rate $\gamma=0.005$ and $n=10$.
    The setting of the noisy Trotter circuit is the same as \cref{fig:decay}.
    }
    \label{fig:pauli_noises}
\end{figure}

\subsection{Error per-time model}

In Trotter circuits, the gates are usually Pauli rotations $e^{-\ii P \dt}$ for certain local Pauli operator $P$ with rotation angle $\dt=t/r$.
In our previous analyses, the noise rate $\gamma$ is constant regardless of the rotation angle.
However, an alternative model called \emph{error-per-time} model in which the noise rate is proportional to the angle, 
i.e., $\gamma \sim \dt$ ,
is also considered in \cite{clintonHamiltonianSimulationAlgorithms2021}.
Both of these error models are idealizations but reasonable for different scenarios.
For example, if each rotation gate is implemented via a pulse,
the error probability will naturally be proportional to the gate duration (rotation angle). 
This model is supported by experiments \cite{stengerSimulatingDynamicsBraiding2021,earnestPulseefficientCircuitTranspilation2021},
where shorter pulse implementations achieve higher fidelities than standard longer ones.
Consequently, if we adopt this model, 
the total physical error of the noisy Trotter circuit would be bounded by the total evolution $t$,
independent of the number of Trotter steps $r$. 
Nevertheless, for quantum devices operated digitally where Pauli rotations are compiled into native elementary gates, the constant-noise model is more reasonable.
Then, the total accumulated error would still be one-step error times $r$ as in previous sections.

\section{Discussion}

In this work, we quantitatively investigated the algorithmic and physical errors of the Trotter circuit with depolarizing noise.
Considering the evolved state information, our exponentially decaying error upper bound provides a more accurate estimation of the optimal number of Trotter steps and an explicit requirement on noise rate for robust noisy Trotter simulation.
In consequence, our analysis offers substantial reduction in circuit depths and number of physical qubits for fault-tolerant Trotter circuits.
Concretely, our findings delineate parameter regimes where Trotter simulation may demonstrate practical and robust quantum advantage despite noise.
To substantiate the theoretical proofs, various numerical results investigate the impacts, such as initial states, observables, and noise channels.
Therefore, by enhancing the understanding of the interplay between noise and algorithmic errors, 
our work is an important step toward practical quantum advantage via quantum simulation.

Besides the Trotter error of states as in this work and Ref.~\cite{childsTheoryTrotterError2021},
the Trotter errors in expectation values of local observables \cite{yuObservableDrivenSpeedupsQuantum2025} could be more optimistic. 
For instance,
Refs.~\cite{heylQuantumLocalizationBounds2019,siebererDigitalQuantumSimulation2019} indicated
that quantum localization keeps Trotter errors in local observables bounded 
if the Trotter step size stays below a critical threshold,
unlike state errors which grow with time and system size.
Analogous to our result, Ref.~\cite{granetDilutionErrorDigital2025} showed 
that the errors in local observables are much lower than expected.
On the other hand, the observable information not only makes Trotter error smaller,
but at the same time renders classical simulation more efficient \cite{begusicRealTimeOperatorEvolution2025}.
Therefore, it is crucial to understand the robustness of Trotter error in local observables with noisy quantum circuits.

Meanwhile, there are still many potential directions to pursue with the noisy Trotter simulation.
First, it would be desirable to combine our analysis with error mitigation techniques such as Refs.~\cite{endoMitigatingAlgorithmicErrors2019, watsonExponentiallyReducedCircuit2025}.
In the other direction, one major open problem is whether we could propose low-overhead error correction schemes \cite{zhouAlgorithmicFaultTolerance2024} 
or partially fault-tolerant schemes \cite{akahoshiPartiallyFaultTolerantQuantum2024,toshioPracticalQuantumAdvantage2025} tailored for Trotter simulation.
    Additionally, for the robustness of post-Trotter quantum simulation algorithms such as Linear Combination of Unitaries \cite{berrySimulatingHamiltonianDynamics2015} and Quantum Signal Processing \cite{lowOptimalHamiltonianSimulation2017}, other error sources should be considered seriously, e.g., the state-preparation-and-measurement (SPAM) errors, and the coherent errors.

\section{Methods}
In this section, we sketch the proof of Propositions 1 and 2 in the main text, 
i.e., the exponential decay of physical and algorithmic errors.

\subsection{Sketch proof of Proposition 1 (Physical error decay)}

The $d$-th step state $\rho_d$ applying the one-qubit depolarizing noise channel $\mathcal{E}_{\depo, \gamma}^{n}$ can be written as
\begin{equation}\label{method:eq:channel_expand}
    \mathcal{E}_{\depo, \gamma}^{n}(\rho_d) = \sum_{F} \gamma^{|F|}(1-\gamma)^{n-|F|} \rho_{\overline{F}} \otimes \frac{I}{2^{|F|}},
\end{equation} 
where $F$ takes all subsets of qubits representing the depolarized subsystems, 
$\overline{F}$ is the complement of $F$ and $\rho_{\overline{F}}$ is the reduced density matrix of $\rho$ on subsystem $\overline{F}$.
The one-step physical error is defined as the trace distance between the depolarized state and the original state.
Since the noise rate $\gamma$ is small, we approximate the physical error to the first-order, 
\begin{align*}
    \trnorm{ \mathcal{E}_{\depo, \gamma}^{n}(\rho_d) - \rho_d } 
     =&  \norm{ -n\gamma \rho_d + \gamma \sum_{F:\abs{F}=1} \rho_{\overline{F}}\otimes \frac{I}{2}}_1 + o(n\gamma) \\
     \leq & n\gamma \sqrt{ 2\calD \qty( \rho_d \Vert  \frac{1}{n}\sum_{F:|F|=1} \rho_{\overline{F}} \otimes \frac{I}{2})  } + o(n\gamma) .
\end{align*}
The second line above is by the Pinsker's inequality 
that upper bounds the trace distance by the relative entropy.
Then, we bound the distance to the locally mixed state by the distance to the globally maximally mixed state $\frac{I}{2^n}$
\begin{align*}
    \calD \qty( \rho_d \Vert  \frac{1}{n}\sum_{F:|F|=1} \rho_{\overline{F}} \otimes \frac{I}{2}) 
    \leq  \calD \qty(\rho_d \Vert \frac{I}{2^n} ) 
    \le  \calD \qty(\rho_0 \Vert \frac{I}{2^n} ) e^{-\gamma d} .
\end{align*}
The last inequality above introducing exponential decay factor $e^{-\gamma d}$ is by the entropy contraction  \cite{francaLimitationsOptimizationAlgorithms2021} 
\begin{equation}\label{method:eq:contraction}
    \calD\qty((\pfc\circ \E_{\gamma}^n)^r(\rho_0) \Vert \sigma) 
    \le (1-\gamma)^{r} \calD(\rho_0 \Vert \sigma).
\end{equation}
where $\rho_0$ is any initial state and $\pfc$ be the channel of one step Trotter unitary.
Since the maximum relative entropy to the maximally mixed state $\calD \qty(\rho_0 \Vert \frac{I}{2^n} )\le n$, we have \cref{thm:exp_phy_bound}.
If the system holds to satisfy 
$\sum_{F:\abs{F}=1} \calD\qty(\rho\|\frac{1}{n} \rho_{\overline{F}}\otimes \frac{I}{2}) = \Theta\qty(\frac{1}{n}) \calD\qty(\rho \| \frac{I}{2^n})$, 
the upper bound of the one-step physical error can be improved to $\Theta(n) \gamma e^{-\frac{1}{2}\gamma d}$.

\subsection{Sketch proof of Proposition 2 (Trotter error decay)}

First, we rewrite the $p$th-order Trotter unitary $\pf_p = U(I+M_p)$ in terms of its multiplicative error operator $M_p$.
Then, the one-step algorithmic error can be written as the trace norm of the commutator between $\rho$ and the Trotter error operator $M_p$
\begin{align*}
    \trnorm{U\rho U^{\dag}-\pf_p\rho \pf^{\dag}_p}
    &= \trnorm{U\rho U^{\dag}-U(I+M_p)\rho (I+M_p^{\dagger}))U^{\dag} } \\
    &= \trnorm{\rho M_p -M_p\rho }
     \equiv \trnorm{[\rho, M_p]},
\end{align*}
 by the unitary invariance of trace norm.
Further, the multiplicative Trotter error operator $M_p$ can be written as the sum of the leading-order local error terms as $M_p:=\delta t^{p+1}\sum_j E_j^{(p)}$ 
where $E_j^{(p)}$ acts on $w$ qubits and $\delta t=t/r$.
For simplicity, we omit the subscript $d$ as well as $p$.
\begin{equation}
    \trnorm{U\rho U^{\dag}-\pf\rho \pf^{\dag}}
    = \trnorm{[\rho, M]}
    = \delta t^{p+1} \trnorm{\qty[\rho, \sum_j E_j]} 
\end{equation}
Then, by inserting cross terms of $\rho':=\rho_{\overline{F}}\otimes I/2^w$, 
the sum of commutator norms $\sum_j \trnorm{ [\rho, E_j] }$ becomes
\begin{align*}
    &\sum_j \trnorm{\rho E_j - \rho' E_j + \rho' E_j - E_j \rho' + E_j\rho' -E_j \rho} \\
    \le & \sum_j \trnorm{ \rho E_j - \rho' E_j } +\trnorm{\rho' E_j -E_j \rho'}  + \trnorm{ E_j\rho' - E_j\rho } \\
    \le & \sum_j 2\trnorm{ \rho  - \rho' } \opnorm{E_j} + \trnorm{[\rho',E_j]}   
    \tag{Hölder's ineq}  \\
    =&  2\trnorm{I/2^w\otimes\rho_{\overline{F}}-\rho}  \sum_j \opnorm{E_j} 
    \tag{locality of $E_j$}  \\ 
    \le& \Theta(n^{1/2})  e^{-\frac{1}{2}\gamma d} \sum_j \opnorm{E_j} 
    \tag{Proposition 1} \\
    \leq& \Theta(n^{1/2}) 
    B_{p} e^{-\frac{1}{2}\gamma d}
    \tag{worst-case Trotter error}
\end{align*}  
The last third line uses $\trnorm{[\rho', E_j]}=0$ because the error operator $E_j$ acts on identity.
The exponential decay factor in the last second line is the same as the one of the physical error proved in Proposition 1.
In the last line, the factor $B_p=\sum_j \opnorm{E_j^{(p)}}$ is from the upper bound of the worst-case Trotter error which is the $p$th-order nested commutator norm \cite{childsTheoryTrotterError2021}. 
So far, we have the exponential decay of the one-step Trotter error \cref{thm:exp_alg_bound}.

\textbf{Acknowledgements}. 
    J.X. and Q.Z. acknowledge funding from Innovation Program for Quantum Science and Technology via Project 2024ZD0301900, National Natural Science Foundation of China (NSFC) via Project No. 12347104 and No. 12305030, Guangdong Basic and Applied Basic Research Foundation via Project 2023A1515012185, Hong Kong Research Grant Council (RGC) via No. 27300823, N\_HKU718/23, and R6010-23, Guangdong Provincial Quantum Science Strategic Initiative No. GDZX2303007, HKU Seed Fund for Basic Research for New Staff via Project 2201100596.

\noindent\textbf{Author contributions}: 
Q.Z. and J.X. proposed the research. 
Q.Z., C.Z. and J.X. completed the theoretical proofs. 
J.X. and C.Z. performed numerical simulations of noisy circuits. 
J.F. and J.X. performed the resource estimation for fault tolerance. 
All authors discussed the results and wrote the manuscript.

\noindent\textbf{Competing interests}: The authors declare that they have no competing interests.

\noindent\textbf{Code availability}:
The code used in this study is available on \href{https://github.com/dzzxzl/Noisy-Trotter-Simulation}{GitHub}.

\bibliographystyle{style/truncate_ref}
\bibliography{bib/ref_aps, bib/seesm}

\begin{thebibliography}{88}%
\makeatletter
\providecommand \@ifxundefined [1]{%
 \@ifx{#1\undefined}
}%
\providecommand \@ifnum [1]{%
 \ifnum #1\expandafter \@firstoftwo
 \else \expandafter \@secondoftwo
 \fi
}%
\providecommand \@ifx [1]{%
 \ifx #1\expandafter \@firstoftwo
 \else \expandafter \@secondoftwo
 \fi
}%
\providecommand \natexlab [1]{#1}%
\providecommand \enquote  [1]{``#1''}%
\providecommand \bibnamefont  [1]{#1}%
\providecommand \bibfnamefont [1]{#1}%
\providecommand \citenamefont [1]{#1}%
\providecommand \href@noop [0]{\@secondoftwo}%
\providecommand \href [0]{\begingroup \@sanitize@url \@href}%
\providecommand \@href[1]{\@@startlink{#1}\@@href}%
\providecommand \@@href[1]{\endgroup#1\@@endlink}%
\providecommand \@sanitize@url [0]{\catcode `\\12\catcode `\$12\catcode `\&12\catcode `\#12\catcode `\^12\catcode `\_12\catcode `\%12\relax}%
\providecommand \@@startlink[1]{}%
\providecommand \@@endlink[0]{}%
\providecommand \url  [0]{\begingroup\@sanitize@url \@url }%
\providecommand \@url [1]{\endgroup\@href {#1}{\urlprefix }}%
\providecommand \urlprefix  [0]{URL }%
\providecommand \Eprint [0]{\href }%
\providecommand \doibase [0]{https://doi.org/}%
\providecommand \selectlanguage [0]{\@gobble}%
\providecommand \bibinfo  [0]{\@secondoftwo}%
\providecommand \bibfield  [0]{\@secondoftwo}%
\providecommand \translation [1]{[#1]}%
\providecommand \BibitemOpen [0]{}%
\providecommand \bibitemStop [0]{}%
\providecommand \bibitemNoStop [0]{.\EOS\space}%
\providecommand \EOS [0]{\spacefactor3000\relax}%
\providecommand \BibitemShut  [1]{\csname bibitem#1\endcsname}%
\let\auto@bib@innerbib\@empty
\bibitem [{\citenamefont {Feynman}(1982)}]{feynmanSimulatingPhysicsComputers1982}%
  \BibitemOpen
  \bibfield  {author} {\bibinfo {author} {\bibfnamefont {R.~P.}\ \bibnamefont {Feynman}},\ }\href {https://doi.org/10.1007/BF02650179} {\bibfield  {journal} {\bibinfo  {journal} {International Journal of Theoretical Physics}\ }\textbf {\bibinfo {volume} {21}},\ \bibinfo {pages} {467} (\bibinfo {year} {1982})}\BibitemShut {NoStop}%
\bibitem [{\citenamefont {Feynman}(1985)}]{feynmanQuantumMechanicalComputers1985}%
  \BibitemOpen
  \bibfield  {author} {\bibinfo {author} {\bibfnamefont {R.~P.}\ \bibnamefont {Feynman}},\ }\href {https://doi.org/10.1364/ON.11.2.000011} {\bibfield  {journal} {\bibinfo  {journal} {Optics News}\ }\textbf {\bibinfo {volume} {11}},\ \bibinfo {pages} {11} (\bibinfo {year} {1985})}\BibitemShut {NoStop}%
\bibitem [{\citenamefont {Georgescu}\ \emph {et~al.}(2014)\citenamefont {Georgescu}, \citenamefont {Ashhab},\ and\ \citenamefont {Nori}}]{georgescuQuantumSimulation2014}%
  \BibitemOpen
  \bibfield  {author} {\bibinfo {author} {\bibfnamefont {I.~M.}\ \bibnamefont {Georgescu}}, \bibinfo {author} {\bibfnamefont {S.}~\bibnamefont {Ashhab}},\ and\ \bibinfo {author} {\bibfnamefont {F.}~\bibnamefont {Nori}},\ }\href {https://doi.org/10.1103/RevModPhys.86.153} {\bibfield  {journal} {\bibinfo  {journal} {Reviews of Modern Physics}\ }\textbf {\bibinfo {volume} {86}},\ \bibinfo {pages} {153} (\bibinfo {year} {2014})},\ \Eprint {https://arxiv.org/abs/1308.6253} {arXiv:1308.6253} \BibitemShut {NoStop}%
\bibitem [{\citenamefont {Cirac}\ and\ \citenamefont {Zoller}(2012)}]{ciracGoalsOpportunitiesQuantum2012}%
  \BibitemOpen
  \bibfield  {author} {\bibinfo {author} {\bibfnamefont {J.~I.}\ \bibnamefont {Cirac}}\ and\ \bibinfo {author} {\bibfnamefont {P.}~\bibnamefont {Zoller}},\ }\href {https://doi.org/10.1038/nphys2275} {\bibfield  {journal} {\bibinfo  {journal} {Nature Physics}\ }\textbf {\bibinfo {volume} {8}},\ \bibinfo {pages} {264} (\bibinfo {year} {2012})}\BibitemShut {NoStop}%
\bibitem [{\citenamefont {Daley}\ \emph {et~al.}(2022)\citenamefont {Daley}, \citenamefont {Bloch}, \citenamefont {Kokail}, \citenamefont {Flannigan}, \citenamefont {Pearson}, \citenamefont {Troyer},\ and\ \citenamefont {Zoller}}]{daleyPracticalQuantumAdvantage2022}%
  \BibitemOpen
  \bibfield  {author} {\bibinfo {author} {\bibfnamefont {A.~J.}\ \bibnamefont {Daley}}, \bibinfo {author} {\bibfnamefont {I.}~\bibnamefont {Bloch}}, \bibinfo {author} {\bibfnamefont {C.}~\bibnamefont {Kokail}}, \bibinfo {author} {\bibfnamefont {S.}~\bibnamefont {Flannigan}}, \bibinfo {author} {\bibfnamefont {N.}~\bibnamefont {Pearson}}, \bibinfo {author} {\bibfnamefont {M.}~\bibnamefont {Troyer}},\ and\ \bibinfo {author} {\bibfnamefont {P.}~\bibnamefont {Zoller}},\ }\href {https://doi.org/10.1038/s41586-022-04940-6} {\bibfield  {journal} {\bibinfo  {journal} {Nature}\ }\textbf {\bibinfo {volume} {607}},\ \bibinfo {pages} {667} (\bibinfo {year} {2022})}\BibitemShut {NoStop}%
\bibitem [{\citenamefont {Bernien}\ \emph {et~al.}(2017)\citenamefont {Bernien}, \citenamefont {Schwartz}, \citenamefont {Keesling}, \citenamefont {Levine}, \citenamefont {Omran}, \citenamefont {Pichler}, \citenamefont {Choi}, \citenamefont {Zibrov}, \citenamefont {Endres}, \citenamefont {Greiner}, \citenamefont {Vuleti{\'c}},\ and\ \citenamefont {Lukin}}]{bernienProbingManybodyDynamics2017}%
  \BibitemOpen
  \bibfield  {author} {\bibinfo {author} {\bibfnamefont {H.}~\bibnamefont {Bernien}}, \bibinfo {author} {\bibfnamefont {S.}~\bibnamefont {Schwartz}}, \bibinfo {author} {\bibfnamefont {A.}~\bibnamefont {Keesling}}, \emph {et~al.},\ }\href {https://doi.org/10.1038/nature24622} {\bibfield  {journal} {\bibinfo  {journal} {Nature}\ }\textbf {\bibinfo {volume} {551}},\ \bibinfo {pages} {579} (\bibinfo {year} {2017})},\ \Eprint {https://arxiv.org/abs/1707.04344} {arXiv:1707.04344} \BibitemShut {NoStop}%
\bibitem [{\citenamefont {Monroe}\ \emph {et~al.}(2021)\citenamefont {Monroe}, \citenamefont {Campbell}, \citenamefont {Duan}, \citenamefont {Gong}, \citenamefont {Gorshkov}, \citenamefont {Hess}, \citenamefont {Islam}, \citenamefont {Kim}, \citenamefont {Linke}, \citenamefont {Pagano}, \citenamefont {Richerme}, \citenamefont {Senko},\ and\ \citenamefont {Yao}}]{monroeProgrammableQuantumSimulations2021}%
  \BibitemOpen
  \bibfield  {author} {\bibinfo {author} {\bibfnamefont {C.}~\bibnamefont {Monroe}}, \bibinfo {author} {\bibfnamefont {W.~C.}\ \bibnamefont {Campbell}}, \bibinfo {author} {\bibfnamefont {L.-M.}\ \bibnamefont {Duan}}, \emph {et~al.},\ }\href {https://doi.org/10.1103/RevModPhys.93.025001} {\bibfield  {journal} {\bibinfo  {journal} {Reviews of Modern Physics}\ }\textbf {\bibinfo {volume} {93}},\ \bibinfo {pages} {025001} (\bibinfo {year} {2021})},\ \Eprint {https://arxiv.org/abs/1912.07845} {arXiv:1912.07845} \BibitemShut {NoStop}%
\bibitem [{\citenamefont {Altman}\ \emph {et~al.}(2021)\citenamefont {Altman}, \citenamefont {Brown}, \citenamefont {Carleo}, \citenamefont {Carr}, \citenamefont {Demler}, \citenamefont {Chin}, \citenamefont {DeMarco}, \citenamefont {Economou}, \citenamefont {Eriksson}, \citenamefont {Fu}, \citenamefont {Greiner}, \citenamefont {Hazzard}, \citenamefont {Hulet}, \citenamefont {Koll{\'a}r}, \citenamefont {Lev}, \citenamefont {Lukin}, \citenamefont {Ma}, \citenamefont {Mi}, \citenamefont {Misra}, \citenamefont {Monroe}, \citenamefont {Murch}, \citenamefont {Nazario}, \citenamefont {Ni}, \citenamefont {Potter}, \citenamefont {Roushan}, \citenamefont {Saffman}, \citenamefont {{Schleier-Smith}}, \citenamefont {Siddiqi}, \citenamefont {Simmonds}, \citenamefont {Singh}, \citenamefont {Spielman}, \citenamefont {Temme}, \citenamefont {Weiss}, \citenamefont {Vu{\v c}kovi{\'c}}, \citenamefont {Vuleti{\'c}}, \citenamefont {Ye},\ and\ \citenamefont {Zwierlein}}]{altmanQuantumSimulatorsArchitectures2021}%
  \BibitemOpen
  \bibfield  {author} {\bibinfo {author} {\bibfnamefont {E.}~\bibnamefont {Altman}}, \bibinfo {author} {\bibfnamefont {K.~R.}\ \bibnamefont {Brown}}, \bibinfo {author} {\bibfnamefont {G.}~\bibnamefont {Carleo}}, \emph {et~al.},\ }\href {https://doi.org/10.1103/PRXQuantum.2.017003} {\bibfield  {journal} {\bibinfo  {journal} {PRX Quantum}\ }\textbf {\bibinfo {volume} {2}},\ \bibinfo {pages} {017003} (\bibinfo {year} {2021})},\ \Eprint {https://arxiv.org/abs/1912.06938} {arXiv:1912.06938} \BibitemShut {NoStop}%
\bibitem [{\citenamefont {Kim}\ \emph {et~al.}(2023)\citenamefont {Kim}, \citenamefont {Eddins}, \citenamefont {Anand}, \citenamefont {Wei}, \citenamefont {{van den Berg}}, \citenamefont {Rosenblatt}, \citenamefont {Nayfeh}, \citenamefont {Wu}, \citenamefont {Zaletel}, \citenamefont {Temme},\ and\ \citenamefont {Kandala}}]{kimEvidenceUtilityQuantum2023}%
  \BibitemOpen
  \bibfield  {author} {\bibinfo {author} {\bibfnamefont {Y.}~\bibnamefont {Kim}}, \bibinfo {author} {\bibfnamefont {A.}~\bibnamefont {Eddins}}, \bibinfo {author} {\bibfnamefont {S.}~\bibnamefont {Anand}}, \emph {et~al.},\ }\href {https://doi.org/10.1038/s41586-023-06096-3} {\bibfield  {journal} {\bibinfo  {journal} {Nature}\ }\textbf {\bibinfo {volume} {618}},\ \bibinfo {pages} {500} (\bibinfo {year} {2023})}\BibitemShut {NoStop}%
\bibitem [{\citenamefont {Wang}\ \emph {et~al.}(2024)\citenamefont {Wang}, \citenamefont {Liu}, \citenamefont {Chen}, \citenamefont {Chen}, \citenamefont {Zhao}, \citenamefont {Ying}, \citenamefont {Shang}, \citenamefont {Wang}, \citenamefont {Huo}, \citenamefont {Peng}, \citenamefont {Zhu}, \citenamefont {Lu},\ and\ \citenamefont {Pan}}]{wangRealizationFractionalQuantum2024}%
  \BibitemOpen
  \bibfield  {author} {\bibinfo {author} {\bibfnamefont {C.}~\bibnamefont {Wang}}, \bibinfo {author} {\bibfnamefont {F.-M.}\ \bibnamefont {Liu}}, \bibinfo {author} {\bibfnamefont {M.-C.}\ \bibnamefont {Chen}}, \emph {et~al.},\ }\href {https://doi.org/10.1126/science.ado3912} {\bibfield  {journal} {\bibinfo  {journal} {Science}\ }\textbf {\bibinfo {volume} {384}},\ \bibinfo {pages} {579} (\bibinfo {year} {2024})},\ \Eprint {https://arxiv.org/abs/2401.17022} {arXiv:2401.17022} \BibitemShut {NoStop}%
\bibitem [{\citenamefont {Guo}\ \emph {et~al.}(2024)\citenamefont {Guo}, \citenamefont {Wu}, \citenamefont {Ye}, \citenamefont {Zhang}, \citenamefont {Lian}, \citenamefont {Yao}, \citenamefont {Wang}, \citenamefont {Yan}, \citenamefont {Yi}, \citenamefont {Xu}, \citenamefont {Li}, \citenamefont {Hou}, \citenamefont {Xu}, \citenamefont {Guo}, \citenamefont {Zhang}, \citenamefont {Qi}, \citenamefont {Zhou}, \citenamefont {He},\ and\ \citenamefont {Duan}}]{guoSiteresolvedTwodimensionalQuantum2024}%
  \BibitemOpen
  \bibfield  {author} {\bibinfo {author} {\bibfnamefont {S.-A.}\ \bibnamefont {Guo}}, \bibinfo {author} {\bibfnamefont {Y.-K.}\ \bibnamefont {Wu}}, \bibinfo {author} {\bibfnamefont {J.}~\bibnamefont {Ye}}, \emph {et~al.},\ }\href {https://doi.org/10.1038/s41586-024-07459-0} {\bibfield  {journal} {\bibinfo  {journal} {Nature}\ }\textbf {\bibinfo {volume} {630}},\ \bibinfo {pages} {613} (\bibinfo {year} {2024})},\ \Eprint {https://arxiv.org/abs/2311.17163} {arXiv:2311.17163} \BibitemShut {NoStop}%
\bibitem [{\citenamefont {Lloyd}(1996)}]{lloydUniversalQuantumSimulators1996}%
  \BibitemOpen
  \bibfield  {author} {\bibinfo {author} {\bibfnamefont {S.}~\bibnamefont {Lloyd}},\ }\href {https://doi.org/10.1126/science.273.5278.1073} {\bibfield  {journal} {\bibinfo  {journal} {Science}\ }\textbf {\bibinfo {volume} {273}},\ \bibinfo {pages} {1073} (\bibinfo {year} {1996})}\BibitemShut {NoStop}%
\bibitem [{\citenamefont {Jordan}\ \emph {et~al.}(2012)\citenamefont {Jordan}, \citenamefont {Lee},\ and\ \citenamefont {Preskill}}]{jordanQuantumAlgorithmsQuantum2012}%
  \BibitemOpen
  \bibfield  {author} {\bibinfo {author} {\bibfnamefont {S.~P.}\ \bibnamefont {Jordan}}, \bibinfo {author} {\bibfnamefont {K.~S.~M.}\ \bibnamefont {Lee}},\ and\ \bibinfo {author} {\bibfnamefont {J.}~\bibnamefont {Preskill}},\ }\href {https://doi.org/10.1126/science.1217069} {\bibfield  {journal} {\bibinfo  {journal} {Science}\ }\textbf {\bibinfo {volume} {336}},\ \bibinfo {pages} {1130} (\bibinfo {year} {2012})},\ \Eprint {https://arxiv.org/abs/1111.3633} {arXiv:1111.3633} \BibitemShut {NoStop}%
\bibitem [{\citenamefont {McArdle}\ \emph {et~al.}(2020)\citenamefont {McArdle}, \citenamefont {Endo}, \citenamefont {{Aspuru-Guzik}}, \citenamefont {Benjamin},\ and\ \citenamefont {Yuan}}]{mcardleQuantumComputationalChemistry2020}%
  \BibitemOpen
  \bibfield  {author} {\bibinfo {author} {\bibfnamefont {S.}~\bibnamefont {McArdle}}, \bibinfo {author} {\bibfnamefont {S.}~\bibnamefont {Endo}}, \bibinfo {author} {\bibfnamefont {A.}~\bibnamefont {{Aspuru-Guzik}}}, \bibinfo {author} {\bibfnamefont {S.~C.}\ \bibnamefont {Benjamin}},\ and\ \bibinfo {author} {\bibfnamefont {X.}~\bibnamefont {Yuan}},\ }\href {https://doi.org/10.1103/RevModPhys.92.015003} {\bibfield  {journal} {\bibinfo  {journal} {Reviews of Modern Physics}\ }\textbf {\bibinfo {volume} {92}},\ \bibinfo {pages} {015003} (\bibinfo {year} {2020})},\ \Eprint {https://arxiv.org/abs/1808.10402} {arXiv:1808.10402} \BibitemShut {NoStop}%
\bibitem [{\citenamefont {Lee}\ \emph {et~al.}(2023)\citenamefont {Lee}, \citenamefont {Lee}, \citenamefont {Zhai}, \citenamefont {Tong}, \citenamefont {Dalzell}, \citenamefont {Kumar}, \citenamefont {Helms}, \citenamefont {Gray}, \citenamefont {Cui}, \citenamefont {Liu}, \citenamefont {Kastoryano}, \citenamefont {Babbush}, \citenamefont {Preskill}, \citenamefont {Reichman}, \citenamefont {Campbell}, \citenamefont {Valeev}, \citenamefont {Lin},\ and\ \citenamefont {Chan}}]{leeEvaluatingEvidenceExponential2023}%
  \BibitemOpen
  \bibfield  {author} {\bibinfo {author} {\bibfnamefont {S.}~\bibnamefont {Lee}}, \bibinfo {author} {\bibfnamefont {J.}~\bibnamefont {Lee}}, \bibinfo {author} {\bibfnamefont {H.}~\bibnamefont {Zhai}}, \emph {et~al.},\ }\href {https://doi.org/10.1038/s41467-023-37587-6} {\bibfield  {journal} {\bibinfo  {journal} {Nature Communications}\ }\textbf {\bibinfo {volume} {14}},\ \bibinfo {pages} {1952} (\bibinfo {year} {2023})},\ \Eprint {https://arxiv.org/abs/2208.02199} {arXiv:2208.02199} \BibitemShut {NoStop}%
\bibitem [{\citenamefont {Harrow}\ \emph {et~al.}(2009)\citenamefont {Harrow}, \citenamefont {Hassidim},\ and\ \citenamefont {Lloyd}}]{harrowQuantumAlgorithmSolving2009}%
  \BibitemOpen
  \bibfield  {author} {\bibinfo {author} {\bibfnamefont {A.~W.}\ \bibnamefont {Harrow}}, \bibinfo {author} {\bibfnamefont {A.}~\bibnamefont {Hassidim}},\ and\ \bibinfo {author} {\bibfnamefont {S.}~\bibnamefont {Lloyd}},\ }\href {https://doi.org/10.1103/PhysRevLett.103.150502} {\bibfield  {journal} {\bibinfo  {journal} {Physical Review Letters}\ }\textbf {\bibinfo {volume} {103}},\ \bibinfo {pages} {150502} (\bibinfo {year} {2009})},\ \Eprint {https://arxiv.org/abs/0811.3171} {arXiv:0811.3171} \BibitemShut {NoStop}%
\bibitem [{\citenamefont {Liu}\ \emph {et~al.}(2021)\citenamefont {Liu}, \citenamefont {Kolden}, \citenamefont {Krovi}, \citenamefont {Loureiro}, \citenamefont {Trivisa},\ and\ \citenamefont {Childs}}]{liuEfficientQuantumAlgorithm2021}%
  \BibitemOpen
  \bibfield  {author} {\bibinfo {author} {\bibfnamefont {J.-P.}\ \bibnamefont {Liu}}, \bibinfo {author} {\bibfnamefont {H.~{\O}.}\ \bibnamefont {Kolden}}, \bibinfo {author} {\bibfnamefont {H.~K.}\ \bibnamefont {Krovi}}, \bibinfo {author} {\bibfnamefont {N.~F.}\ \bibnamefont {Loureiro}}, \bibinfo {author} {\bibfnamefont {K.}~\bibnamefont {Trivisa}},\ and\ \bibinfo {author} {\bibfnamefont {A.~M.}\ \bibnamefont {Childs}},\ }\href {https://doi.org/10.1073/pnas.2026805118} {\bibfield  {journal} {\bibinfo  {journal} {Proceedings of the National Academy of Sciences}\ }\textbf {\bibinfo {volume} {118}},\ \bibinfo {pages} {e2026805118} (\bibinfo {year} {2021})},\ \Eprint {https://arxiv.org/abs/2011.03185} {arXiv:2011.03185} \BibitemShut {NoStop}%
\bibitem [{\citenamefont {Farhi}\ \emph {et~al.}(2014)\citenamefont {Farhi}, \citenamefont {Goldstone},\ and\ \citenamefont {Gutmann}}]{farhiQuantumApproximateOptimization2014}%
  \BibitemOpen
  \bibfield  {author} {\bibinfo {author} {\bibfnamefont {E.}~\bibnamefont {Farhi}}, \bibinfo {author} {\bibfnamefont {J.}~\bibnamefont {Goldstone}},\ and\ \bibinfo {author} {\bibfnamefont {S.}~\bibnamefont {Gutmann}},\ }\href@noop {} {\bibinfo {title} {A {{Quantum Approximate Optimization Algorithm}}}} (\bibinfo {year} {2014}),\ \Eprint {https://arxiv.org/abs/1411.4028} {arXiv:1411.4028} \BibitemShut {NoStop}%
\bibitem [{\citenamefont {Albash}\ and\ \citenamefont {Lidar}(2018)}]{albashAdiabaticQuantumComputation2018}%
  \BibitemOpen
  \bibfield  {author} {\bibinfo {author} {\bibfnamefont {T.}~\bibnamefont {Albash}}\ and\ \bibinfo {author} {\bibfnamefont {D.~A.}\ \bibnamefont {Lidar}},\ }\href {https://doi.org/10.1103/RevModPhys.90.015002} {\bibfield  {journal} {\bibinfo  {journal} {Reviews of Modern Physics}\ }\textbf {\bibinfo {volume} {90}},\ \bibinfo {pages} {015002} (\bibinfo {year} {2018})},\ \Eprint {https://arxiv.org/abs/1611.04471} {arXiv:1611.04471} \BibitemShut {NoStop}%
\bibitem [{\citenamefont {Childs}\ and\ \citenamefont {Wiebe}(2012)}]{childsHamiltonianSimulationUsing2012}%
  \BibitemOpen
  \bibfield  {author} {\bibinfo {author} {\bibfnamefont {A.~M.}\ \bibnamefont {Childs}}\ and\ \bibinfo {author} {\bibfnamefont {N.}~\bibnamefont {Wiebe}},\ }\bibfield  {journal} {\bibinfo  {journal} {Quantum Information and Computation}\ }\textbf {\bibinfo {volume} {12}},\ \href {https://doi.org/10.26421/QIC12.11-12} {10.26421/QIC12.11-12} (\bibinfo {year} {2012}),\ \Eprint {https://arxiv.org/abs/1202.5822} {arXiv:1202.5822} \BibitemShut {NoStop}%
\bibitem [{\citenamefont {Berry}\ \emph {et~al.}(2015{\natexlab{a}})\citenamefont {Berry}, \citenamefont {Childs}, \citenamefont {Cleve}, \citenamefont {Kothari},\ and\ \citenamefont {Somma}}]{berrySimulatingHamiltonianDynamics2015}%
  \BibitemOpen
  \bibfield  {author} {\bibinfo {author} {\bibfnamefont {D.~W.}\ \bibnamefont {Berry}}, \bibinfo {author} {\bibfnamefont {A.~M.}\ \bibnamefont {Childs}}, \bibinfo {author} {\bibfnamefont {R.}~\bibnamefont {Cleve}}, \bibinfo {author} {\bibfnamefont {R.}~\bibnamefont {Kothari}},\ and\ \bibinfo {author} {\bibfnamefont {R.~D.}\ \bibnamefont {Somma}},\ }\href {https://doi.org/10.1103/PhysRevLett.114.090502} {\bibfield  {journal} {\bibinfo  {journal} {Physical Review Letters}\ }\textbf {\bibinfo {volume} {114}},\ \bibinfo {pages} {090502} (\bibinfo {year} {2015}{\natexlab{a}})},\ \Eprint {https://arxiv.org/abs/1412.4687} {arXiv:1412.4687} \BibitemShut {NoStop}%
\bibitem [{\citenamefont {Berry}\ \emph {et~al.}(2015{\natexlab{b}})\citenamefont {Berry}, \citenamefont {Childs},\ and\ \citenamefont {Kothari}}]{berryHamiltonianSimulationNearly2015}%
  \BibitemOpen
  \bibfield  {author} {\bibinfo {author} {\bibfnamefont {D.~W.}\ \bibnamefont {Berry}}, \bibinfo {author} {\bibfnamefont {A.~M.}\ \bibnamefont {Childs}},\ and\ \bibinfo {author} {\bibfnamefont {R.}~\bibnamefont {Kothari}},\ }in\ \href {https://doi.org/10.1109/FOCS.2015.54} {\emph {\bibinfo {booktitle} {2015 {{IEEE}} 56th {{Annual Symposium}} on {{Foundations}} of {{Computer Science}}}}}\ (\bibinfo {year} {2015})\ pp.\ \bibinfo {pages} {792--809},\ \Eprint {https://arxiv.org/abs/1501.01715} {arXiv:1501.01715} \BibitemShut {NoStop}%
\bibitem [{\citenamefont {Low}\ and\ \citenamefont {Chuang}(2017)}]{lowOptimalHamiltonianSimulation2017}%
  \BibitemOpen
  \bibfield  {author} {\bibinfo {author} {\bibfnamefont {G.~H.}\ \bibnamefont {Low}}\ and\ \bibinfo {author} {\bibfnamefont {I.~L.}\ \bibnamefont {Chuang}},\ }\href {https://doi.org/10.1103/PhysRevLett.118.010501} {\bibfield  {journal} {\bibinfo  {journal} {Physical Review Letters}\ }\textbf {\bibinfo {volume} {118}},\ \bibinfo {pages} {010501} (\bibinfo {year} {2017})},\ \Eprint {https://arxiv.org/abs/1606.02685} {arXiv:1606.02685} \BibitemShut {NoStop}%
\bibitem [{\citenamefont {Haah}\ \emph {et~al.}(2018)\citenamefont {Haah}, \citenamefont {Hastings}, \citenamefont {Kothari},\ and\ \citenamefont {Low}}]{haahQuantumAlgorithmSimulating2018}%
  \BibitemOpen
  \bibfield  {author} {\bibinfo {author} {\bibfnamefont {J.}~\bibnamefont {Haah}}, \bibinfo {author} {\bibfnamefont {M.}~\bibnamefont {Hastings}}, \bibinfo {author} {\bibfnamefont {R.}~\bibnamefont {Kothari}},\ and\ \bibinfo {author} {\bibfnamefont {G.~H.}\ \bibnamefont {Low}},\ }in\ \href {https://doi.org/10.1109/FOCS.2018.00041} {\emph {\bibinfo {booktitle} {2018 {{IEEE}} 59th {{Annual Symposium}} on {{Foundations}} of {{Computer Science}} ({{FOCS}})}}}\ (\bibinfo {year} {2018})\ pp.\ \bibinfo {pages} {350--360},\ \Eprint {https://arxiv.org/abs/1801.03922} {arXiv:1801.03922} \BibitemShut {NoStop}%
\bibitem [{\citenamefont {Berry}\ \emph {et~al.}(2007)\citenamefont {Berry}, \citenamefont {Ahokas}, \citenamefont {Cleve},\ and\ \citenamefont {Sanders}}]{berryEfficientQuantumAlgorithms2007}%
  \BibitemOpen
  \bibfield  {author} {\bibinfo {author} {\bibfnamefont {D.~W.}\ \bibnamefont {Berry}}, \bibinfo {author} {\bibfnamefont {G.}~\bibnamefont {Ahokas}}, \bibinfo {author} {\bibfnamefont {R.}~\bibnamefont {Cleve}},\ and\ \bibinfo {author} {\bibfnamefont {B.~C.}\ \bibnamefont {Sanders}},\ }\href {https://doi.org/10.1007/s00220-006-0150-x} {\bibfield  {journal} {\bibinfo  {journal} {Communications in Mathematical Physics}\ }\textbf {\bibinfo {volume} {270}},\ \bibinfo {pages} {359} (\bibinfo {year} {2007})},\ \Eprint {https://arxiv.org/abs/quant-ph/0508139} {arXiv:quant-ph/0508139} \BibitemShut {NoStop}%
\bibitem [{\citenamefont {Campbell}(2019)}]{campbellRandomCompilerFast2019}%
  \BibitemOpen
  \bibfield  {author} {\bibinfo {author} {\bibfnamefont {E.}~\bibnamefont {Campbell}},\ }\href {https://doi.org/10.1103/PhysRevLett.123.070503} {\bibfield  {journal} {\bibinfo  {journal} {Physical Review Letters}\ }\textbf {\bibinfo {volume} {123}},\ \bibinfo {pages} {070503} (\bibinfo {year} {2019})},\ \Eprint {https://arxiv.org/abs/1811.08017} {arXiv:1811.08017} \BibitemShut {NoStop}%
\bibitem [{\citenamefont {Childs}\ \emph {et~al.}(2019)\citenamefont {Childs}, \citenamefont {Ostrander},\ and\ \citenamefont {Su}}]{childsFasterQuantumSimulation2019}%
  \BibitemOpen
  \bibfield  {author} {\bibinfo {author} {\bibfnamefont {A.~M.}\ \bibnamefont {Childs}}, \bibinfo {author} {\bibfnamefont {A.}~\bibnamefont {Ostrander}},\ and\ \bibinfo {author} {\bibfnamefont {Y.}~\bibnamefont {Su}},\ }\href {https://doi.org/10.22331/q-2019-09-02-182} {\bibfield  {journal} {\bibinfo  {journal} {Quantum}\ }\textbf {\bibinfo {volume} {3}},\ \bibinfo {pages} {182} (\bibinfo {year} {2019})},\ \Eprint {https://arxiv.org/abs/1805.08385} {arXiv:1805.08385} \BibitemShut {NoStop}%
\bibitem [{\citenamefont {An}\ \emph {et~al.}(2021)\citenamefont {An}, \citenamefont {Fang},\ and\ \citenamefont {Lin}}]{anTimedependentUnboundedHamiltonian2021}%
  \BibitemOpen
  \bibfield  {author} {\bibinfo {author} {\bibfnamefont {D.}~\bibnamefont {An}}, \bibinfo {author} {\bibfnamefont {D.}~\bibnamefont {Fang}},\ and\ \bibinfo {author} {\bibfnamefont {L.}~\bibnamefont {Lin}},\ }\href {https://doi.org/10.22331/q-2021-05-26-459} {\bibfield  {journal} {\bibinfo  {journal} {Quantum}\ }\textbf {\bibinfo {volume} {5}},\ \bibinfo {pages} {459} (\bibinfo {year} {2021})},\ \Eprint {https://arxiv.org/abs/2012.13105} {arXiv:2012.13105} \BibitemShut {NoStop}%
\bibitem [{\citenamefont {Bosse}\ \emph {et~al.}(2025)\citenamefont {Bosse}, \citenamefont {Childs}, \citenamefont {Derby}, \citenamefont {Gambetta}, \citenamefont {Montanaro},\ and\ \citenamefont {Santos}}]{bosseEfficientPracticalHamiltonian2025}%
  \BibitemOpen
  \bibfield  {author} {\bibinfo {author} {\bibfnamefont {J.~L.}\ \bibnamefont {Bosse}}, \bibinfo {author} {\bibfnamefont {A.~M.}\ \bibnamefont {Childs}}, \bibinfo {author} {\bibfnamefont {C.}~\bibnamefont {Derby}}, \bibinfo {author} {\bibfnamefont {F.~M.}\ \bibnamefont {Gambetta}}, \bibinfo {author} {\bibfnamefont {A.}~\bibnamefont {Montanaro}},\ and\ \bibinfo {author} {\bibfnamefont {R.~A.}\ \bibnamefont {Santos}},\ }\href {https://doi.org/10.1038/s41467-025-57580-5} {\bibfield  {journal} {\bibinfo  {journal} {Nature Communications}\ }\textbf {\bibinfo {volume} {16}},\ \bibinfo {pages} {2673} (\bibinfo {year} {2025})},\ \Eprint {https://arxiv.org/abs/2403.08729} {arXiv:2403.08729} \BibitemShut {NoStop}%
\bibitem [{\citenamefont {Fang}\ \emph {et~al.}(2025)\citenamefont {Fang}, \citenamefont {Wu},\ and\ \citenamefont {Soffer}}]{fangTrotterErrorManybody2025}%
  \BibitemOpen
  \bibfield  {author} {\bibinfo {author} {\bibfnamefont {D.}~\bibnamefont {Fang}}, \bibinfo {author} {\bibfnamefont {X.}~\bibnamefont {Wu}},\ and\ \bibinfo {author} {\bibfnamefont {A.}~\bibnamefont {Soffer}},\ }\href {https://doi.org/10.48550/arXiv.2507.22707} {\bibinfo {title} {On the {{Trotter Error}} in {{Many-body Quantum Dynamics}} with {{Coulomb Potentials}}}} (\bibinfo {year} {2025}),\ \Eprint {https://arxiv.org/abs/2507.22707} {arXiv:2507.22707} \BibitemShut {NoStop}%
\bibitem [{\citenamefont {Childs}\ \emph {et~al.}(2018)\citenamefont {Childs}, \citenamefont {Maslov}, \citenamefont {Nam}, \citenamefont {Ross},\ and\ \citenamefont {Su}}]{childsFirstQuantumSimulation2018}%
  \BibitemOpen
  \bibfield  {author} {\bibinfo {author} {\bibfnamefont {A.~M.}\ \bibnamefont {Childs}}, \bibinfo {author} {\bibfnamefont {D.}~\bibnamefont {Maslov}}, \bibinfo {author} {\bibfnamefont {Y.}~\bibnamefont {Nam}}, \bibinfo {author} {\bibfnamefont {N.~J.}\ \bibnamefont {Ross}},\ and\ \bibinfo {author} {\bibfnamefont {Y.}~\bibnamefont {Su}},\ }\href {https://doi.org/10.1073/pnas.1801723115} {\bibfield  {journal} {\bibinfo  {journal} {Proceedings of the National Academy of Sciences}\ }\textbf {\bibinfo {volume} {115}},\ \bibinfo {pages} {9456} (\bibinfo {year} {2018})},\ \Eprint {https://arxiv.org/abs/1711.10980} {arXiv:1711.10980} \BibitemShut {NoStop}%
\bibitem [{\citenamefont {Childs}\ and\ \citenamefont {Su}(2019)}]{childsNearlyOptimalLattice2019}%
  \BibitemOpen
  \bibfield  {author} {\bibinfo {author} {\bibfnamefont {A.~M.}\ \bibnamefont {Childs}}\ and\ \bibinfo {author} {\bibfnamefont {Y.}~\bibnamefont {Su}},\ }\href {https://doi.org/10.1103/PhysRevLett.123.050503} {\bibfield  {journal} {\bibinfo  {journal} {Physical Review Letters}\ }\textbf {\bibinfo {volume} {123}},\ \bibinfo {pages} {050503} (\bibinfo {year} {2019})},\ \Eprint {https://arxiv.org/abs/1901.00564} {arXiv:1901.00564} \BibitemShut {NoStop}%
\bibitem [{\citenamefont {Childs}\ \emph {et~al.}(2021)\citenamefont {Childs}, \citenamefont {Su}, \citenamefont {Tran}, \citenamefont {Wiebe},\ and\ \citenamefont {Zhu}}]{childsTheoryTrotterError2021}%
  \BibitemOpen
  \bibfield  {author} {\bibinfo {author} {\bibfnamefont {A.~M.}\ \bibnamefont {Childs}}, \bibinfo {author} {\bibfnamefont {Y.}~\bibnamefont {Su}}, \bibinfo {author} {\bibfnamefont {M.~C.}\ \bibnamefont {Tran}}, \bibinfo {author} {\bibfnamefont {N.}~\bibnamefont {Wiebe}},\ and\ \bibinfo {author} {\bibfnamefont {S.}~\bibnamefont {Zhu}},\ }\href {https://doi.org/10.1103/PhysRevX.11.011020} {\bibfield  {journal} {\bibinfo  {journal} {Physical Review X}\ }\textbf {\bibinfo {volume} {11}},\ \bibinfo {pages} {011020} (\bibinfo {year} {2021})},\ \Eprint {https://arxiv.org/abs/1912.08854} {arXiv:1912.08854} \BibitemShut {NoStop}%
\bibitem [{\citenamefont {Zhao}\ \emph {et~al.}(2021)\citenamefont {Zhao}, \citenamefont {Zhou}, \citenamefont {Shaw}, \citenamefont {Li},\ and\ \citenamefont {Childs}}]{zhaoHamiltonianSimulationRandom2021}%
  \BibitemOpen
  \bibfield  {author} {\bibinfo {author} {\bibfnamefont {Q.}~\bibnamefont {Zhao}}, \bibinfo {author} {\bibfnamefont {Y.}~\bibnamefont {Zhou}}, \bibinfo {author} {\bibfnamefont {A.~F.}\ \bibnamefont {Shaw}}, \bibinfo {author} {\bibfnamefont {T.}~\bibnamefont {Li}},\ and\ \bibinfo {author} {\bibfnamefont {A.~M.}\ \bibnamefont {Childs}},\ }\href {https://doi.org/10.1103/PhysRevLett.129.270502} {\bibfield  {journal} {\bibinfo  {journal} {Physical Review Letters}\ }\textbf {\bibinfo {volume} {127}},\ \bibinfo {pages} {270502} (\bibinfo {year} {2021})},\ \Eprint {https://arxiv.org/abs/2111.04773} {arXiv:2111.04773} \BibitemShut {NoStop}%
\bibitem [{\citenamefont {{\c S}ahino{\u g}lu}\ and\ \citenamefont {Somma}(2021)}]{sahinogluHamiltonianSimulationLowenergy2021}%
  \BibitemOpen
  \bibfield  {author} {\bibinfo {author} {\bibfnamefont {B.}~\bibnamefont {{\c S}ahino{\u g}lu}}\ and\ \bibinfo {author} {\bibfnamefont {R.~D.}\ \bibnamefont {Somma}},\ }\href {https://doi.org/10.1038/s41534-021-00451-w} {\bibfield  {journal} {\bibinfo  {journal} {npj Quantum Information}\ }\textbf {\bibinfo {volume} {7}},\ \bibinfo {pages} {119} (\bibinfo {year} {2021})},\ \Eprint {https://arxiv.org/abs/2006.02660} {arXiv:2006.02660} \BibitemShut {NoStop}%
\bibitem [{\citenamefont {Gong}\ \emph {et~al.}(2024)\citenamefont {Gong}, \citenamefont {Zhou},\ and\ \citenamefont {Li}}]{gongComplexityDigitalQuantum2024}%
  \BibitemOpen
  \bibfield  {author} {\bibinfo {author} {\bibfnamefont {W.}~\bibnamefont {Gong}}, \bibinfo {author} {\bibfnamefont {S.}~\bibnamefont {Zhou}},\ and\ \bibinfo {author} {\bibfnamefont {T.}~\bibnamefont {Li}},\ }\href {https://doi.org/10.22331/q-2024-07-15-1409} {\bibfield  {journal} {\bibinfo  {journal} {Quantum}\ }\textbf {\bibinfo {volume} {8}},\ \bibinfo {pages} {1409} (\bibinfo {year} {2024})},\ \Eprint {https://arxiv.org/abs/2312.08867} {arXiv:2312.08867} \BibitemShut {NoStop}%
\bibitem [{\citenamefont {Zlokapa}\ and\ \citenamefont {Somma}(2024)}]{zlokapaHamiltonianSimulationLowenergy2024}%
  \BibitemOpen
  \bibfield  {author} {\bibinfo {author} {\bibfnamefont {A.}~\bibnamefont {Zlokapa}}\ and\ \bibinfo {author} {\bibfnamefont {R.~D.}\ \bibnamefont {Somma}},\ }\href {https://doi.org/10.22331/q-2024-08-27-1449} {\bibfield  {journal} {\bibinfo  {journal} {Quantum}\ }\textbf {\bibinfo {volume} {8}},\ \bibinfo {pages} {1449} (\bibinfo {year} {2024})},\ \Eprint {https://arxiv.org/abs/2404.03644} {arXiv:2404.03644} \BibitemShut {NoStop}%
\bibitem [{\citenamefont {Zhao}\ \emph {et~al.}(2025)\citenamefont {Zhao}, \citenamefont {Zhou},\ and\ \citenamefont {Childs}}]{zhaoEntanglementAcceleratesQuantum2025}%
  \BibitemOpen
  \bibfield  {author} {\bibinfo {author} {\bibfnamefont {Q.}~\bibnamefont {Zhao}}, \bibinfo {author} {\bibfnamefont {Y.}~\bibnamefont {Zhou}},\ and\ \bibinfo {author} {\bibfnamefont {A.~M.}\ \bibnamefont {Childs}},\ }\href {https://doi.org/10.1038/s41567-025-02945-2} {\bibfield  {journal} {\bibinfo  {journal} {Nature Physics}\ }\textbf {\bibinfo {volume} {21}},\ \bibinfo {pages} {1338} (\bibinfo {year} {2025})},\ \Eprint {https://arxiv.org/abs/2406.02379} {arXiv:2406.02379} \BibitemShut {NoStop}%
\bibitem [{\citenamefont {Mizuta}\ and\ \citenamefont {Kuwahara}(2025)}]{mizutaTrotterizationSubstantiallyEfficient2025}%
  \BibitemOpen
  \bibfield  {author} {\bibinfo {author} {\bibfnamefont {K.}~\bibnamefont {Mizuta}}\ and\ \bibinfo {author} {\bibfnamefont {T.}~\bibnamefont {Kuwahara}},\ }\href {https://doi.org/10.1103/q87n-5xhz} {\bibfield  {journal} {\bibinfo  {journal} {Physical Review Letters}\ }\textbf {\bibinfo {volume} {135}},\ \bibinfo {pages} {130602} (\bibinfo {year} {2025})},\ \Eprint {https://arxiv.org/abs/2504.20746} {arXiv:2504.20746} \BibitemShut {NoStop}%
\bibitem [{\citenamefont {Zhang}\ \emph {et~al.}(2017)\citenamefont {Zhang}, \citenamefont {Pagano}, \citenamefont {Hess}, \citenamefont {Kyprianidis}, \citenamefont {Becker}, \citenamefont {Kaplan}, \citenamefont {Gorshkov}, \citenamefont {Gong},\ and\ \citenamefont {Monroe}}]{zhangObservationManyBodyDynamical2017}%
  \BibitemOpen
  \bibfield  {author} {\bibinfo {author} {\bibfnamefont {J.}~\bibnamefont {Zhang}}, \bibinfo {author} {\bibfnamefont {G.}~\bibnamefont {Pagano}}, \bibinfo {author} {\bibfnamefont {P.~W.}\ \bibnamefont {Hess}}, \bibinfo {author} {\bibfnamefont {A.}~\bibnamefont {Kyprianidis}}, \bibinfo {author} {\bibfnamefont {P.}~\bibnamefont {Becker}}, \bibinfo {author} {\bibfnamefont {H.}~\bibnamefont {Kaplan}}, \bibinfo {author} {\bibfnamefont {A.~V.}\ \bibnamefont {Gorshkov}}, \bibinfo {author} {\bibfnamefont {Z.-X.}\ \bibnamefont {Gong}},\ and\ \bibinfo {author} {\bibfnamefont {C.}~\bibnamefont {Monroe}},\ }\href {https://doi.org/10.1038/nature24654} {\bibfield  {journal} {\bibinfo  {journal} {Nature}\ }\textbf {\bibinfo {volume} {551}},\ \bibinfo {pages} {601} (\bibinfo {year} {2017})},\ \Eprint {https://arxiv.org/abs/1708.01044} {arXiv:1708.01044} \BibitemShut {NoStop}%
\bibitem [{\citenamefont {Preskill}(2018)}]{preskillQuantumComputingNISQ2018}%
  \BibitemOpen
  \bibfield  {author} {\bibinfo {author} {\bibfnamefont {J.}~\bibnamefont {Preskill}},\ }\href {https://doi.org/10.22331/q-2018-08-06-79} {\bibfield  {journal} {\bibinfo  {journal} {Quantum}\ }\textbf {\bibinfo {volume} {2}},\ \bibinfo {pages} {79} (\bibinfo {year} {2018})},\ \Eprint {https://arxiv.org/abs/1801.00862} {arXiv:1801.00862} \BibitemShut {NoStop}%
\bibitem [{\citenamefont {Bharti}\ \emph {et~al.}(2022)\citenamefont {Bharti}, \citenamefont {{Cervera-Lierta}}, \citenamefont {Kyaw}, \citenamefont {Haug}, \citenamefont {{Alperin-Lea}}, \citenamefont {Anand}, \citenamefont {Degroote}, \citenamefont {Heimonen}, \citenamefont {Kottmann}, \citenamefont {Menke}, \citenamefont {Mok}, \citenamefont {Sim}, \citenamefont {Kwek},\ and\ \citenamefont {{Aspuru-Guzik}}}]{bhartiNoisyIntermediatescaleQuantum2022}%
  \BibitemOpen
  \bibfield  {author} {\bibinfo {author} {\bibfnamefont {K.}~\bibnamefont {Bharti}}, \bibinfo {author} {\bibfnamefont {A.}~\bibnamefont {{Cervera-Lierta}}}, \bibinfo {author} {\bibfnamefont {T.~H.}\ \bibnamefont {Kyaw}}, \emph {et~al.},\ }\href {https://doi.org/10.1103/RevModPhys.94.015004} {\bibfield  {journal} {\bibinfo  {journal} {Reviews of Modern Physics}\ }\textbf {\bibinfo {volume} {94}},\ \bibinfo {pages} {015004} (\bibinfo {year} {2022})},\ \Eprint {https://arxiv.org/abs/2101.08448} {arXiv:2101.08448} \BibitemShut {NoStop}%
\bibitem [{\citenamefont {Begu{\v s}i{\'c}}\ \emph {et~al.}(2024)\citenamefont {Begu{\v s}i{\'c}}, \citenamefont {Gray},\ and\ \citenamefont {Chan}}]{begusicFastConvergedClassical2024}%
  \BibitemOpen
  \bibfield  {author} {\bibinfo {author} {\bibfnamefont {T.}~\bibnamefont {Begu{\v s}i{\'c}}}, \bibinfo {author} {\bibfnamefont {J.}~\bibnamefont {Gray}},\ and\ \bibinfo {author} {\bibfnamefont {G.~K.-L.}\ \bibnamefont {Chan}},\ }\href {https://doi.org/10.1126/sciadv.adk4321} {\bibfield  {journal} {\bibinfo  {journal} {Science Advances}\ }\textbf {\bibinfo {volume} {10}},\ \bibinfo {pages} {eadk4321} (\bibinfo {year} {2024})},\ \Eprint {https://arxiv.org/abs/2308.05077} {arXiv:2308.05077} \BibitemShut {NoStop}%
\bibitem [{\citenamefont {Tindall}\ \emph {et~al.}(2024)\citenamefont {Tindall}, \citenamefont {Fishman}, \citenamefont {Stoudenmire},\ and\ \citenamefont {Sels}}]{tindallEfficientTensorNetwork2024}%
  \BibitemOpen
  \bibfield  {author} {\bibinfo {author} {\bibfnamefont {J.}~\bibnamefont {Tindall}}, \bibinfo {author} {\bibfnamefont {M.}~\bibnamefont {Fishman}}, \bibinfo {author} {\bibfnamefont {E.~M.}\ \bibnamefont {Stoudenmire}},\ and\ \bibinfo {author} {\bibfnamefont {D.}~\bibnamefont {Sels}},\ }\href {https://doi.org/10.1103/PRXQuantum.5.010308} {\bibfield  {journal} {\bibinfo  {journal} {PRX Quantum}\ }\textbf {\bibinfo {volume} {5}},\ \bibinfo {pages} {010308} (\bibinfo {year} {2024})},\ \Eprint {https://arxiv.org/abs/2306.14887} {arXiv:2306.14887} \BibitemShut {NoStop}%
\bibitem [{\citenamefont {Shao}\ \emph {et~al.}(2024)\citenamefont {Shao}, \citenamefont {Wei}, \citenamefont {Cheng},\ and\ \citenamefont {Liu}}]{shaoSimulatingNoisyVariational2024}%
  \BibitemOpen
  \bibfield  {author} {\bibinfo {author} {\bibfnamefont {Y.}~\bibnamefont {Shao}}, \bibinfo {author} {\bibfnamefont {F.}~\bibnamefont {Wei}}, \bibinfo {author} {\bibfnamefont {S.}~\bibnamefont {Cheng}},\ and\ \bibinfo {author} {\bibfnamefont {Z.}~\bibnamefont {Liu}},\ }\href {https://doi.org/10.1103/PhysRevLett.133.120603} {\bibfield  {journal} {\bibinfo  {journal} {Physical Review Letters}\ }\textbf {\bibinfo {volume} {133}},\ \bibinfo {pages} {120603} (\bibinfo {year} {2024})},\ \Eprint {https://arxiv.org/abs/2306.05804} {arXiv:2306.05804} \BibitemShut {NoStop}%
\bibitem [{\citenamefont {Fontana}\ \emph {et~al.}(2025)\citenamefont {Fontana}, \citenamefont {Rudolph}, \citenamefont {Duncan}, \citenamefont {Rungger},\ and\ \citenamefont {C{\^i}rstoiu}}]{fontanaClassicalSimulationsNoisy2025}%
  \BibitemOpen
  \bibfield  {author} {\bibinfo {author} {\bibfnamefont {E.}~\bibnamefont {Fontana}}, \bibinfo {author} {\bibfnamefont {M.~S.}\ \bibnamefont {Rudolph}}, \bibinfo {author} {\bibfnamefont {R.}~\bibnamefont {Duncan}}, \bibinfo {author} {\bibfnamefont {I.}~\bibnamefont {Rungger}},\ and\ \bibinfo {author} {\bibfnamefont {C.}~\bibnamefont {C{\^i}rstoiu}},\ }\href {https://doi.org/10.1038/s41534-024-00955-1} {\bibfield  {journal} {\bibinfo  {journal} {npj Quantum Information}\ }\textbf {\bibinfo {volume} {11}},\ \bibinfo {pages} {1} (\bibinfo {year} {2025})},\ \Eprint {https://arxiv.org/abs/2306.05400} {arXiv:2306.05400} \BibitemShut {NoStop}%
\bibitem [{\citenamefont {Martinez}\ \emph {et~al.}(2025)\citenamefont {Martinez}, \citenamefont {Angrisani}, \citenamefont {Pankovets}, \citenamefont {Fawzi},\ and\ \citenamefont {Stilck~Fran{\c c}a}}]{martinezEfficientSimulationParametrized2025}%
  \BibitemOpen
  \bibfield  {author} {\bibinfo {author} {\bibfnamefont {V.}~\bibnamefont {Martinez}}, \bibinfo {author} {\bibfnamefont {A.}~\bibnamefont {Angrisani}}, \bibinfo {author} {\bibfnamefont {E.}~\bibnamefont {Pankovets}}, \bibinfo {author} {\bibfnamefont {O.}~\bibnamefont {Fawzi}},\ and\ \bibinfo {author} {\bibfnamefont {D.}~\bibnamefont {Stilck~Fran{\c c}a}},\ }\href {https://doi.org/10.1103/j1gg-s6zb} {\bibfield  {journal} {\bibinfo  {journal} {Physical Review Letters}\ }\textbf {\bibinfo {volume} {134}},\ \bibinfo {pages} {250602} (\bibinfo {year} {2025})},\ \Eprint {https://arxiv.org/abs/2501.13050} {arXiv:2501.13050} \BibitemShut {NoStop}%
\bibitem [{\citenamefont {Knee}\ and\ \citenamefont {Munro}(2015)}]{kneeOptimalTrotterizationUniversal2015}%
  \BibitemOpen
  \bibfield  {author} {\bibinfo {author} {\bibfnamefont {G.~C.}\ \bibnamefont {Knee}}\ and\ \bibinfo {author} {\bibfnamefont {W.~J.}\ \bibnamefont {Munro}},\ }\href {https://doi.org/10.1103/PhysRevA.91.052327} {\bibfield  {journal} {\bibinfo  {journal} {Physical Review A}\ }\textbf {\bibinfo {volume} {91}},\ \bibinfo {pages} {052327} (\bibinfo {year} {2015})},\ \Eprint {https://arxiv.org/abs/1502.04536} {arXiv:1502.04536} \BibitemShut {NoStop}%
\bibitem [{\citenamefont {Endo}\ \emph {et~al.}(2019)\citenamefont {Endo}, \citenamefont {Zhao}, \citenamefont {Li}, \citenamefont {Benjamin},\ and\ \citenamefont {Yuan}}]{endoMitigatingAlgorithmicErrors2019}%
  \BibitemOpen
  \bibfield  {author} {\bibinfo {author} {\bibfnamefont {S.}~\bibnamefont {Endo}}, \bibinfo {author} {\bibfnamefont {Q.}~\bibnamefont {Zhao}}, \bibinfo {author} {\bibfnamefont {Y.}~\bibnamefont {Li}}, \bibinfo {author} {\bibfnamefont {S.}~\bibnamefont {Benjamin}},\ and\ \bibinfo {author} {\bibfnamefont {X.}~\bibnamefont {Yuan}},\ }\href {https://doi.org/10.1103/PhysRevA.99.012334} {\bibfield  {journal} {\bibinfo  {journal} {Physical Review A}\ }\textbf {\bibinfo {volume} {99}},\ \bibinfo {pages} {012334} (\bibinfo {year} {2019})},\ \Eprint {https://arxiv.org/abs/1808.03623} {arXiv:1808.03623} \BibitemShut {NoStop}%
\bibitem [{\citenamefont {Hakkaku}\ \emph {et~al.}(2025)\citenamefont {Hakkaku}, \citenamefont {Suzuki}, \citenamefont {Tokunaga},\ and\ \citenamefont {Endo}}]{hakkakuDataEfficientErrorMitigation2025}%
  \BibitemOpen
  \bibfield  {author} {\bibinfo {author} {\bibfnamefont {S.}~\bibnamefont {Hakkaku}}, \bibinfo {author} {\bibfnamefont {Y.}~\bibnamefont {Suzuki}}, \bibinfo {author} {\bibfnamefont {Y.}~\bibnamefont {Tokunaga}},\ and\ \bibinfo {author} {\bibfnamefont {S.}~\bibnamefont {Endo}},\ }\href {https://doi.org/10.48550/arXiv.2503.05052} {\bibinfo {title} {Data-{{Efficient Error Mitigation}} for {{Physical}} and {{Algorithmic Errors}} in a {{Hamiltonian Simulation}}}} (\bibinfo {year} {2025}),\ \Eprint {https://arxiv.org/abs/2503.05052} {arXiv:2503.05052} \BibitemShut {NoStop}%
\bibitem [{\citenamefont {{Ben-Or}}\ \emph {et~al.}(2013)\citenamefont {{Ben-Or}}, \citenamefont {Gottesman},\ and\ \citenamefont {Hassidim}}]{ben-orQuantumRefrigerator2013}%
  \BibitemOpen
  \bibfield  {author} {\bibinfo {author} {\bibfnamefont {M.}~\bibnamefont {{Ben-Or}}}, \bibinfo {author} {\bibfnamefont {D.}~\bibnamefont {Gottesman}},\ and\ \bibinfo {author} {\bibfnamefont {A.}~\bibnamefont {Hassidim}},\ }\href {http://arxiv.org/abs/1301.1995} {\bibinfo {title} {Quantum {{Refrigerator}}}} (\bibinfo {year} {2013}),\ \Eprint {https://arxiv.org/abs/1301.1995} {arXiv:1301.1995} \BibitemShut {NoStop}%
\bibitem [{\citenamefont {{M{\"u}ller-Hermes}}\ \emph {et~al.}(2016)\citenamefont {{M{\"u}ller-Hermes}}, \citenamefont {Franca},\ and\ \citenamefont {Wolf}}]{muller-hermesRelativeEntropyConvergence2016}%
  \BibitemOpen
  \bibfield  {author} {\bibinfo {author} {\bibfnamefont {A.}~\bibnamefont {{M{\"u}ller-Hermes}}}, \bibinfo {author} {\bibfnamefont {D.~S.}\ \bibnamefont {Franca}},\ and\ \bibinfo {author} {\bibfnamefont {M.~M.}\ \bibnamefont {Wolf}},\ }\href {https://doi.org/10.1063/1.4939560} {\bibfield  {journal} {\bibinfo  {journal} {Journal of Mathematical Physics}\ }\textbf {\bibinfo {volume} {57}},\ \bibinfo {pages} {022202} (\bibinfo {year} {2016})},\ \Eprint {https://arxiv.org/abs/1508.07021} {arXiv:1508.07021} \BibitemShut {NoStop}%
\bibitem [{\citenamefont {Franca}\ and\ \citenamefont {{Garcia-Patron}}(2021)}]{francaLimitationsOptimizationAlgorithms2021}%
  \BibitemOpen
  \bibfield  {author} {\bibinfo {author} {\bibfnamefont {D.~S.}\ \bibnamefont {Franca}}\ and\ \bibinfo {author} {\bibfnamefont {R.}~\bibnamefont {{Garcia-Patron}}},\ }\href {https://doi.org/10.1038/s41567-021-01356-3} {\bibfield  {journal} {\bibinfo  {journal} {Nature Physics}\ }\textbf {\bibinfo {volume} {17}},\ \bibinfo {pages} {1221} (\bibinfo {year} {2021})},\ \Eprint {https://arxiv.org/abs/2009.05532} {arXiv:2009.05532} \BibitemShut {NoStop}%
\bibitem [{\citenamefont {Suzuki}(1991)}]{suzukiGeneralTheoryFractal1991}%
  \BibitemOpen
  \bibfield  {author} {\bibinfo {author} {\bibfnamefont {M.}~\bibnamefont {Suzuki}},\ }\href {https://doi.org/10.1063/1.529425} {\bibfield  {journal} {\bibinfo  {journal} {Journal of Mathematical Physics}\ }\textbf {\bibinfo {volume} {32}},\ \bibinfo {pages} {400} (\bibinfo {year} {1991})}\BibitemShut {NoStop}%
\bibitem [{Note1()}]{Note1}%
  \BibitemOpen
  \bibinfo {note} {The state $\rho $ is left alone with probability $1-\gamma $, and the operators $X$, $Y$ and $Z$ applied each with probability $\gamma /3$. Equivalently, $\protect \mathcal {E}_{\gamma }^{\protect \textup {depo}}(\rho )=(1-\gamma ^\prime )\rho + \protect \frac {\gamma ^\prime }{2} I$ where $\gamma ^\prime = \protect \frac {4}{3}\gamma $}\BibitemShut {NoStop}%
\bibitem [{\citenamefont {Clinton}\ \emph {et~al.}(2021)\citenamefont {Clinton}, \citenamefont {Bausch},\ and\ \citenamefont {Cubitt}}]{clintonHamiltonianSimulationAlgorithms2021}%
  \BibitemOpen
  \bibfield  {author} {\bibinfo {author} {\bibfnamefont {L.}~\bibnamefont {Clinton}}, \bibinfo {author} {\bibfnamefont {J.}~\bibnamefont {Bausch}},\ and\ \bibinfo {author} {\bibfnamefont {T.}~\bibnamefont {Cubitt}},\ }\href {https://doi.org/10.1038/s41467-021-25196-0} {\bibfield  {journal} {\bibinfo  {journal} {Nature Communications}\ }\textbf {\bibinfo {volume} {12}},\ \bibinfo {pages} {4989} (\bibinfo {year} {2021})},\ \Eprint {https://arxiv.org/abs/2003.06886} {arXiv:2003.06886} \BibitemShut {NoStop}%
\bibitem [{\citenamefont {Chen}\ and\ \citenamefont {Brand{\~a}o}(2024)}]{chenAverageCaseSpeedupProduct2024}%
  \BibitemOpen
  \bibfield  {author} {\bibinfo {author} {\bibfnamefont {C.-F.}\ \bibnamefont {Chen}}\ and\ \bibinfo {author} {\bibfnamefont {F.~G. S.~L.}\ \bibnamefont {Brand{\~a}o}},\ }\href {https://doi.org/10.1007/s00220-023-04912-5} {\bibfield  {journal} {\bibinfo  {journal} {Communications in Mathematical Physics}\ }\textbf {\bibinfo {volume} {405}},\ \bibinfo {pages} {32} (\bibinfo {year} {2024})},\ \Eprint {https://arxiv.org/abs/2111.05324} {arXiv:2111.05324} \BibitemShut {NoStop}%
\bibitem [{\citenamefont {Heyl}\ \emph {et~al.}(2019)\citenamefont {Heyl}, \citenamefont {Hauke},\ and\ \citenamefont {Zoller}}]{heylQuantumLocalizationBounds2019}%
  \BibitemOpen
  \bibfield  {author} {\bibinfo {author} {\bibfnamefont {M.}~\bibnamefont {Heyl}}, \bibinfo {author} {\bibfnamefont {P.}~\bibnamefont {Hauke}},\ and\ \bibinfo {author} {\bibfnamefont {P.}~\bibnamefont {Zoller}},\ }\href {https://doi.org/10.1126/sciadv.aau8342} {\bibfield  {journal} {\bibinfo  {journal} {Science Advances}\ }\textbf {\bibinfo {volume} {5}},\ \bibinfo {pages} {eaau8342} (\bibinfo {year} {2019})},\ \Eprint {https://arxiv.org/abs/1806.11123} {arXiv:1806.11123} \BibitemShut {NoStop}%
\bibitem [{\citenamefont {Granet}\ and\ \citenamefont {Dreyer}(2025)}]{granetDilutionErrorDigital2025}%
  \BibitemOpen
  \bibfield  {author} {\bibinfo {author} {\bibfnamefont {E.}~\bibnamefont {Granet}}\ and\ \bibinfo {author} {\bibfnamefont {H.}~\bibnamefont {Dreyer}},\ }\href {https://doi.org/10.1103/PRXQuantum.6.010333} {\bibfield  {journal} {\bibinfo  {journal} {PRX Quantum}\ }\textbf {\bibinfo {volume} {6}},\ \bibinfo {pages} {010333} (\bibinfo {year} {2025})},\ \Eprint {https://arxiv.org/abs/2409.04254} {arXiv:2409.04254} \BibitemShut {NoStop}%
\bibitem [{\citenamefont {Yu}\ \emph {et~al.}(2025)\citenamefont {Yu}, \citenamefont {Xu},\ and\ \citenamefont {Zhao}}]{yuObservableDrivenSpeedupsQuantum2025}%
  \BibitemOpen
  \bibfield  {author} {\bibinfo {author} {\bibfnamefont {W.}~\bibnamefont {Yu}}, \bibinfo {author} {\bibfnamefont {J.}~\bibnamefont {Xu}},\ and\ \bibinfo {author} {\bibfnamefont {Q.}~\bibnamefont {Zhao}},\ }\href {https://doi.org/10.1038/s42005-025-02260-5} {\bibfield  {journal} {\bibinfo  {journal} {Communications Physics}\ }\textbf {\bibinfo {volume} {8}},\ \bibinfo {pages} {340} (\bibinfo {year} {2025})},\ \Eprint {https://arxiv.org/abs/2407.14497} {arXiv:2407.14497} \BibitemShut {NoStop}%
\bibitem [{\citenamefont {Avtandilyan}\ and\ \citenamefont {Pogosov}(2024)}]{avtandilyanOptimalorderTrotterSuzuki2024}%
  \BibitemOpen
  \bibfield  {author} {\bibinfo {author} {\bibfnamefont {A.~A.}\ \bibnamefont {Avtandilyan}}\ and\ \bibinfo {author} {\bibfnamefont {W.~V.}\ \bibnamefont {Pogosov}},\ }\href {https://doi.org/10.1007/s11128-024-04627-z} {\bibfield  {journal} {\bibinfo  {journal} {Quantum Information Processing}\ }\textbf {\bibinfo {volume} {24}},\ \bibinfo {pages} {8} (\bibinfo {year} {2024})},\ \Eprint {https://arxiv.org/abs/2405.01131} {arXiv:2405.01131} \BibitemShut {NoStop}%
\bibitem [{\citenamefont {Bravyi}\ \emph {et~al.}(2021)\citenamefont {Bravyi}, \citenamefont {Gosset},\ and\ \citenamefont {Movassagh}}]{bravyiClassicalAlgorithmsQuantum2021}%
  \BibitemOpen
  \bibfield  {author} {\bibinfo {author} {\bibfnamefont {S.}~\bibnamefont {Bravyi}}, \bibinfo {author} {\bibfnamefont {D.}~\bibnamefont {Gosset}},\ and\ \bibinfo {author} {\bibfnamefont {R.}~\bibnamefont {Movassagh}},\ }\href {https://doi.org/10.1038/s41567-020-01109-8} {\bibfield  {journal} {\bibinfo  {journal} {Nature Physics}\ }\textbf {\bibinfo {volume} {17}},\ \bibinfo {pages} {337} (\bibinfo {year} {2021})},\ \Eprint {https://arxiv.org/abs/1909.11485} {arXiv:1909.11485} \BibitemShut {NoStop}%
\bibitem [{\citenamefont {Fefferman}\ \emph {et~al.}(2024)\citenamefont {Fefferman}, \citenamefont {Ghosh}, \citenamefont {Gullans}, \citenamefont {Kuroiwa},\ and\ \citenamefont {Sharma}}]{feffermanEffectNonunitalNoise2024}%
  \BibitemOpen
  \bibfield  {author} {\bibinfo {author} {\bibfnamefont {B.}~\bibnamefont {Fefferman}}, \bibinfo {author} {\bibfnamefont {S.}~\bibnamefont {Ghosh}}, \bibinfo {author} {\bibfnamefont {M.}~\bibnamefont {Gullans}}, \bibinfo {author} {\bibfnamefont {K.}~\bibnamefont {Kuroiwa}},\ and\ \bibinfo {author} {\bibfnamefont {K.}~\bibnamefont {Sharma}},\ }\href {https://doi.org/10.1103/PRXQuantum.5.030317} {\bibfield  {journal} {\bibinfo  {journal} {PRX Quantum}\ }\textbf {\bibinfo {volume} {5}},\ \bibinfo {pages} {030317} (\bibinfo {year} {2024})},\ \Eprint {https://arxiv.org/abs/2306.16659} {arXiv:2306.16659} \BibitemShut {NoStop}%
\bibitem [{\citenamefont {Arute}\ \emph {et~al.}(2019)\citenamefont {Arute}, \citenamefont {Arya}, \citenamefont {Babbush}, \citenamefont {Bacon}, \citenamefont {Bardin}, \citenamefont {Barends}, \citenamefont {Biswas}, \citenamefont {Boixo}, \citenamefont {Brandao}, \citenamefont {Buell}, \citenamefont {Burkett}, \citenamefont {Chen}, \citenamefont {Chen}, \citenamefont {Chiaro}, \citenamefont {Collins}, \citenamefont {Courtney}, \citenamefont {Dunsworth}, \citenamefont {Farhi}, \citenamefont {Foxen}, \citenamefont {Fowler}, \citenamefont {Gidney}, \citenamefont {Giustina}, \citenamefont {Graff}, \citenamefont {Guerin}, \citenamefont {Habegger}, \citenamefont {Harrigan}, \citenamefont {Hartmann}, \citenamefont {Ho}, \citenamefont {Hoffmann}, \citenamefont {Huang}, \citenamefont {Humble}, \citenamefont {Isakov}, \citenamefont {Jeffrey}, \citenamefont {Jiang}, \citenamefont {Kafri}, \citenamefont {Kechedzhi}, \citenamefont {Kelly}, \citenamefont {Klimov}, \citenamefont {Knysh}, \citenamefont {Korotkov}, \citenamefont {Kostritsa}, \citenamefont {Landhuis}, \citenamefont {Lindmark}, \citenamefont {Lucero}, \citenamefont {Lyakh}, \citenamefont {Mandra}, \citenamefont {McClean}, \citenamefont {McEwen}, \citenamefont {Megrant}, \citenamefont {Mi}, \citenamefont {Michielsen}, \citenamefont {Mohseni}, \citenamefont {Mutus}, \citenamefont {Naaman}, \citenamefont {Neeley}, \citenamefont {Neill}, \citenamefont {Niu}, \citenamefont {Ostby}, \citenamefont {Petukhov}, \citenamefont {Platt}, \citenamefont {Quintana}, \citenamefont {Rieffel}, \citenamefont {Roushan}, \citenamefont {Rubin}, \citenamefont {Sank}, \citenamefont {Satzinger}, \citenamefont {Smelyanskiy}, \citenamefont {Sung}, \citenamefont {Trevithick}, \citenamefont {Vainsencher}, \citenamefont {Villalonga}, \citenamefont {White}, \citenamefont {Yao}, \citenamefont {Yeh}, \citenamefont {Zalcman}, \citenamefont {Neven},\ and\ \citenamefont {Martinis}}]{aruteQuantumSupremacyUsing2019}%
  \BibitemOpen
  \bibfield  {author} {\bibinfo {author} {\bibfnamefont {F.}~\bibnamefont {Arute}}, \bibinfo {author} {\bibfnamefont {K.}~\bibnamefont {Arya}}, \bibinfo {author} {\bibfnamefont {R.}~\bibnamefont {Babbush}}, \emph {et~al.},\ }\href {https://doi.org/10.1038/s41586-019-1666-5} {\bibfield  {journal} {\bibinfo  {journal} {Nature}\ }\textbf {\bibinfo {volume} {574}},\ \bibinfo {pages} {505} (\bibinfo {year} {2019})},\ \Eprint {https://arxiv.org/abs/1910.11333} {arXiv:1910.11333} \BibitemShut {NoStop}%
\bibitem [{\citenamefont {Morvan}\ \emph {et~al.}(2024)\citenamefont {Morvan}, \citenamefont {Villalonga}, \citenamefont {Mi}, \citenamefont {Mandr{\`a}}, \citenamefont {Bengtsson}, \citenamefont {Klimov}, \citenamefont {Chen}, \citenamefont {Hong}, \citenamefont {Erickson}, \citenamefont {Drozdov}, \citenamefont {Chau}, \citenamefont {Laun}, \citenamefont {Movassagh}, \citenamefont {Asfaw}, \citenamefont {Brand{\~a}o}, \citenamefont {Peralta}, \citenamefont {Abanin}, \citenamefont {Acharya}, \citenamefont {Allen}, \citenamefont {Andersen}, \citenamefont {Anderson}, \citenamefont {Ansmann}, \citenamefont {Arute}, \citenamefont {Arya}, \citenamefont {Atalaya}, \citenamefont {Bardin}, \citenamefont {Bilmes}, \citenamefont {Bortoli}, \citenamefont {Bourassa}, \citenamefont {Bovaird}, \citenamefont {Brill}, \citenamefont {Broughton}, \citenamefont {Buckley}, \citenamefont {Buell}, \citenamefont {Burger}, \citenamefont {Burkett}, \citenamefont {Bushnell}, \citenamefont {Campero}, \citenamefont {Chang}, \citenamefont {Chiaro}, \citenamefont {Chik}, \citenamefont {Chou}, \citenamefont {Cogan}, \citenamefont {Collins}, \citenamefont {Conner}, \citenamefont {Courtney}, \citenamefont {Crook}, \citenamefont {Curtin}, \citenamefont {Debroy}, \citenamefont {Barba}, \citenamefont {Demura}, \citenamefont {Paolo}, \citenamefont {Dunsworth}, \citenamefont {Faoro}, \citenamefont {Farhi}, \citenamefont {Fatemi}, \citenamefont {Ferreira}, \citenamefont {Burgos}, \citenamefont {Forati}, \citenamefont {Fowler}, \citenamefont {Foxen}, \citenamefont {Garcia}, \citenamefont {Genois}, \citenamefont {Giang}, \citenamefont {Gidney}, \citenamefont {Gilboa}, \citenamefont {Giustina}, \citenamefont {Gosula}, \citenamefont {Dau}, \citenamefont {Gross}, \citenamefont {Habegger}, \citenamefont {Hamilton}, \citenamefont {Hansen}, \citenamefont {Harrigan}, \citenamefont {Harrington}, \citenamefont {Heu}, \citenamefont {Hoffmann}, \citenamefont {Huang}, \citenamefont {Huff}, \citenamefont {Huggins}, \citenamefont {Ioffe}, \citenamefont {Isakov}, \citenamefont {Iveland}, \citenamefont {Jeffrey}, \citenamefont {Jiang}, \citenamefont {Jones}, \citenamefont {Juhas}, \citenamefont {Kafri}, \citenamefont {Khattar}, \citenamefont {Khezri}, \citenamefont {Kieferov{\'a}}, \citenamefont {Kim}, \citenamefont {Kitaev}, \citenamefont {Klots}, \citenamefont {Korotkov}, \citenamefont {Kostritsa}, \citenamefont {Kreikebaum}, \citenamefont {Landhuis}, \citenamefont {Laptev}, \citenamefont {Lau}, \citenamefont {Laws}, \citenamefont {Lee}, \citenamefont {Lee}, \citenamefont {Lensky}, \citenamefont {Lester}, \citenamefont {Lill}, \citenamefont {Liu}, \citenamefont {Livingston}, \citenamefont {Locharla}, \citenamefont {Malone}, \citenamefont {Martin}, \citenamefont {Martin}, \citenamefont {McClean}, \citenamefont {McEwen}, \citenamefont {Miao}, \citenamefont {Mieszala}, \citenamefont {Montazeri}, \citenamefont {Mruczkiewicz}, \citenamefont {Naaman}, \citenamefont {Neeley}, \citenamefont {Neill}, \citenamefont {Nersisyan}, \citenamefont {Newman}, \citenamefont {Ng}, \citenamefont {Nguyen}, \citenamefont {Nguyen}, \citenamefont {Niu}, \citenamefont {O'Brien}, \citenamefont {Omonije}, \citenamefont {Opremcak}, \citenamefont {Petukhov}, \citenamefont {Potter}, \citenamefont {Pryadko}, \citenamefont {Quintana}, \citenamefont {Rhodes}, \citenamefont {Rocque}, \citenamefont {Rosenberg}, \citenamefont {Rubin}, \citenamefont {Saei}, \citenamefont {Sank}, \citenamefont {Sankaragomathi}, \citenamefont {Satzinger}, \citenamefont {Schurkus}, \citenamefont {Schuster}, \citenamefont {Shearn}, \citenamefont {Shorter}, \citenamefont {Shutty}, \citenamefont {Shvarts}, \citenamefont {Sivak}, \citenamefont {Skruzny}, \citenamefont {Smith}, \citenamefont {Somma}, \citenamefont {Sterling}, \citenamefont {Strain}, \citenamefont {Szalay}, \citenamefont {Thor}, \citenamefont {Torres}, \citenamefont {Vidal}, \citenamefont {Heidweiller}, \citenamefont {White}, \citenamefont {Woo}, \citenamefont {Xing}, \citenamefont {Yao}, \citenamefont {Yeh}, \citenamefont {Yoo}, \citenamefont {Young}, \citenamefont {Zalcman}, \citenamefont {Zhang}, \citenamefont {Zhu}, \citenamefont {Zobrist}, \citenamefont {Rieffel}, \citenamefont {Biswas}, \citenamefont {Babbush}, \citenamefont {Bacon}, \citenamefont {Hilton}, \citenamefont {Lucero}, \citenamefont {Neven}, \citenamefont {Megrant}, \citenamefont {Kelly}, \citenamefont {Roushan}, \citenamefont {Aleiner}, \citenamefont {Smelyanskiy}, \citenamefont {Kechedzhi}, \citenamefont {Chen},\ and\ \citenamefont {Boixo}}]{morvanPhaseTransitionsRandom2024}%
  \BibitemOpen
  \bibfield  {author} {\bibinfo {author} {\bibfnamefont {A.}~\bibnamefont {Morvan}}, \bibinfo {author} {\bibfnamefont {B.}~\bibnamefont {Villalonga}}, \bibinfo {author} {\bibfnamefont {X.}~\bibnamefont {Mi}}, \emph {et~al.},\ }\href {https://doi.org/10.1038/s41586-024-07998-6} {\bibfield  {journal} {\bibinfo  {journal} {Nature}\ }\textbf {\bibinfo {volume} {634}},\ \bibinfo {pages} {328} (\bibinfo {year} {2024})},\ \Eprint {https://arxiv.org/abs/2304.11119} {arXiv:2304.11119} \BibitemShut {NoStop}%
\bibitem [{\citenamefont {Acharya}\ \emph {et~al.}(2024)\citenamefont {Acharya}, \citenamefont {Abanin}, \citenamefont {{Aghababaie-Beni}}, \citenamefont {Aleiner}, \citenamefont {Andersen}, \citenamefont {Ansmann}, \citenamefont {Arute}, \citenamefont {Arya}, \citenamefont {Asfaw}, \citenamefont {Astrakhantsev}, \citenamefont {Atalaya}, \citenamefont {Babbush}, \citenamefont {Bacon}, \citenamefont {Ballard}, \citenamefont {Bardin}, \citenamefont {Bausch}, \citenamefont {Bengtsson}, \citenamefont {Bilmes}, \citenamefont {Blackwell}, \citenamefont {Boixo}, \citenamefont {Bortoli}, \citenamefont {Bourassa}, \citenamefont {Bovaird}, \citenamefont {Brill}, \citenamefont {Broughton}, \citenamefont {Browne}, \citenamefont {Buchea}, \citenamefont {Buckley}, \citenamefont {Buell}, \citenamefont {Burger}, \citenamefont {Burkett}, \citenamefont {Bushnell}, \citenamefont {Cabrera}, \citenamefont {Campero}, \citenamefont {Chang}, \citenamefont {Chen}, \citenamefont {Chen}, \citenamefont {Chiaro}, \citenamefont {Chik}, \citenamefont {Chou}, \citenamefont {Claes}, \citenamefont {Cleland}, \citenamefont {Cogan}, \citenamefont {Collins}, \citenamefont {Conner}, \citenamefont {Courtney}, \citenamefont {Crook}, \citenamefont {Curtin}, \citenamefont {Das}, \citenamefont {Davies}, \citenamefont {De~Lorenzo}, \citenamefont {Debroy}, \citenamefont {Demura}, \citenamefont {Devoret}, \citenamefont {Di~Paolo}, \citenamefont {Donohoe}, \citenamefont {Drozdov}, \citenamefont {Dunsworth}, \citenamefont {Earle}, \citenamefont {Edlich}, \citenamefont {Eickbusch}, \citenamefont {Elbag}, \citenamefont {Elzouka}, \citenamefont {Erickson}, \citenamefont {Faoro}, \citenamefont {Farhi}, \citenamefont {Ferreira}, \citenamefont {Burgos}, \citenamefont {Forati}, \citenamefont {Fowler}, \citenamefont {Foxen}, \citenamefont {Ganjam}, \citenamefont {Garcia}, \citenamefont {Gasca}, \citenamefont {Genois}, \citenamefont {Giang}, \citenamefont {Gidney}, \citenamefont {Gilboa}, \citenamefont {Gosula}, \citenamefont {Dau}, \citenamefont {Graumann}, \citenamefont {Greene}, \citenamefont {Gross}, \citenamefont {Habegger}, \citenamefont {Hall}, \citenamefont {Hamilton}, \citenamefont {Hansen}, \citenamefont {Harrigan}, \citenamefont {Harrington}, \citenamefont {Heras}, \citenamefont {Heslin}, \citenamefont {Heu}, \citenamefont {Higgott}, \citenamefont {Hill}, \citenamefont {Hilton}, \citenamefont {Holland}, \citenamefont {Hong}, \citenamefont {Huang}, \citenamefont {Huff}, \citenamefont {Huggins}, \citenamefont {Ioffe}, \citenamefont {Isakov}, \citenamefont {Iveland}, \citenamefont {Jeffrey}, \citenamefont {Jiang}, \citenamefont {Jones}, \citenamefont {Jordan}, \citenamefont {Joshi}, \citenamefont {Juhas}, \citenamefont {Kafri}, \citenamefont {Kang}, \citenamefont {Karamlou}, \citenamefont {Kechedzhi}, \citenamefont {Kelly}, \citenamefont {Khaire}, \citenamefont {Khattar}, \citenamefont {Khezri}, \citenamefont {Kim}, \citenamefont {Klimov}, \citenamefont {Klots}, \citenamefont {Kobrin}, \citenamefont {Kohli}, \citenamefont {Korotkov}, \citenamefont {Kostritsa}, \citenamefont {Kothari}, \citenamefont {Kozlovskii}, \citenamefont {Kreikebaum}, \citenamefont {Kurilovich}, \citenamefont {Lacroix}, \citenamefont {Landhuis}, \citenamefont {{Lange-Dei}}, \citenamefont {Langley}, \citenamefont {Laptev}, \citenamefont {Lau}, \citenamefont {Le~Guevel}, \citenamefont {Ledford}, \citenamefont {Lee}, \citenamefont {Lee}, \citenamefont {Lensky}, \citenamefont {Leon}, \citenamefont {Lester}, \citenamefont {Li}, \citenamefont {Li}, \citenamefont {Lill}, \citenamefont {Liu}, \citenamefont {Livingston}, \citenamefont {Locharla}, \citenamefont {Lucero}, \citenamefont {Lundahl}, \citenamefont {Lunt}, \citenamefont {Madhuk}, \citenamefont {Malone}, \citenamefont {Maloney}, \citenamefont {Mandr{\`a}}, \citenamefont {Manyika}, \citenamefont {Martin}, \citenamefont {Martin}, \citenamefont {Martin}, \citenamefont {Maxfield}, \citenamefont {McClean}, \citenamefont {McEwen}, \citenamefont {Meeks}, \citenamefont {Megrant}, \citenamefont {Mi}, \citenamefont {Miao}, \citenamefont {Mieszala}, \citenamefont {Molavi}, \citenamefont {Molina}, \citenamefont {Montazeri}, \citenamefont {Morvan}, \citenamefont {Movassagh}, \citenamefont {Mruczkiewicz}, \citenamefont {Naaman}, \citenamefont {Neeley}, \citenamefont {Neill}, \citenamefont {Nersisyan}, \citenamefont {Neven}, \citenamefont {Newman}, \citenamefont {Ng}, \citenamefont {Nguyen}, \citenamefont {Nguyen}, \citenamefont {Ni}, \citenamefont {Niu}, \citenamefont {O'Brien}, \citenamefont {Oliver}, \citenamefont {Opremcak}, \citenamefont {Ottosson}, \citenamefont {Petukhov}, \citenamefont {Pizzuto}, \citenamefont {Platt}, \citenamefont {Potter}, \citenamefont {Pritchard}, \citenamefont {Pryadko}, \citenamefont {Quintana}, \citenamefont {Ramachandran}, \citenamefont {Reagor}, \citenamefont {Redding}, \citenamefont {Rhodes}, \citenamefont {Roberts}, \citenamefont {Rosenberg}, \citenamefont {Rosenfeld}, \citenamefont {Roushan}, \citenamefont {Rubin}, \citenamefont {Saei}, \citenamefont {Sank}, \citenamefont {Sankaragomathi}, \citenamefont {Satzinger}, \citenamefont {Schurkus}, \citenamefont {Schuster}, \citenamefont {Senior}, \citenamefont {Shearn}, \citenamefont {Shorter}, \citenamefont {Shutty}, \citenamefont {Shvarts}, \citenamefont {Singh}, \citenamefont {Sivak}, \citenamefont {Skruzny}, \citenamefont {Small}, \citenamefont {Smelyanskiy}, \citenamefont {Smith}, \citenamefont {Somma}, \citenamefont {Springer}, \citenamefont {Sterling}, \citenamefont {Strain}, \citenamefont {Suchard}, \citenamefont {Szasz}, \citenamefont {Sztein}, \citenamefont {Thor}, \citenamefont {Torres}, \citenamefont {Torunbalci}, \citenamefont {Vaishnav}, \citenamefont {Vargas}, \citenamefont {Vdovichev}, \citenamefont {Vidal}, \citenamefont {Villalonga}, \citenamefont {Heidweiller}, \citenamefont {Waltman}, \citenamefont {Wang}, \citenamefont {Ware}, \citenamefont {Weber}, \citenamefont {Weidel}, \citenamefont {White}, \citenamefont {Wong}, \citenamefont {Woo}, \citenamefont {Xing}, \citenamefont {Yao}, \citenamefont {Yeh}, \citenamefont {Ying}, \citenamefont {Yoo}, \citenamefont {Yosri}, \citenamefont {Young}, \citenamefont {Zalcman}, \citenamefont {Zhang}, \citenamefont {Zhu}, \citenamefont {Zobrist},\ and\ \citenamefont {{Google Quantum AI and Collaborators}}}]{acharyaQuantumErrorCorrection2024}%
  \BibitemOpen
  \bibfield  {author} {\bibinfo {author} {\bibfnamefont {R.}~\bibnamefont {Acharya}}, \bibinfo {author} {\bibfnamefont {D.~A.}\ \bibnamefont {Abanin}}, \bibinfo {author} {\bibfnamefont {L.}~\bibnamefont {{Aghababaie-Beni}}}, \emph {et~al.},\ }\href {https://doi.org/10.1038/s41586-024-08449-y} {\bibfield  {journal} {\bibinfo  {journal} {Nature}\ ,\ \bibinfo {pages} {1}} (\bibinfo {year} {2024})},\ \Eprint {https://arxiv.org/abs/2408.13687} {arXiv:2408.13687} \BibitemShut {NoStop}%
\bibitem [{\citenamefont {Bravyi}\ and\ \citenamefont {Kitaev}(2005)}]{bravyiUniversalQuantumComputation2005}%
  \BibitemOpen
  \bibfield  {author} {\bibinfo {author} {\bibfnamefont {S.}~\bibnamefont {Bravyi}}\ and\ \bibinfo {author} {\bibfnamefont {A.}~\bibnamefont {Kitaev}},\ }\href {https://doi.org/10.1103/PhysRevA.71.022316} {\bibfield  {journal} {\bibinfo  {journal} {Physical Review A}\ }\textbf {\bibinfo {volume} {71}},\ \bibinfo {pages} {022316} (\bibinfo {year} {2005})},\ \Eprint {https://arxiv.org/abs/quant-ph/0403025} {arXiv:quant-ph/0403025} \BibitemShut {NoStop}%
\bibitem [{\citenamefont {Bravyi}\ and\ \citenamefont {Haah}(2012)}]{bravyiMagicstateDistillationLow2012}%
  \BibitemOpen
  \bibfield  {author} {\bibinfo {author} {\bibfnamefont {S.}~\bibnamefont {Bravyi}}\ and\ \bibinfo {author} {\bibfnamefont {J.}~\bibnamefont {Haah}},\ }\href {https://doi.org/10.1103/PhysRevA.86.052329} {\bibfield  {journal} {\bibinfo  {journal} {Physical Review A}\ }\textbf {\bibinfo {volume} {86}},\ \bibinfo {pages} {052329} (\bibinfo {year} {2012})}\BibitemShut {NoStop}%
\bibitem [{\citenamefont {Wills}\ \emph {et~al.}(2025)\citenamefont {Wills}, \citenamefont {Hsieh},\ and\ \citenamefont {Yamasaki}}]{willsConstantoverheadMagicState2025}%
  \BibitemOpen
  \bibfield  {author} {\bibinfo {author} {\bibfnamefont {A.}~\bibnamefont {Wills}}, \bibinfo {author} {\bibfnamefont {M.-H.}\ \bibnamefont {Hsieh}},\ and\ \bibinfo {author} {\bibfnamefont {H.}~\bibnamefont {Yamasaki}},\ }\href {https://doi.org/10.1038/s41567-025-03026-0} {\bibfield  {journal} {\bibinfo  {journal} {Nature Physics}\ ,\ \bibinfo {pages} {1}} (\bibinfo {year} {2025})},\ \Eprint {https://arxiv.org/abs/2408.07764} {arXiv:2408.07764} \BibitemShut {NoStop}%
\bibitem [{\citenamefont {Fowler}\ \emph {et~al.}(2012)\citenamefont {Fowler}, \citenamefont {Mariantoni}, \citenamefont {Martinis},\ and\ \citenamefont {Cleland}}]{fowlerSurfaceCodesPractical2012}%
  \BibitemOpen
  \bibfield  {author} {\bibinfo {author} {\bibfnamefont {A.~G.}\ \bibnamefont {Fowler}}, \bibinfo {author} {\bibfnamefont {M.}~\bibnamefont {Mariantoni}}, \bibinfo {author} {\bibfnamefont {J.~M.}\ \bibnamefont {Martinis}},\ and\ \bibinfo {author} {\bibfnamefont {A.~N.}\ \bibnamefont {Cleland}},\ }\href {https://doi.org/10.1103/PhysRevA.86.032324} {\bibfield  {journal} {\bibinfo  {journal} {Physical Review A}\ }\textbf {\bibinfo {volume} {86}},\ \bibinfo {pages} {032324} (\bibinfo {year} {2012})},\ \Eprint {https://arxiv.org/abs/1208.0928} {arXiv:1208.0928} \BibitemShut {NoStop}%
\bibitem [{\citenamefont {Heim}\ \emph {et~al.}(2016)\citenamefont {Heim}, \citenamefont {Svore},\ and\ \citenamefont {Hastings}}]{heimOptimalCircuitLevelDecoding2016}%
  \BibitemOpen
  \bibfield  {author} {\bibinfo {author} {\bibfnamefont {B.}~\bibnamefont {Heim}}, \bibinfo {author} {\bibfnamefont {K.~M.}\ \bibnamefont {Svore}},\ and\ \bibinfo {author} {\bibfnamefont {M.~B.}\ \bibnamefont {Hastings}},\ }\href {https://doi.org/10.48550/arXiv.1609.06373} {\bibinfo {title} {Optimal {{Circuit-Level Decoding}} for {{Surface Codes}}}} (\bibinfo {year} {2016}),\ \Eprint {https://arxiv.org/abs/1609.06373} {arXiv:1609.06373} \BibitemShut {NoStop}%
\bibitem [{\citenamefont {Shtanko}\ and\ \citenamefont {Sharma}(2025)}]{shtankoComplexityLocalQuantum2025}%
  \BibitemOpen
  \bibfield  {author} {\bibinfo {author} {\bibfnamefont {O.}~\bibnamefont {Shtanko}}\ and\ \bibinfo {author} {\bibfnamefont {K.}~\bibnamefont {Sharma}},\ }\href {https://doi.org/10.1103/xwz2-wfr4} {\bibfield  {journal} {\bibinfo  {journal} {PRX Quantum}\ }\textbf {\bibinfo {volume} {6}},\ \bibinfo {pages} {030347} (\bibinfo {year} {2025})}\BibitemShut {NoStop}%
\bibitem [{\citenamefont {Stenger}\ \emph {et~al.}(2021)\citenamefont {Stenger}, \citenamefont {Bronn}, \citenamefont {Egger},\ and\ \citenamefont {Pekker}}]{stengerSimulatingDynamicsBraiding2021}%
  \BibitemOpen
  \bibfield  {author} {\bibinfo {author} {\bibfnamefont {J.~P.~T.}\ \bibnamefont {Stenger}}, \bibinfo {author} {\bibfnamefont {N.~T.}\ \bibnamefont {Bronn}}, \bibinfo {author} {\bibfnamefont {D.~J.}\ \bibnamefont {Egger}},\ and\ \bibinfo {author} {\bibfnamefont {D.}~\bibnamefont {Pekker}},\ }\href {https://doi.org/10.1103/PhysRevResearch.3.033171} {\bibfield  {journal} {\bibinfo  {journal} {Physical Review Research}\ }\textbf {\bibinfo {volume} {3}},\ \bibinfo {pages} {033171} (\bibinfo {year} {2021})},\ \Eprint {https://arxiv.org/abs/2012.11660} {arXiv:2012.11660} \BibitemShut {NoStop}%
\bibitem [{\citenamefont {Earnest}\ \emph {et~al.}(2021)\citenamefont {Earnest}, \citenamefont {Tornow},\ and\ \citenamefont {Egger}}]{earnestPulseefficientCircuitTranspilation2021}%
  \BibitemOpen
  \bibfield  {author} {\bibinfo {author} {\bibfnamefont {N.}~\bibnamefont {Earnest}}, \bibinfo {author} {\bibfnamefont {C.}~\bibnamefont {Tornow}},\ and\ \bibinfo {author} {\bibfnamefont {D.~J.}\ \bibnamefont {Egger}},\ }\href {https://doi.org/10.1103/PhysRevResearch.3.043088} {\bibfield  {journal} {\bibinfo  {journal} {Physical Review Research}\ }\textbf {\bibinfo {volume} {3}},\ \bibinfo {pages} {043088} (\bibinfo {year} {2021})},\ \Eprint {https://arxiv.org/abs/2105.01063} {arXiv:2105.01063} \BibitemShut {NoStop}%
\bibitem [{\citenamefont {Sieberer}\ \emph {et~al.}(2019)\citenamefont {Sieberer}, \citenamefont {Olsacher}, \citenamefont {Elben}, \citenamefont {Heyl}, \citenamefont {Hauke}, \citenamefont {Haake},\ and\ \citenamefont {Zoller}}]{siebererDigitalQuantumSimulation2019}%
  \BibitemOpen
  \bibfield  {author} {\bibinfo {author} {\bibfnamefont {L.~M.}\ \bibnamefont {Sieberer}}, \bibinfo {author} {\bibfnamefont {T.}~\bibnamefont {Olsacher}}, \bibinfo {author} {\bibfnamefont {A.}~\bibnamefont {Elben}}, \bibinfo {author} {\bibfnamefont {M.}~\bibnamefont {Heyl}}, \bibinfo {author} {\bibfnamefont {P.}~\bibnamefont {Hauke}}, \bibinfo {author} {\bibfnamefont {F.}~\bibnamefont {Haake}},\ and\ \bibinfo {author} {\bibfnamefont {P.}~\bibnamefont {Zoller}},\ }\href {https://doi.org/10.1038/s41534-019-0192-5} {\bibfield  {journal} {\bibinfo  {journal} {npj Quantum Information}\ }\textbf {\bibinfo {volume} {5}},\ \bibinfo {pages} {78} (\bibinfo {year} {2019})},\ \Eprint {https://arxiv.org/abs/1812.05876} {arXiv:1812.05876} \BibitemShut {NoStop}%
\bibitem [{\citenamefont {Begu{\v s}i{\'c}}\ and\ \citenamefont {Chan}(2025)}]{begusicRealTimeOperatorEvolution2025}%
  \BibitemOpen
  \bibfield  {author} {\bibinfo {author} {\bibfnamefont {T.}~\bibnamefont {Begu{\v s}i{\'c}}}\ and\ \bibinfo {author} {\bibfnamefont {G.~K.-L.}\ \bibnamefont {Chan}},\ }\href {https://doi.org/10.1103/PRXQuantum.6.020302} {\bibfield  {journal} {\bibinfo  {journal} {PRX Quantum}\ }\textbf {\bibinfo {volume} {6}},\ \bibinfo {pages} {020302} (\bibinfo {year} {2025})},\ \Eprint {https://arxiv.org/abs/2409.03097} {arXiv:2409.03097} \BibitemShut {NoStop}%
\bibitem [{\citenamefont {Watson}\ and\ \citenamefont {Watkins}(2025)}]{watsonExponentiallyReducedCircuit2025}%
  \BibitemOpen
  \bibfield  {author} {\bibinfo {author} {\bibfnamefont {J.~D.}\ \bibnamefont {Watson}}\ and\ \bibinfo {author} {\bibfnamefont {J.}~\bibnamefont {Watkins}},\ }\href {https://doi.org/10.1103/kw39-yxq5} {\bibfield  {journal} {\bibinfo  {journal} {PRX Quantum}\ }\textbf {\bibinfo {volume} {6}},\ \bibinfo {pages} {030325} (\bibinfo {year} {2025})},\ \Eprint {https://arxiv.org/abs/2408.14385} {arXiv:2408.14385} \BibitemShut {NoStop}%
\bibitem [{\citenamefont {Zhou}\ \emph {et~al.}(2024)\citenamefont {Zhou}, \citenamefont {Zhao}, \citenamefont {Cain}, \citenamefont {Bluvstein}, \citenamefont {Duckering}, \citenamefont {Hu}, \citenamefont {Wang}, \citenamefont {Kubica},\ and\ \citenamefont {Lukin}}]{zhouAlgorithmicFaultTolerance2024}%
  \BibitemOpen
  \bibfield  {author} {\bibinfo {author} {\bibfnamefont {H.}~\bibnamefont {Zhou}}, \bibinfo {author} {\bibfnamefont {C.}~\bibnamefont {Zhao}}, \bibinfo {author} {\bibfnamefont {M.}~\bibnamefont {Cain}}, \bibinfo {author} {\bibfnamefont {D.}~\bibnamefont {Bluvstein}}, \bibinfo {author} {\bibfnamefont {C.}~\bibnamefont {Duckering}}, \bibinfo {author} {\bibfnamefont {H.-Y.}\ \bibnamefont {Hu}}, \bibinfo {author} {\bibfnamefont {S.-T.}\ \bibnamefont {Wang}}, \bibinfo {author} {\bibfnamefont {A.}~\bibnamefont {Kubica}},\ and\ \bibinfo {author} {\bibfnamefont {M.~D.}\ \bibnamefont {Lukin}},\ }\href {http://arxiv.org/abs/2406.17653} {\bibinfo {title} {Algorithmic {{Fault Tolerance}} for {{Fast Quantum Computing}}}} (\bibinfo {year} {2024}),\ \Eprint {https://arxiv.org/abs/2406.17653} {arXiv:2406.17653} \BibitemShut {NoStop}%
\bibitem [{\citenamefont {Akahoshi}\ \emph {et~al.}(2024)\citenamefont {Akahoshi}, \citenamefont {Maruyama}, \citenamefont {Oshima}, \citenamefont {Sato},\ and\ \citenamefont {Fujii}}]{akahoshiPartiallyFaultTolerantQuantum2024}%
  \BibitemOpen
  \bibfield  {author} {\bibinfo {author} {\bibfnamefont {Y.}~\bibnamefont {Akahoshi}}, \bibinfo {author} {\bibfnamefont {K.}~\bibnamefont {Maruyama}}, \bibinfo {author} {\bibfnamefont {H.}~\bibnamefont {Oshima}}, \bibinfo {author} {\bibfnamefont {S.}~\bibnamefont {Sato}},\ and\ \bibinfo {author} {\bibfnamefont {K.}~\bibnamefont {Fujii}},\ }\href {https://doi.org/10.1103/PRXQuantum.5.010337} {\bibfield  {journal} {\bibinfo  {journal} {PRX Quantum}\ }\textbf {\bibinfo {volume} {5}},\ \bibinfo {pages} {010337} (\bibinfo {year} {2024})},\ \Eprint {https://arxiv.org/abs/2303.13181} {arXiv:2303.13181} \BibitemShut {NoStop}%
\bibitem [{\citenamefont {Toshio}\ \emph {et~al.}(2025)\citenamefont {Toshio}, \citenamefont {Akahoshi}, \citenamefont {Fujisaki}, \citenamefont {Oshima}, \citenamefont {Sato},\ and\ \citenamefont {Fujii}}]{toshioPracticalQuantumAdvantage2025}%
  \BibitemOpen
  \bibfield  {author} {\bibinfo {author} {\bibfnamefont {R.}~\bibnamefont {Toshio}}, \bibinfo {author} {\bibfnamefont {Y.}~\bibnamefont {Akahoshi}}, \bibinfo {author} {\bibfnamefont {J.}~\bibnamefont {Fujisaki}}, \bibinfo {author} {\bibfnamefont {H.}~\bibnamefont {Oshima}}, \bibinfo {author} {\bibfnamefont {S.}~\bibnamefont {Sato}},\ and\ \bibinfo {author} {\bibfnamefont {K.}~\bibnamefont {Fujii}},\ }\href {https://doi.org/10.1103/PhysRevX.15.021057} {\bibfield  {journal} {\bibinfo  {journal} {Physical Review X}\ }\textbf {\bibinfo {volume} {15}},\ \bibinfo {pages} {021057} (\bibinfo {year} {2025})}\BibitemShut {NoStop}%
\bibitem [{\citenamefont {Magesan}\ \emph {et~al.}(2012)\citenamefont {Magesan}, \citenamefont {Gambetta},\ and\ \citenamefont {Emerson}}]{magesanCharacterizingQuantumGates2012}%
  \BibitemOpen
  \bibfield  {author} {\bibinfo {author} {\bibfnamefont {E.}~\bibnamefont {Magesan}}, \bibinfo {author} {\bibfnamefont {J.~M.}\ \bibnamefont {Gambetta}},\ and\ \bibinfo {author} {\bibfnamefont {J.}~\bibnamefont {Emerson}},\ }\href {https://doi.org/10.1103/PhysRevA.85.042311} {\bibfield  {journal} {\bibinfo  {journal} {Physical Review A}\ }\textbf {\bibinfo {volume} {85}},\ \bibinfo {pages} {042311} (\bibinfo {year} {2012})},\ \Eprint {https://arxiv.org/abs/1109.6887} {arXiv:1109.6887} \BibitemShut {NoStop}%
\bibitem [{\citenamefont {Yan}\ \emph {et~al.}(2023)\citenamefont {Yan}, \citenamefont {Du}, \citenamefont {Chen},\ and\ \citenamefont {Ma}}]{yanLimitationsNoisyQuantum2023}%
  \BibitemOpen
  \bibfield  {author} {\bibinfo {author} {\bibfnamefont {Y.}~\bibnamefont {Yan}}, \bibinfo {author} {\bibfnamefont {Z.}~\bibnamefont {Du}}, \bibinfo {author} {\bibfnamefont {J.}~\bibnamefont {Chen}},\ and\ \bibinfo {author} {\bibfnamefont {X.}~\bibnamefont {Ma}},\ }\href {http://arxiv.org/abs/2306.02836} {\bibinfo {title} {Limitations of {{Noisy Quantum Devices}} in {{Computational}} and {{Entangling Power}}}} (\bibinfo {year} {2023}),\ \Eprint {https://arxiv.org/abs/2306.02836} {arXiv:2306.02836} \BibitemShut {NoStop}%
\bibitem [{\citenamefont {{Javadi-Abhari}}\ \emph {et~al.}(2024)\citenamefont {{Javadi-Abhari}}, \citenamefont {Treinish}, \citenamefont {Krsulich}, \citenamefont {Wood}, \citenamefont {Lishman}, \citenamefont {Gacon}, \citenamefont {Martiel}, \citenamefont {Nation}, \citenamefont {Bishop}, \citenamefont {Cross}, \citenamefont {Johnson},\ and\ \citenamefont {Gambetta}}]{javadi-abhariQuantumComputingQiskit2024}%
  \BibitemOpen
  \bibfield  {author} {\bibinfo {author} {\bibfnamefont {A.}~\bibnamefont {{Javadi-Abhari}}}, \bibinfo {author} {\bibfnamefont {M.}~\bibnamefont {Treinish}}, \bibinfo {author} {\bibfnamefont {K.}~\bibnamefont {Krsulich}}, \emph {et~al.},\ }\href {https://doi.org/10.48550/arXiv.2405.08810} {\bibinfo {title} {Quantum computing with {{Qiskit}}}} (\bibinfo {year} {2024}),\ \Eprint {https://arxiv.org/abs/2405.08810} {arXiv:2405.08810} \BibitemShut {NoStop}%
\bibitem [{\citenamefont {McClean}\ \emph {et~al.}(2020)\citenamefont {McClean}, \citenamefont {Rubin}, \citenamefont {Sung}, \citenamefont {Kivlichan}, \citenamefont {{Bonet-Monroig}}, \citenamefont {Cao}, \citenamefont {Dai}, \citenamefont {Fried}, \citenamefont {Gidney}, \citenamefont {Gimby}, \citenamefont {Gokhale}, \citenamefont {H{\"a}ner}, \citenamefont {Hardikar}, \citenamefont {Havl{\'i}{\v c}ek}, \citenamefont {Higgott}, \citenamefont {Huang}, \citenamefont {Izaac}, \citenamefont {Jiang}, \citenamefont {Liu}, \citenamefont {McArdle}, \citenamefont {Neeley}, \citenamefont {O'Brien}, \citenamefont {O'Gorman}, \citenamefont {Ozfidan}, \citenamefont {Radin}, \citenamefont {Romero}, \citenamefont {Sawaya}, \citenamefont {Senjean}, \citenamefont {Setia}, \citenamefont {Sim}, \citenamefont {Steiger}, \citenamefont {Steudtner}, \citenamefont {Sun}, \citenamefont {Sun}, \citenamefont {Wang}, \citenamefont {Zhang},\ and\ \citenamefont {Babbush}}]{mccleanOpenFermionElectronicStructure2020}%
  \BibitemOpen
  \bibfield  {author} {\bibinfo {author} {\bibfnamefont {J.~R.}\ \bibnamefont {McClean}}, \bibinfo {author} {\bibfnamefont {N.~C.}\ \bibnamefont {Rubin}}, \bibinfo {author} {\bibfnamefont {K.~J.}\ \bibnamefont {Sung}}, \emph {et~al.},\ }\href {https://doi.org/10.1088/2058-9565/ab8ebc} {\bibfield  {journal} {\bibinfo  {journal} {Quantum Science and Technology}\ }\textbf {\bibinfo {volume} {5}},\ \bibinfo {pages} {034014} (\bibinfo {year} {2020})},\ \Eprint {https://arxiv.org/abs/1710.07629} {arXiv:1710.07629} \BibitemShut {NoStop}%
\bibitem [{\citenamefont {Sun}\ \emph {et~al.}(2020)\citenamefont {Sun}, \citenamefont {Zhang}, \citenamefont {Banerjee}, \citenamefont {Bao}, \citenamefont {Barbry}, \citenamefont {Blunt}, \citenamefont {Bogdanov}, \citenamefont {Booth}, \citenamefont {Chen}, \citenamefont {Cui}, \citenamefont {Eriksen}, \citenamefont {Gao}, \citenamefont {Guo}, \citenamefont {Hermann}, \citenamefont {Hermes}, \citenamefont {Koh}, \citenamefont {Koval}, \citenamefont {Lehtola}, \citenamefont {Li}, \citenamefont {Liu}, \citenamefont {Mardirossian}, \citenamefont {McClain}, \citenamefont {Motta}, \citenamefont {Mussard}, \citenamefont {Pham}, \citenamefont {Pulkin}, \citenamefont {Purwanto}, \citenamefont {Robinson}, \citenamefont {Ronca}, \citenamefont {Sayfutyarova}, \citenamefont {Scheurer}, \citenamefont {Schurkus}, \citenamefont {Smith}, \citenamefont {Sun}, \citenamefont {Sun}, \citenamefont {Upadhyay}, \citenamefont {Wagner}, \citenamefont {Wang}, \citenamefont {White}, \citenamefont {Whitfield}, \citenamefont {Williamson}, \citenamefont {Wouters}, \citenamefont {Yang}, \citenamefont {Yu}, \citenamefont {Zhu}, \citenamefont {Berkelbach}, \citenamefont {Sharma}, \citenamefont {Sokolov},\ and\ \citenamefont {Chan}}]{sunRecentDevelopmentsPySCF2020}%
  \BibitemOpen
  \bibfield  {author} {\bibinfo {author} {\bibfnamefont {Q.}~\bibnamefont {Sun}}, \bibinfo {author} {\bibfnamefont {X.}~\bibnamefont {Zhang}}, \bibinfo {author} {\bibfnamefont {S.}~\bibnamefont {Banerjee}}, \emph {et~al.},\ }\href {https://doi.org/10.1063/5.0006074} {\bibfield  {journal} {\bibinfo  {journal} {The Journal of Chemical Physics}\ }\textbf {\bibinfo {volume} {153}},\ \bibinfo {pages} {024109} (\bibinfo {year} {2020})}\BibitemShut {NoStop}%
\bibitem [{\citenamefont {Hubbard}(1963)}]{hubbardElectronCorrelationsNarrow1963}%
  \BibitemOpen
  \bibfield  {author} {\bibinfo {author} {\bibfnamefont {J.}~\bibnamefont {Hubbard}},\ }\bibfield  {journal} {\bibinfo  {journal} {Proceedings of the Royal Society of London. Series A. Mathematical and Physical Sciences}\ }\href {https://doi.org/10.1098/rspa.1963.0204} {10.1098/rspa.1963.0204} (\bibinfo {year} {1963})\BibitemShut {NoStop}%
\bibitem [{\citenamefont {Schubert}\ and\ \citenamefont {Mendl}(2023)}]{schubertTrotterErrorCommutator2023}%
  \BibitemOpen
  \bibfield  {author} {\bibinfo {author} {\bibfnamefont {A.}~\bibnamefont {Schubert}}\ and\ \bibinfo {author} {\bibfnamefont {C.~B.}\ \bibnamefont {Mendl}},\ }\href {https://doi.org/10.1103/PhysRevB.108.195105} {\bibfield  {journal} {\bibinfo  {journal} {Physical Review B}\ }\textbf {\bibinfo {volume} {108}},\ \bibinfo {pages} {195105} (\bibinfo {year} {2023})},\ \Eprint {https://arxiv.org/abs/2306.10603} {arXiv:2306.10603} \BibitemShut {NoStop}%
\bibitem [{\citenamefont {Nielsen}(2005)}]{nielsenFermionicCanonicalCommutation2005}%
  \BibitemOpen
  \bibfield  {author} {\bibinfo {author} {\bibfnamefont {M.~A.}\ \bibnamefont {Nielsen}},\ }\href {https://futureofmatter.com/assets/fermions_and_jordan_wigner.pdf} {\bibinfo {title} {The {{Fermionic}} canonical commutation relations and the {{Jordan-Wigner}} transform}} (\bibinfo {year} {2005})\BibitemShut {NoStop}%
\end{thebibliography}%
\onecolumngrid

\ifnum\onlymaintext=0
    \appendix
    \newpage
    \begin{center}
    {\bf \large Supplementary Materials to \\ \emph \mytitle} 
    \end{center}
    \renewcommand{\tocname}{Appendix Contents}

\section{Proofs of theorems}\label{apd:proofs}
\subsection{Preliminaries}
In this paper, different norms are used to quantify the errors in different cases. 
In general, these norms can be defined from different perspectives.
Here, we adopt the notion of Schatten norms.
\begin{definition}[Schatten norm]\label{def:norm}
    The Schatten $p$-norm of a matrix $A$ is defined as 
    $\norm{A}_p=[\Tr(\abs{A}^p)]^{1/p}$.
    Three commonly used Schatten norms are:
    \begin{itemize}
        \item The Schatten 1-norm (also called trace norm or nuclear norm) $\norm{A}_{\tr}\equiv\norm{A}_1:=\Tr(\sqrt{AA^\dagger})=\sum_j\abs{\lambda}_j$ 
        that is the sum of the absolute value of eigenvalues of $A$.
        It is commonly used to measure the distance between two density matrices, e.g. $A:=\rho-\tilde{\rho}$.
        
        \item The Schatten $\infty$-norm (also known as spectral norm or operator norm, induced-2 norm) $\norm{A}_\infty$ of an operator $A$ is the largest singular value of $A$,
        i.e., $\norm{A}_\infty\equiv\norm{A}_{\mathrm{op}}:=\max_{\ket{\psi}} \norm{A\ket{\psi}}_2$ where $\norm{\cdot}_2$ is the $l_2$-norm for vectors. 
        So, it corresponds to the operator error induced by the worst input (state).
        
        \item
        The Schatten 2-norm (also known as Frobenius or Hilbert-Schmidt norm) is $\fnorm{A}\equiv\norm{A}_2:=\sqrt{\Tr(AA^\dagger)}$ 
        which is the (root mean square) eigenvalue.
        Intuitively, it captures the average-case analysis.

    \end{itemize}
\end{definition}

\begin{definition}[Diamond norm]\label{def:diamond_norm}
    For a quantum channel $\E:\mathbb{C}^n\to\mathbb{C}^n$, diamond norm is the trace norm of the output of a trivial extension of $\E$, 
    maximized over all possible inputs with trace norm at most one
    \begin{equation}
        \dnorm{\E}:=\sup_{A}\qty{\trnorm{(\E \otimes \mathbb{I}_n)(A)}: \trnorm{A}\le 1}
        \label{eq:def_dnorm}
    \end{equation}
    where $A\in L(\mathbb{C}^n\otimes \mathbb{C}^n)$ .
\end{definition}

\subsection{Upper bounds of physical error}\label{apd:physical}
\subsubsection{The worst-case (state-independent) analysis}
We first derive the worst-case upper bound of the physical error by the diamond norm of channels.
\begin{lemma}[Diamond distance between Pauli channels \cite{magesanCharacterizingQuantumGates2012}]
    Given two $n$-qubit Pauli channels $\E_{\vb{q}}(\rho)=\sum_{i=0}^{4^n-1} q_i P_i \rho P_i^\dagger$ and $\E_{\vb{r}}(\rho)=\sum_{i=0}^{4^n-1} r_i P_i \rho P_i^\dagger$, 
    then their diamond norm distance is
    \begin{equation}
        \dnorm{\E_{\vb{q}}-\E_{\vb{r}}} = \norm{\vb{q}-\vb{r}}_1 = \sum_{i=0}^{4^n-1} \abs{p_i-r_i}.
    \end{equation}
\end{lemma}
Consider the one-qubit depolarizing channel $\E^{\depo}_{\gamma}$ with noise rate $\gamma \in [0,1]$ after applying a gate
\begin{align}\label{apd:eq:single_depolarizing}
    \E^{\depo}_{\gamma}(\rho) 
    &:= (1-\gamma)\rho + \frac{\gamma}{3} \qty(X\rho X + Y \rho Y + Z\rho Z) ,
\end{align}
where $\rho$ is the one-qubit density matrix.
\begin{corollary}[Diamond distance of 1-qubit depolarizing channel]\label{apd:lem:diamond_distance_single_depolar}
    The diamond distance between a one-qubit depolarizing channel \cref{apd:eq:single_depolarizing} with noise rate $\gamma$ and the identity channel $\IC$ is 
    $\dnorm{\E_{\depo,\gamma}^{1}-\IC}=2\gamma$.
\end{corollary}
The worst-case, state-independent bound is as follows.
\begin{lemma}[Diamond norm upper bound of physical error]\label{apd:lemma:physical_error}
    The physical error induced by the tensor product of one-qubit depolarizing channels with noise rate $\gamma$ has the diamond norm upper bound
    $\dnorm{\E_{\depo,\gamma}^{n}-\IC^{\otimes n}}=2[1-(1-\gamma)^n]\le2n\gamma$.
\end{lemma}

\cref{apd:lemma:physical_error} is directly implied by the chain property of channel distance:
$\dnorm{\E_1\circ \E_2- \E_1'\circ \E_2'} \le \dnorm{\E_1-\E_1'} + \dnorm{\E_2-\E_2'}$.
Nevertheless, the diamond norm bound is too pessimistic because it is considered the worst state in an extended Hilbert space.
Therefore, we derive a refined upper bound using the trace norm.

\subsubsection{The state-dependent analysis}
Since there are no ancilla qubits in the Trotter circuit, 
the trace norm $\trnorm{\cdot}$ is enough and more suitable for our error analysis.
We first introduce \cref{thm:pinsker_ineq} which gives the upper bound of trace distance by the relative entropy.
\begin{lemma}[Pinsker's inequality]\label{thm:pinsker_ineq}
    Given two density matrices $\rho$ and $\sigma$,
    the trace distance between them $\trnorm{\rho-\sigma}$ has the upper bound by their relative entropy, i.e.,
    \begin{equation}\label{eq:pinsker_ineq}
        \frac{1}{2}\trnorm{\rho-\sigma}^2
        \le \calD(\rho\Vert\sigma).
    \end{equation}
    where $\calD(\rho \Vert \sigma):=\Tr(\rho(\log(\rho)-\log(\sigma)))$ is the relative entropy between $\rho$ and $\sigma$.
\end{lemma}

Then, we utilize the entropy contraction of the depolarizing channel to give upper bounds of the errors of the Trotter simulation. 
We first begin with a stronger quantum Shearer's inequality proved in \cite{yanLimitationsNoisyQuantum2023}, 
which is used for the proof of the entropy contraction. 
The entropy contraction is also proved with other techniques in \cite{muller-hermesRelativeEntropyConvergence2016,francaLimitationsOptimizationAlgorithms2021}.

\begin{lemma}[A stronger quantum Shearer's inequality]\label{lem:shearer}
    Consider $t\in \mathbb{N}$ and a family $\mathcal{F}\subset 2^{\{1,2,\cdots,n\}}$ of subsets of $\{1,2,\cdots,n\}$ such that each $i$ is included in more than $t$ elements of $\mathcal{F}$. For any state $\rho \in \mathcal{D}\left( \mathbb{C}^{d_1} \otimes \mathbb{C}^{d_2} \otimes \cdots \otimes  \mathbb{C}^{d_n} \right)$, we have
    \begin{equation}
        \sum_{F\in \mathcal{F}} S(\rho_F) \ge tS(\rho),
    \end{equation}
    in which $\rho_F$ is the reduced density matrix of $\rho$ on subsystems $F$.
\end{lemma}

With the stronger quantum Shearer's inequality, the entropy contraction given by the following lemma can be proved.

\begin{lemma}[Entropy contraction for depolarizing noise \cite{muller-hermesRelativeEntropyConvergence2016,yanLimitationsNoisyQuantum2023}]\label{thm:entropy_contraction}
    For any $n$-qubit state $\rho$ and one-qubit depolarizing channel $\E$ of strength $\gamma$, 
    we have
    \begin{equation}
        S(\E(\rho)) \ge (1-\gamma) S(\rho) + \gamma n,
    \end{equation} 
    where $S=-\Tr(\rho\log\rho).$
    The relative entropy then contracts as
    \begin{equation}
        \calD\qty(\E(\rho) \| \frac{I}{2^n}) \le (1-\gamma)  \calD\qty(\rho \| \frac{I}{2^n}).
    \end{equation} 
    where $\calD\qty(\rho \| \frac{I}{2^n}) := \calD(\rho \Vert \sigma):=\Tr(\rho(\log(\rho)-\log(\sigma))) = n - S(\rho)$.

\end{lemma}
Since the unitary layer does not change the fixed point state, we naturally have the following corollary.
\begin{corollary}[\cite{francaLimitationsOptimizationAlgorithms2021}]\label{apd:thm:contraction}
    Let $\E_{\gamma}^n$ be the tensor product of one-qubit depolarizing channels with fixed point $\sigma=\I/2^n$ and
    $\pfc$ be the channel of one step Trotter, 
    we have the relative entropy contraction
    \begin{equation}
        \calD\qty((\pfc\circ \E_{\gamma}^n)^r(\rho_0) \Vert \sigma) 
        \le (1-\gamma)^{r} \calD(\rho_0 \Vert \sigma).
    \end{equation}
    for any initial state $\rho_0$.
\end{corollary}

Another property, the joint convexity of relative entropy is used for later proofs.
\begin{lemma}[Joint convexity of relative entropy]
    \label{lem:joint_convexity}
    The relative entropy is jointly convex in its arguments, i.e., for any $\lambda\in[0,1]$, 
    \begin{equation}
        \calD(\lambda \rho_1 + (1-\lambda)\rho_2 \| \lambda \sigma_1 + (1-\lambda)\sigma_2) \le \lambda \calD(\rho_1\|\sigma_1) + (1-\lambda)\calD(\rho_2\|\sigma_2).
    \end{equation}
\end{lemma}

\begin{figure}[!t]
    \centering
    \includegraphics[width=0.45\linewidth]{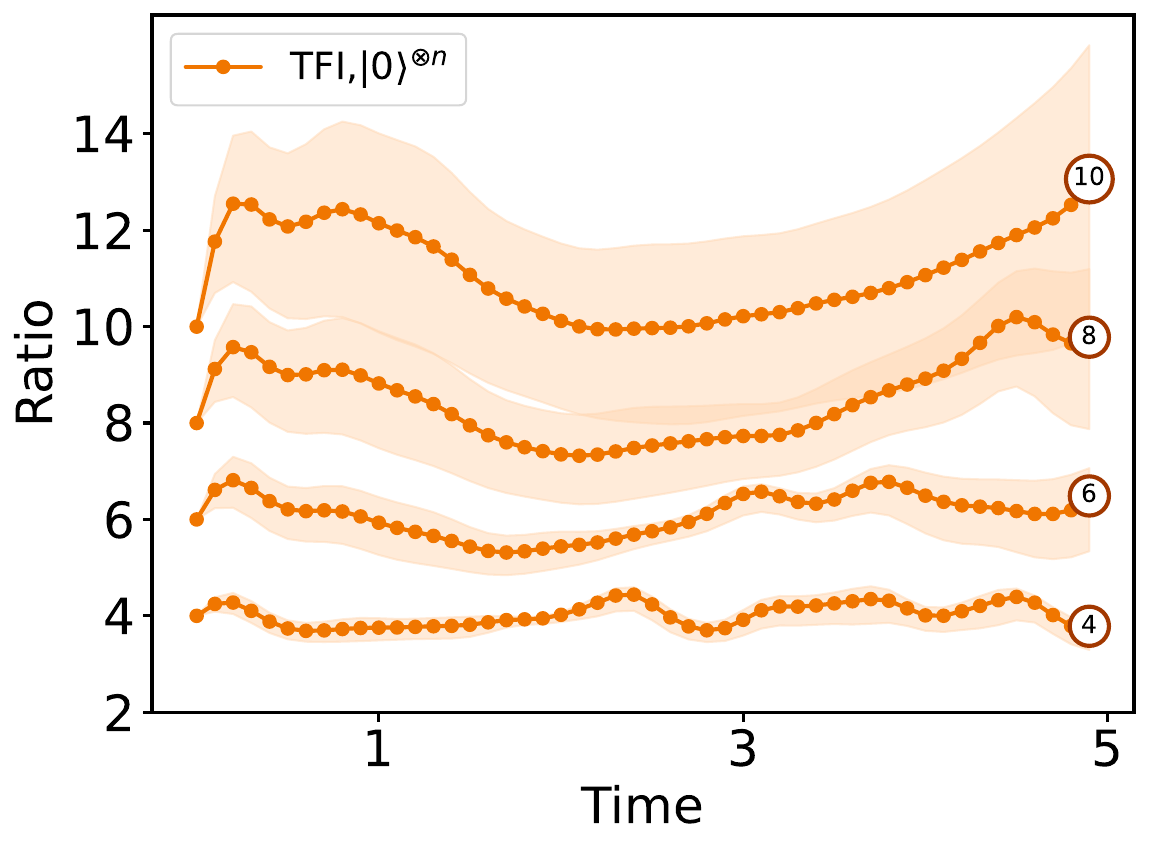}
    \caption{
    The ratio of the entropy distance to the globally mixed state and the entropy distance to the locally mixed state. 
    Plotted for TFI model with the initial state $\ket{0}^{\otimes n}$. 
        The number in the circle denotes the system size $n$.
        The solid lines are plotted for $\gamma = 0.006$, and the ratio range for $\gamma \in [0.002, 0.010]$ is noted by the shaded area.
    }
    \label{fig:glob_loc_ent_ratio}
\end{figure}

We prove the Trotter error decay by first relating the error to the distance between the evolved state and the locally mixed state. 
Specifically, in the first order approximation, the error is proportional to $\norm{\rho - \frac{1}{n}\sum_{F:\abs{F}=1} \rho_{\overline{F}}\otimes \frac{I}{2}}_1$, where $F$ takes the subsystems of the state, and $\rho_{\overline{F}}$ denotes the reduced density matrix on the complement of $F$. 
Since we have the entropy contraction with respect to the globally mixed state, we provide a relation between the distance to the locally mixed state and the distance to the globally mixed state. 
Using \cref{lem:shearer}, we have the following relation.

\begin{lemma}
    \label{lem:local_global_relation}
    Let $\rho$ be an $n$-qubit quantum state. 
    $\calD(\cdot \| \cdot)$ denotes the relative entropy. 
    $F$ takes the subsets of $\{1,2,\cdots,n\}$. $\rho_{\overline{F}}$ denotes the reduced density matrix of $\rho$ on subsystem $\overline{F}$. Then, 
    \begin{equation}
        \frac{1}{n}\calD\qty(\rho \| \frac{I}{2^n}) \le \frac{1}{n} \sum_{F:\abs{F}=1} \calD\qty(\rho\|\rho_{\overline{F}}\otimes \frac{I}{2}) \le \calD\qty(\rho \| \frac{I}{2^n}).
    \end{equation}
    This means that the average distance of $\rho$ to the locally mixed state is bounded within range $\qty[ \frac{1}{n} \calD\qty(\rho \| \frac{I}{2^n}), \calD\qty(\rho \| \frac{I}{2^n})  ]$.
\end{lemma}

\begin{proof}
    By direct computation, the relative entropy between $\rho$ and the maximally mixed state $I/2^n$ is
    \begin{gather}
        \calD\qty(\rho\| \frac{I}{2^n}) = n - S(\rho), 
    \end{gather}
    and the relative entropy between $\rho$ and the locally mixed state $ \rho_{\overline{F}} \otimes \frac{I}{2}$ is
    \begin{equation}
        \calD\qty(\rho \| \rho_{\overline{F}} \otimes \frac{I}{2}) = 1 + S(\rho_{\overline{F}}) - S(\rho).
    \end{equation}
    Since $\abs{\overline{F}} = n - 1$, the entropy $S(\rho_{\overline{F}})\le n - 1$. 
    Thus, we have $\calD\qty(\rho \| \rho_{\overline{F}} \otimes \frac{I}{2}) \le \calD\qty(\rho\| \frac{I}{2^n})$ 
    such that $\frac{1}{n} \sum_{F:\abs{F}=1} \calD\qty(\rho\|\rho_{\overline{F}}\otimes \frac{I}{2}) \le \calD\qty(\rho \| \frac{I}{2^n})$.

    Note that 
    \begin{equation}
        \sum_{F:\abs{F}=1} \calD\qty(\rho\|\rho_{\overline{F}}\otimes \frac{I}{2}) = n +  \sum_{F:\abs{F}=1} S(\rho_{\overline{F}}) - n S(\rho).
    \end{equation}
    In the second term, each site $i$ is included $n-1$ times. By \cref{lem:shearer}, 
    \begin{equation}
        \sum_{F:\abs{F}=1} \calD\qty(\rho\|\rho_{\overline{F}}\otimes \frac{I}{2}) = n +  \sum_{F:\abs{F}=1} S(\rho_{\overline{F}}) - n S(\rho) \ge n + (n-1)S(\rho) - nS(\rho) = n - S(\rho) =  \calD\qty(\rho \| \frac{I}{2^n}).
    \end{equation}
    We have the desired inequality.
\end{proof}

With the entropy contraction \cref{apd:thm:contraction} and the Pinsker's inequality \cref{thm:pinsker_ineq}, 
we can prove the exponential decay of physical error in noisy Trotter circuits.

\begin{proposition}[Exponential decay of physical error]\label{apd:thm:exp_phy_bound}
    Given an $n$-qubit depolarizing channel $\E_{\depo, \gamma}^{n}$ with noise rate $\gamma$ and assume $\gamma=o(n^{-1/2})$,
    the physical error in one step has the exponential decay upper bound
    \begin{equation}
        \eps_{p,\gamma}^{\phy}(d) := 
        \trnorm{ \E_{\depo, \gamma}^{n}(\rho_d) - \rho_d }
        \le \sqrt{2} n^{\frac{3}{2}}\gamma e^{-\frac{1}{2}\gamma d} + o(n\gamma),
    \end{equation}
    where $\rho_d$ is the $n$-qubit state of the $d$-th noisy $p$th-order Trotter step.
    Additionally, if the evolved holds to satisfy
    \begin{equation}
       \sum_{F:\abs{F}=1} \calD\qty(\rho\|  \frac{1}{n}  \rho_{\overline{F}}\otimes \frac{I}{2}) 
       = \Theta\qty(\frac{1}{n}) \calD\qty(\rho \| \frac{I}{2^n}),
    \end{equation}
    then the physical error in one step has the exponential decay upper bound 
    \begin{equation}
        \trnorm{ \E_{\depo, \gamma}^{n}(\rho) - \rho }
        \le \sqrt{2} n\gamma e^{-\frac{1}{2}\gamma d} + o(n\gamma).
    \end{equation}
\end{proposition}
\begin{proof}
    The $d$-th step state after one layer of noise channel is 
    \begin{equation}
        \mathcal{E}_{\depo, \gamma}^{n}(\rho_d) = \sum_{F} \gamma^{|F|}(1-\gamma)^{n-|F|} \rho_{\overline{F}} \otimes \frac{I}{2^{|F|}},
    \end{equation} 
    where $F$ takes all subsets of qubits representing the depolarized subsystems, 
    $\overline{F}$ is the complement of $F$ and $\rho_{\overline{F}}$ is the reduced density matrix of $\rho$ on subsystem $\overline{F}$.
    The physical error defined as is the trace distance between the depolarized state and the original state,
    \begin{align}
        \trnorm{ \mathcal{E}_{\depo, \gamma}^{n}(\rho) - \rho } = 
        \left\| \sum_{F} \gamma^{|F|}(1-\gamma)^{n-|F|} \rho_{\overline{F}} \otimes \frac{I}{2^{|F|}} - \rho \right\|_{1}
    \end{align}
    where we omit the subscript $d$ for simplicity.
    Since $\gamma$ is small, we approximate to the first-order, i.e. $F$ takes all cardinality-one subsets and the empty set. Then,
    \begin{align}
        \trnorm{ \E_{\depo, \gamma}^{n}(\rho) - \rho } 
        = & \left\| \sum_{k=0}^n (-\gamma)^k \binom{n}{k} \rho + \gamma(1-\gamma)^{n-1} \sum_{F:\abs{F}=1} \rho_{\overline{F}}\otimes \frac{I}{2} \right .+ \left.\sum_{k=2}^{n} \gamma^k (1-\gamma)^{n-k} \sum_{F:\abs{F}=k}\rho_{\overline{F}}\otimes \frac{I}{2^k}  - \rho \right\|_1 
        \\ 
        = & \norm{ -n\gamma \rho + \gamma \sum_{F:\abs{F}=1} \rho_{\overline{F}}\otimes \frac{I}{2}}_1 + o(n\gamma) \tag{first-order approximation} \\
        \leq & n\gamma \sqrt{ 2\calD \qty( \rho \Vert  \frac{1}{n}\sum_{F:|F|=1} \rho_{\overline{F}} \otimes \frac{I}{2})  } + o(n\gamma) 
        \tag{Pinsker's inequality, \cref{thm:pinsker_ineq}} \\
        \le & n\gamma \sqrt{ 2\cdot \frac{1}{n}\sum_{F:|F|=1}\calD \qty( \rho \Vert   \rho_{\overline{F}} \otimes \frac{I}{2})  } + o(n\gamma)  
        \tag{joint convexity, \cref{lem:joint_convexity}} \\ 
        \leq & n\gamma \sqrt{ 2 \calD \qty(\rho \Vert \frac{I}{2^n} ) } + o(n\gamma) \tag{\cref{lem:local_global_relation}} \\
        \le & n\gamma \sqrt{ 2 \calD \qty(\rho_0 \Vert \frac{I}{2^n} ) e^{-\gamma d}} + o(n\gamma) \tag{\cref{thm:entropy_contraction}} \\ 
        \leq &  \sqrt{2}  n^{\frac{3}{2}}\gamma e^{-\frac{1}{2}\gamma d} + o(n\gamma).
    \end{align}
    The first line is direct expansion. 
    The second line is the first-order approximation. 
    The second line to the third is by the Pinsker's inequality \cref{thm:pinsker_ineq} and the fourth line is by the joint convexity in \cref{lem:joint_convexity}. 
    The next lines are by the relation of the distance to the locally mixed state and the distance to the globally mixed state in \cref{lem:local_global_relation}, 
    and by the entropy contraction in \cref{thm:entropy_contraction}.
    If the system holds to satisfy 
    $\sum_{F:\abs{F}=1} \calD\qty(\rho\|\frac{1}{n} \rho_{\overline{F}}\otimes \frac{I}{2}) = \Theta\qty(\frac{1}{n}) \calD\qty(\rho \| \frac{I}{2^n})$ (see numerical evidence in \cref{fig:glob_loc_ent_ratio}), 
    then the last line can be improved with a $\frac{1}{n}$ factor in the square root, leading to the final bound $\sqrt{2} n\gamma e^{-\frac{1}{2}\gamma d} + o(n\gamma)$.
\end{proof}

The one-qubit depolarizing noise model should be taken as a toy model to exemplify our techniques. However, results on the entropy contraction are available for other relevant noise models

\subsection{Upper bounds of algorithmic error}\label{apd:sec:alg}
Before analyzing the Trotter error affected by noise, we first recall the general upper bound of Trotter error in the noiseless case.

\subsubsection{The worst-case (state-independent) analysis}
The worst-case (state-independent) upper bound on Trotter error is as follows.
\begin{lemma}[Trotter error with commutator scaling \cite{childsTheoryTrotterError2021}]\label{apd:thm:pf}
    Let $H=\sum_{l=1}^L H_l$ be a Hermitian operator consisting with $L$ summands, and let $t \ge 0$. 
    Let  $\pf_p(t)$ be a $p$th-order product formula. 
    Define $\acmm_{p} :=\sum\limits_{l_1,\dots,l_{p+1}=1}^{L} \opnorm{[H_{l_1},[H_{l_2},\dots,[H_{l_p},H_{l_{p+1}}]]]}$. 
    Then, the Trotter error $\opnorm{\pf_p(t)-e^{-\ii Ht}}$ can be asymptotically bounded as
    $\bigO(\acmm_{p} t^{p+1})$.
    To achieve precision $\epsilon$, we need Trotter steps 
    $r=\bigO\qty(\acmm_{p}^{1/p} t^{1+1/p}/\epsilon^{1/p})$
    for precision $\opnorm{\pf^r_p(t/r)-e^{-\ii Ht}}\le \epsilon$.
\end{lemma}
In \cref{fig:worst_algorithmic_error}, we calculate the worst-case one-step Trotter error bound and the empirical error for system sizes $n$ from 4 to 12.
And we fit the linear functions of the theoretical error bound and the empirical error in terms of $n$ such that they can be extrapolated to larger system sizes.
It is worthwhile noting that for the average-case analysis (assuming the initial state is Haar random), 
the Trotter (algorithmic) error upper bound is the same except the spectral norm $\opnorm{\cdot}$ replaced by the normalized Frobenius norm $\norm{\cdot}_{2}$ \cite{zhaoHamiltonianSimulationRandom2021,zhaoEntanglementAcceleratesQuantum2025}.

\begin{figure}[!t]
    \centering
    \includegraphics[width=0.45\linewidth]{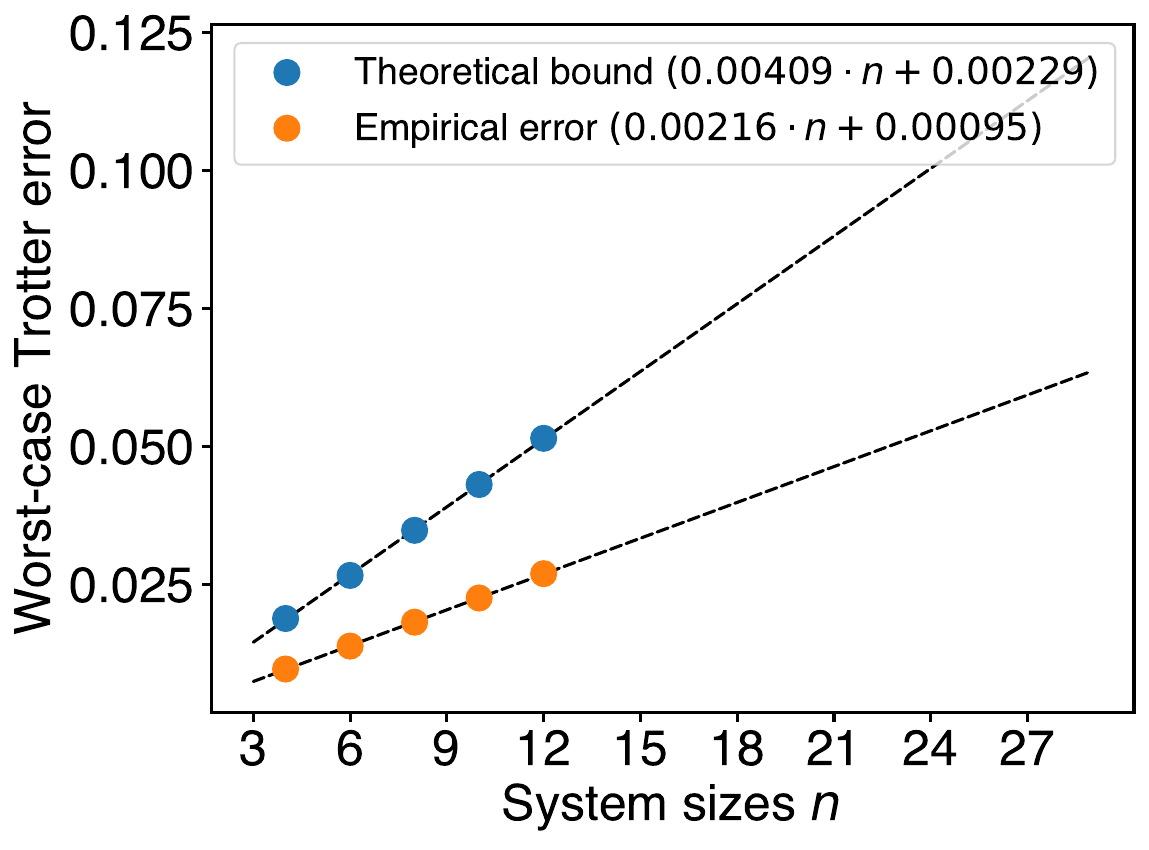}
    \caption{
    The worst-case (one step) Trotter error commutator upper bound versus system size $n$.
    The numerical result is for the TFI Hamiltonian of different sizes $n$ with the XZ grouping by the second-order product formula (PF2) with $t=10$ and $r=100$.
    The worst-case theoretical bound is evaluated by exact calculation of the operator norm $\alpha_p$ of the nested commutator,
    while the empirical error is the operator norm of the difference between the ideal evolution and the Trotter evolution operator.
    For $n\le 12$, the data is the exact calculation. 
    While for large system sizes $n$,
    the Trotter error bound is extrapolated by the linear fitting.
    }
    \label{fig:worst_algorithmic_error}
\end{figure}

\subsubsection{The state-dependent analysis}
By taking the state information into account, we derive the upper bound of the Trotter (algorithmic) error with noisy circuits.
\begin{lemma}[Trace norm of algorithmic error]\label{apd:thm:trace_norm_alg_error_bound}
    Given a state $\rho$, one-step ideal evolution $U=e^{-\ii H \delta t}$, 
    and one-step $p$-th order Trotter evolution $\pf_p(\delta t)$, 
    the one-step algorithmic (Trotter) error measured by the trace norm is 
    \begin{align}
        \trnorm{U\rho U^{\dag}-\pf\rho \pf^{\dag}} = \trnorm{ [\rho, M] },
    \end{align}
    where $\pf_p=U(I+M_p)$ and $M_p$ is the multiplicative Trotter error operator.
\end{lemma}
\begin{proof}[Proof of \cref{apd:thm:trace_norm_alg_error_bound}]\label{apd:prf:trace_norm_alg_error_bound}
    With $\pf_p = U(I+M_p)$,
    we can rewrite the one-step algorithmic error as 
    \begin{align}
        \trnorm{U\rho U^{\dag}-\pf_p\rho \pf^{\dag}_p}
        &= \trnorm{U\rho U^{\dag}-U(I+M_p)\rho (I+M_p^{\dagger}))U^{\dag} } \\
        &= \trnorm{\rho -(I+M_p)\rho (I+M_p^{\dagger})} \tag{unitarity}\\
        &= \trnorm{\rho(I+M_p) -(I+M_p)\rho} \tag{unitarity} \\
        &= \trnorm{\rho M_p -M_p\rho }
         \equiv \trnorm{[\rho, M_p]}
    \end{align}  
\end{proof}
Then, we have the upper bound of the trace norm of algorithmic error with exponential decay by considering the state information along evolution.
\begin{proposition}[Exponential decay of algorithmic error]\label{apd:thm:exp_alg_bound}
    Given a $p$th-order Trotter circuit $\pf_p$ with the noise rate $\gamma$,
    the $d$-th one-step algorithmic error $\epsilon_{p,\gamma}^{\alg}$ has the upper bound
    \begin{equation}
        \epsilon_{p,\gamma}^{\alg}(d) := 
        \trnorm{ U \rho_d U^\dagger - \pf_p \rho_d \pf_p^\dagger }
        \leq \Theta(n^{1/2}) 
        B_{p} \frac{t^{p+1}}{r^{p+1}} e^{-\frac{1}{2} \gamma d}
    \end{equation}
    where $\rho_d$ is the state of the $d$-th noisy Trotter step and 
    $B_{p}:=\sum_j \opnorm{E_j^{(p)}}$ is a factor related to the Hamiltonian and the order of the Trotter formula.
\end{proposition}
\begin{proof}
    The multiplicative Trotter error operator $M_p$ can be written as the sum of the leading-order local error terms $M_p:=\delta t^{p+1}\sum_j E_j^{(p)}$ 
    where $E_j^{(p)}$ acts on $w$ qubits and $\delta t=t/r$.
    Let $\rho':=\rho_{\overline{F}}\otimes I/2^w$
    and omit the subscript $d$ as well as $p$ in the following proof for simplicity.
    \begin{align}
        \trnorm{U\rho U^{\dag}-\pf\rho \pf^{\dag}}
        &= \trnorm{[\rho, M]} 
        = \delta t^{p+1} \trnorm{\qty[\rho, \sum_j E_j]} \tag{definition} \\
        &\le \delta t^{p+1} \sum_j \trnorm{ [\rho, E_j] } 
        \equiv \delta t^{p+1} \sum_j \trnorm{\rho E_j - \rho' E_j + \rho' E_j - E_j \rho' + E_j\rho' -E_j \rho} \\
        &\le \delta t^{p+1} \sum_j \trnorm{ \rho E_j - \rho' E_j } +\trnorm{\rho' E_j -E_j \rho'}  + \trnorm{ E_j\rho' - E_j\rho }
        \tag{triangle inequality} \\
        &\le \delta t^{p+1} \sum_j 2\trnorm{ \rho  - \rho' } \opnorm{E_j} + \trnorm{[\rho',E_j]}   
        \tag{Hölder's inequality}  \\
        &= \delta t^{p+1} 2\trnorm{I/2^w\otimes\rho_{\overline{F}}-\rho}  \sum_j \opnorm{E_j} 
        \tag{locality}  \\ 
        &\le \delta t^{p+1}\Theta(n^{1/2})  e^{-\frac{1}{2}\gamma d} \sum_j \opnorm{E_j} 
        \tag{\cref{apd:thm:exp_phy_bound}}  \\ 
        &\leq \Theta(n^{1/2}) 
        B_{p} \frac{t^{p+1}}{r^{p+1}} e^{-\frac{1}{2}\gamma d}
        \tag{\cref{apd:thm:pf}}
    \end{align}  
    The last third line uses $\trnorm{[\rho', E_j]}=0$ because the error operator $E_j$ acts on identity.
    The exponential decay factor in the last second line is the same as the one of the physical error proved in \cref{apd:thm:exp_phy_bound}.
    In the last line, the factor $B_p=\sum_j \opnorm{E_j^{(p)}}$ is the upper bound of the worst-case Trotter error
    $\opnorm{\sum_j E_j^{(p)}}=\bigO(\acmm_p)$ where $\acmm_p$ is the $p$th-order nested commutator norm, only depends on the Hamiltonian $H$ and the order $p$ of the Trotter formula.
\end{proof}

\subsection{Optimal Trotter number and noise rate requirement}\label{apd:sec:threshold}
\begin{proposition}[Optimal Trotter number and noise rate requirement]
    Consider the $p$th-order product formula circuit with depolarizing noise rate $\gamma$. 
    We assume the total error per step at Trotter step $d$ is of the form 
    $\eps_{p,\gamma}^\tot (d) = C\gamma\Upsilon e^{-c\gamma\Upsilon d} + B_p \frac{t^{p+1}}{r^{p+1}}e^{-b\gamma\Upsilon d}$. 
    $C$ and $B$ are the prefactor coefficients for the physical (algorithmic) error, 
    while $c$ and $b$ are the decay rates, respectively. 
    The overhead $\Upsilon = \bigO( 2^p)$ is the number of layers implemented in one Trotter step in higher-order formulas. 
    Assume $c<C$ and $b\approx c$.
    To achieve the smallest error $\eps$, the optimal number of Trotter steps is
    $r_\opt(\gamma) = \qty(\frac{pB_p}{C\gamma\Upsilon})^{\frac{1}{p+1}}t$. 
    To achieve such a small error $\eps$, 
    the noise rate has to be lower than 
    $\gammastar = \frac{1}{C\qty(B_p)^{\frac{1}{p}}} \qty(\frac{\epsilon}{t})^{1+\frac{1}{p}} \frac{p}{\Upsilon (p+1)^{1+\frac{1}{p}}}$.
    Taking $C=\Theta(n), B_p=\Theta(n)$, we have 
    $r_\opt(\gamma) = \Theta\qty( \gamma^{-\frac{1}{p+1}} t)$ and 
    $\gammastar = \bigO\qty(\qty(\frac{\epsilon}{nt})^{1+\frac{1}{p}})$. 
    Particularly, $ r_{\opt}(\gammastar) = \Theta\qty( t\qty(\frac{nt}{\epsilon})^{\frac{1}{p}} )$.
\end{proposition}
\begin{proof}
    The accumulated error we bound by triangle inequality as the sum of the errors of all steps, i.e., $\epsilon_{p, \gamma}^{\acc}(r) = \sum_{d=1}^{r} \epsilon_{p, \gamma}^{\tot}(d)$. 
    To let the accumulated error be lower than $\epsilon$, we have to at least let the physical error portion be lower than $\epsilon$, i.e., $\sum_{l=1}^r C\gamma\Upsilon e^{-c\gamma\Upsilon l} \le \epsilon$. 

    By taking the summation, we have the accumulated physical error,
    \begin{equation}
        C\gamma\Upsilon \frac{1-e^{-c\gamma\Upsilon r}}{1-e^{-c\gamma\Upsilon}} \le \epsilon.
    \end{equation}
    Since $C\gamma\Upsilon$ should be small, and $c<C$ by the assumption. 
    The above can be approximated to 
    \begin{equation}
        \frac{C}{c} (1-e^{-c\gamma\Upsilon r}) \le \epsilon.
    \end{equation}
    Since $\epsilon c/C$ is small, the parenthesis approximates $1-e^{-c\gamma\Upsilon r}\approx c\gamma\Upsilon r$.
    Since $b\approx c$, we have that both physical and algorithmic error decay approximately linearly 
    \begin{equation}
        \epsilon_{p, \gamma}^{\acc}(r) \approx 
        C\gamma\Upsilon r + 
        B_p \frac{t^{p+1}}{r^p} \le \epsilon,
    \end{equation}
    if the desired $\epsilon$ is small. 
    By standard Arithmetic-Geometric inequality, we have 
    \begin{equation}
       \min_r \qty( C\gamma\Upsilon r + B_p \frac{t^{p+1}}{r^p} ) = 
       (p+1) \qty( \frac{1}{p^p} \qty(C\gamma\Upsilon)^p B_p t^{p+1}  )^{\frac{1}{p+1}} = 
       (p+1)\qty( \frac{C\gamma\Upsilon}{p} )^{\frac{p}{p+1}} \qty(B_p)^{\frac{1}{p+1}}t.
    \end{equation}
    This is the lowest error the Trotter simulation can give, 
    and thus it must be lower than $\epsilon$, 
    imposing a noise requirement $\gammastar$
    \begin{equation}\label{eq:gamma_thr}
        \gammastar \le \frac{1}{\qty(B_p)^{\frac{1}{p}}} \qty(\frac{\epsilon}{t})^{1+\frac{1}{p}} \frac{p}{C\Upsilon (p+1)^{1+\frac{1}{p}}}.
    \end{equation}
    And the optimal $r$ is obtained 
    when $\frac{1}{p}C\gamma\Upsilon r = B_p \frac{t^{p+1}}{r^p}$:
    \begin{equation}\label{eq:optimal_r}
        r_{\opt}(\gamma) = 
        \qty(\frac{pB_p}{C\gamma\Upsilon})^{\frac{1}{p+1}}t.
    \end{equation}
    And let $r_\opt^*:=r_\opt(\gammastar)$, then
    \begin{equation}
        r_\opt^* 
        := \qty(\frac{pB_p}{C\gammastar\Upsilon})^{\frac{1}{p+1}}t 
        = \qty(B_p)^{\frac{1}{p}} (p+1)^{\frac{1}{p}} \qty(\frac{t^{p+1}}{\epsilon})^{\frac{1}{p}}.
    \end{equation}
    Assuming $C=\Theta(n), B_p=\Theta(n)$, 
    we have the asymptotic scaling of \cref{eq:gamma_thr} and \cref{eq:optimal_r}
    \begin{equation}
        \gammastar = \Theta\qty(\qty(\frac{\epsilon}{nt})^{1+\frac{1}{p}}), 
        \quad r_{\opt}(\gamma) = \Theta\qty( \gamma^{-\frac{1}{p+1}} t).
    \end{equation}
    Particularly, the optimal Trotter number at the noise rate threshold has the scaling
    \begin{equation}
        r_{\opt}^* 
        = \Theta\qty( t\qty(\frac{nt}{\epsilon})^{\frac{1}{p}} ).
    \end{equation}

\end{proof}
    Our result agrees with the results of the specific Hamiltonians in \cite{avtandilyanOptimalorderTrotterSuzuki2024}.
    We give examples of low-order Trotter formulas.
    For the first-order Trotter, take $\Upsilon=2$, we have 
    \begin{equation}
        \gammastar = 
        \frac{1}{8CB_p} \qty(\frac{\epsilon}{t})^2 = 
        \Theta\qty(\frac{\epsilon^2}{n^2t^2}), \;
        r_\opt(\gamma=\gammastar) = 
        2B_p\frac{t^2}{\epsilon}=
        \Theta\qty(\frac{nt^2}{\epsilon}).
    \end{equation} 
    Similarly, for the second-order Trotter, take $\Upsilon=4$,
    \begin{equation}
        \gammastar = \frac{1}{6\sqrt{3} 
        C(B_p)^{\frac{1}{2}}} 
        \qty(\frac{\eps}{t})^{\frac{3}{2}} 
        = \Theta \qty(\qty(\frac{\epsilon}{nt})^{\frac{3}{2}}),  \;
        r_\opt(\gamma=\gammastar) =  \sqrt{3} 
        (B_p)^{\frac{1}{2}}t 
        \qty(\frac{t}{\epsilon})^{\frac{1}{2}} = \Theta \qty( t \qty(\frac{nt}{\epsilon})^{\frac{1}{2}} ).
    \end{equation}

\section{Additional numerical results}\label{apd:sec:numeric}
To substantiate the generality of our theoretical bounds, 
we provide additional numerical results from several aspects: 
noise channels;
physical models;
observables and the error of expectation values.
Our numerical code is mainly built on \textsc{Qiskit} \cite{javadi-abhariQuantumComputingQiskit2024}, \textsc{OpenFermion} \cite{mccleanOpenFermionElectronicStructure2020}, and \textsc{PySCF} \cite{sunRecentDevelopmentsPySCF2020}.

\begin{figure}[!t]
    \centering
    \includegraphics[width=0.85\linewidth]{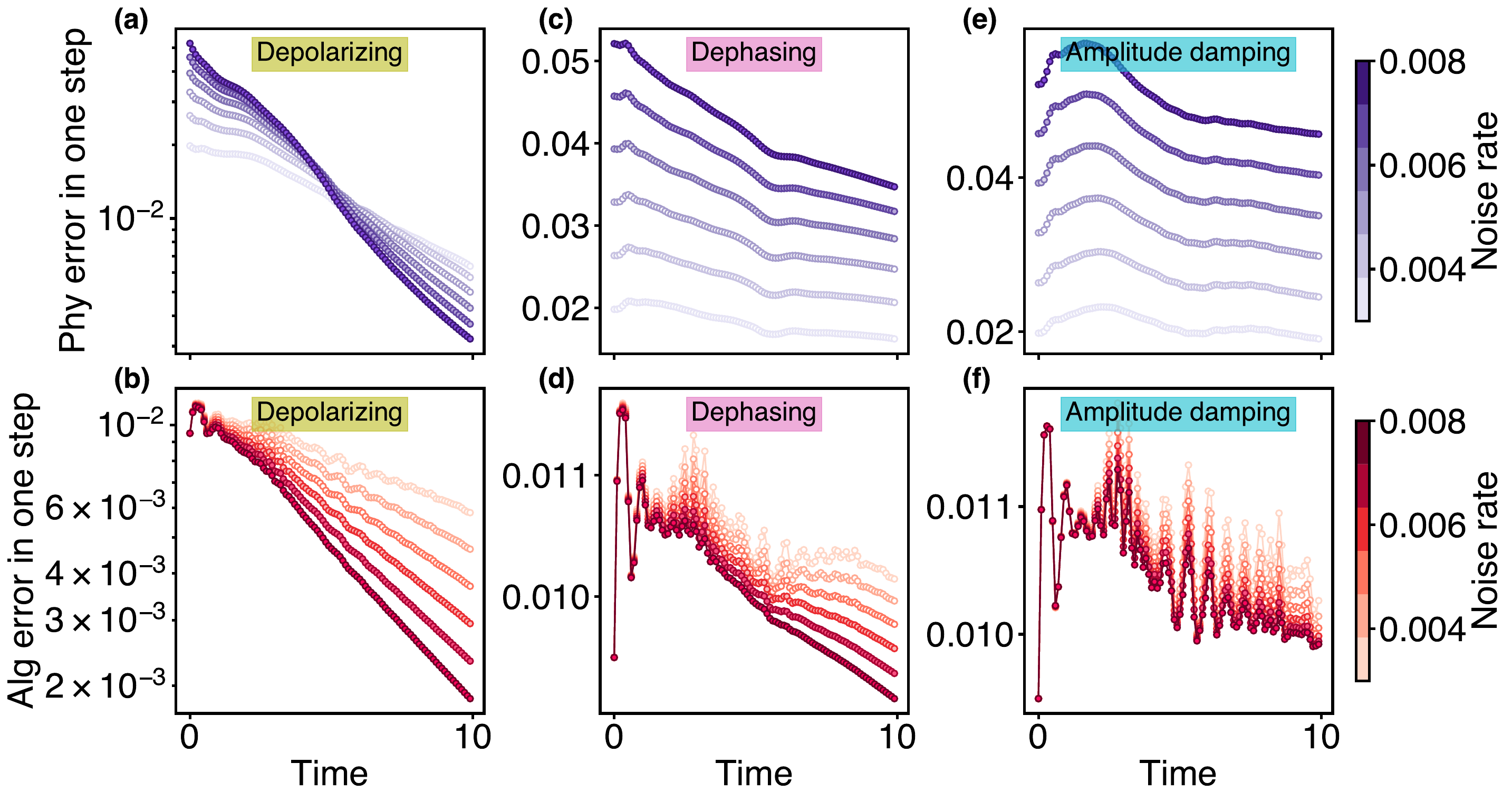}
    \caption{
    The one-step physical and algorithmic errors decay under different types of noise channels.
    We adopt the TFI Hamiltonian with $n=10, t=n$ and PF2 with $r=100$.
    (a) (b) The local depolarizing channel \cref{apd:eq:single_depolarizing}. 
    (c) (d) The local dephasing channel \cref{eq:single_dephasing}.
    (e) (f) The amplitude-damping noise channels \cref{eq:single_amp_damping}.
    }
    \label{fig:other_noise}
\end{figure}

\subsection{Impact of noise channels}\label{apd:noise_channels}
Besides the depolarizing channel that we have discussed in the main text, 
there are two widely used noise channels for modeling realistic noise, i.e., 
the dephasing noise channel and the amplitude damping channel.
The dephasing noise channel on a single qubit is
\begin{align}\label{eq:single_dephasing}
    \E^{\deph}_{\gamma}(\rho) 
    &:= (1-\gamma)\rho + \gamma Z\rho Z .
\end{align}
The one-qubit (local) dephasing noise channel on one layer of the $n$-qubit circuit is the tensor-product of \cref{eq:single_dephasing}.

On the contrary to unital noise channels (i.e. $\E(I)=I$), the typical non-unital noise channel is the amplitude-damping channel:
$\E_\gamma^{\damp}(\rho)=E_0\rho E_0^\dagger + E_1 \rho E_1^\dagger$ 
where the Kraus operators are
\begin{equation}\label{eq:single_amp_damping}
    E_0=\begin{pmatrix}1&0\\ 0&\sqrt{1-\gamma}\end{pmatrix}
    , \quad
    E_1=\begin{pmatrix}0&\sqrt{\gamma} \\ 0&0\end{pmatrix}.
\end{equation}

In \cref{fig:other_noise}, we compare the one-step (physical and algorithmic) error decay of the (local) depolarizing, dephasing, and amplitude-damping noise channel.
It can be expected that the errors of the depolarizing channel decay much faster than the dephasing channel. 
On the other hand, due to the non-unital property of the amplitude-damping channel, the errors of this channel do not exhibit apparent exponential decay but have large fluctuations.

\subsection{Common physical models}\label{apd:more_hamiltonians}
While our proof is independent of Hamiltonians, we provide more numeric results on additional common physical models, including the Heisenberg model with power-law decaying interaction, the Fermi-Hubbard model, and the Hydrogen chain, to validate the generality of our theory.

\begin{figure*}
    \centering
    \includegraphics[width=0.9\textwidth]{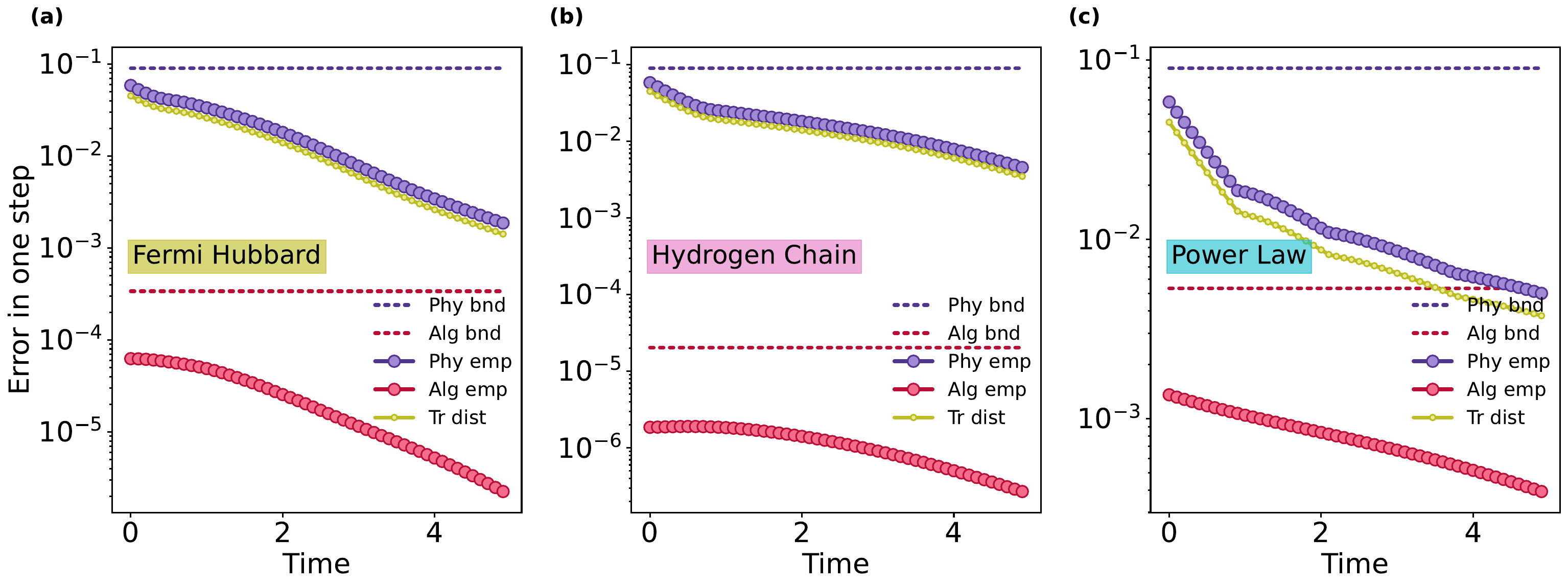}
    \caption{
    Numerical results of various physical Hamiltonians.
    We take the initial state $\ket{+}^{\otimes n}$, evolution time $t=5$, and Trotter step $r=100$. 
    The noise rate $\gamma=0.009$.
    (a) The Fermi-Hubbard model \cref{eq:hubbard} with $n=10$ and parameters $v=u=1$.
    (b) The Hydrogen-Chain model with 5 Hydrogen atoms ($n=10$ qubits) and $\text{bond length}=3$. 
    (c) The power-Law decaying Heisenberg model \cref{eq:power_law}: $n=10$ qubits and the power-law decaying coefficient $\alpha=4$.
    }
    \label{fig:decay_physical_models}
\end{figure*}

\subsubsection{Heisenberg model with power law decaying interaction}
We consider the 1D Heisenberg model $\ham_{\pow}$ with the power-law interactions 
\begin{equation}\label{eq:power_law}
    \ham_{\pow}=
    \sum_{j=1}^{n-1} \sum_{k=j+1}^n \frac{1}{\abs{k-j}^\alpha} 
    (X_jX_k + Y_j Y_k + Z_j Z_k) + \sum_{j=1}^n h_j Z_j
\end{equation}
where we take the power-law decaying exponent as $\alpha=4$.
We use the XYZ grouping for the power-law interaction Hamiltonian
\begin{equation}
    \ham_{\pow} = H_{X} + H_{Y} + H_{Z}.
\end{equation}
where each group only contains the same type of Pauli operators, 
e.g., $H_Z =\sum_{j=1}^{n-1} \sum_{k=j+1}^n \frac{1}{\abs{k-j}^\alpha} Z_j Z_k + \sum_{j=1}^n h_j Z_j$.

\subsubsection{Fermi-Hubbard model}
The Fermi-Hubbard model is a key focus in condensed matter physics due to its relevance in metal-insulator transitions, quantum magnetism, and high-temperature superconductivity \cite{hubbardElectronCorrelationsNarrow1963,schubertTrotterErrorCommutator2023}.
We study the Fermi-Hubbard model on a one-dimensional lattice of $n$ sites described by the Hamiltonian
$H_{\textsc{FH}} = H_{\even} + H_{\odd} + H_{\intt}$ with three groups
\begin{align}\label{eq:hubbard}
    H_{\even} &= v \sum_{j=1}^{\lfloor{n/2\rfloor}} \sum_{\sigma\in\qty{\uparrow,\downarrow}}  a_{2j-1,\sigma}^\dagger a_{2j,\sigma} + a_{2j,\sigma}^\dagger a_{2j-1,\sigma}, \\
    H_{\odd} &= v \sum_{j=1}^{\lceil{n/2\rceil}-1} \sum_{\sigma\in\qty{\uparrow,\downarrow}} a_{2j,\sigma}^\dagger a_{2j+1,\sigma} + a_{2j+1,\sigma}^\dagger a_{2j,\sigma}, \\ 
    H_{\intt} &= u \sum_{j=1}^n n_{j,\uparrow} n_{j,\downarrow}
\end{align}
where $j$ refer to neighboring lattice sites in the first sum, $v \in R$ is the kinetic hopping coefficient, and $u > 0$ is the on-site interaction strength. 
$a_{j,\sigma}^\dagger$, $a_{j,\sigma}$ and $\hat{n}_{j,\sigma}=\hat{a}_{j,\sigma}^\dagger \hat{a}_{j,\sigma}$ are the Fermionic creation, annihilation, and number operators, respectively, acting on the site $j$ and spin $\sigma \in \qty{\uparrow, \downarrow }$.
\cite{mccleanOpenFermionElectronicStructure2020}
The fermion operator satisfies the anti-commutation relation 
\begin{equation}\label{eq:fermion_anti_commutation}
    \qty{\acr_p, \acr_q} = 0,\;
    \qty{\aan_p, \aan_q} = 0,\;
    \qty{\acr_p, \aan_q} = \delta_{pq}.
\end{equation}

\subsubsection{Chemistry molecule}\label{sec:chem}
For quantum molecular systems, the Hamiltonian $\ham$ takes the following form
$\ham_{\chm}:=A+\frac{1}{2}V+C$ with
\begin{equation}\label{eq:molecule}
    A=\sum_{i,j=1}^n h_{ij}\ \hat{a}_i^\dagger \hat{a}_j,  \;
    V=\frac{1}{2}\sum_{i,j,k,l=1}^n V_{ijkl}\ \hat{a}_i^\dagger \hat{a}_j^\dagger \hat{a}_k \hat{a}_l,
\end{equation}
where $n$ is the number of spin orbitals of the molecular system; $C$ is a constant term; $\hat{a}_i^\dagger$ and $\hat{a}_i$ are the fermionic generation and annihilation operators, respectively; $h_{ij}$ and $V_{ijkl}$ are the corresponding coefficients for the one-body and two-body interactions, respectively. 

The chain of Hydrogen atoms is the toy model for the ground state energy problem.
The fermionic Hamiltonian can be converted to the qubit Hamiltonian, which consists of Pauli operators, through the Jordan-Wigner transformation \cite{nielsenFermionicCanonicalCommutation2005,sunRecentDevelopmentsPySCF2020}.
Then, we group the Hamiltonian into commuting groups $H=\sum_{l=1}^{L} H_l$ using a greedy heuristic.
For the Hydrogen chain with 5 Hydrogen atoms, its Hamiltonian with the STO-3G basis is represented by $n=10$ qubits and $444$ Pauli strings grouped into $L=27$ terms.

\subsection{Observable evolution and expectation value}\label{apd:observables}

\begin{figure}[!t]
    \centering
    \includegraphics[width=0.5\linewidth]{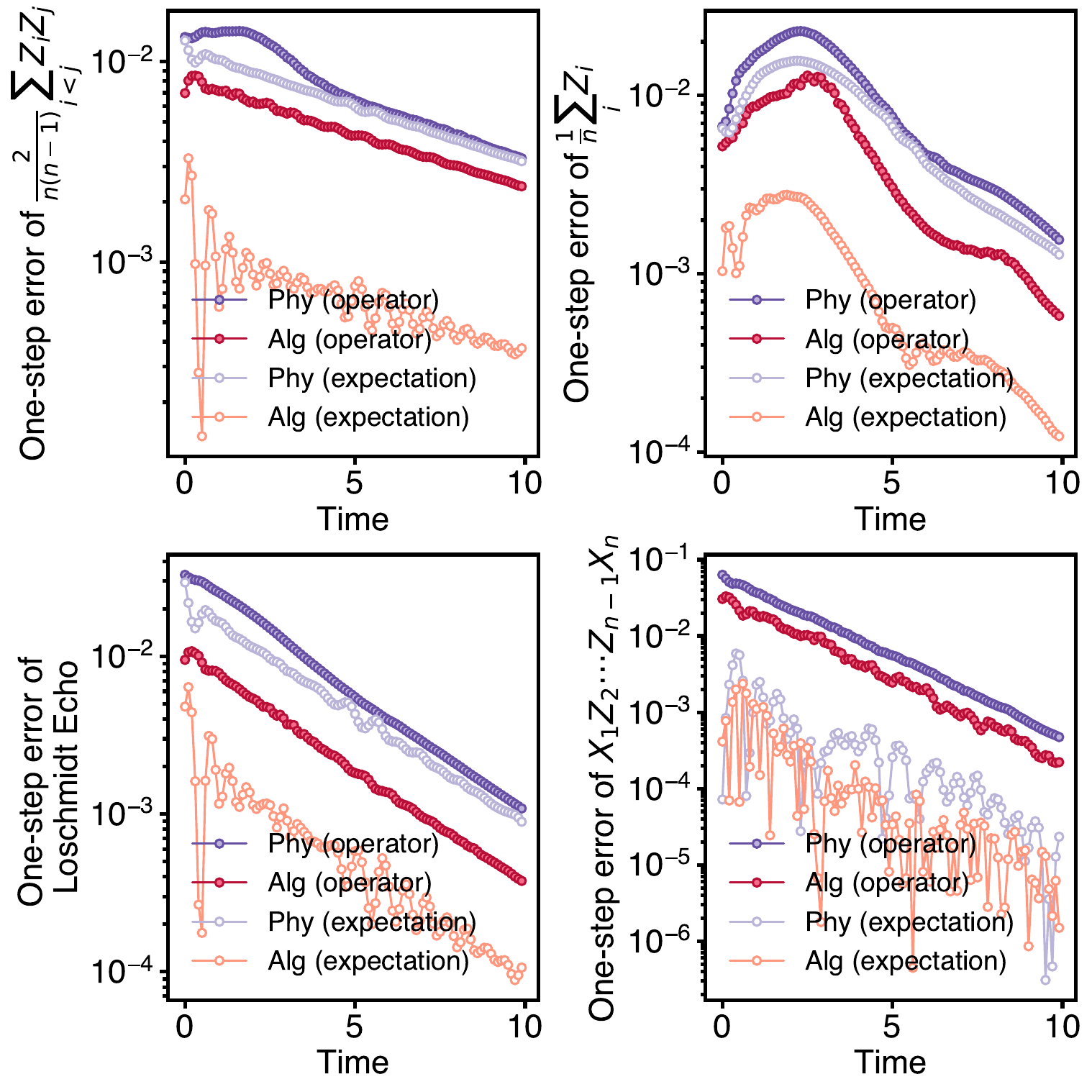}
    \caption{
    The physical and algorithmic error in common observables and their expectation values.
    We choose the TFI Hamiltonian with $n=10$ qubits, the PF2 with XZ grouping, and noise rate $\gamma=0.005$.
    The observable error curve plotted for typical observables, including single-site spin $\frac{1}{n}\sum_{i}Z_i$, two-site spin $\frac{2}{n(n-1)}\sum_{i<j}Z_iZ_j$, Loschmidt echo $\op{0^n}$ and string parameter $X_1Z_2\cdots Z_{n-1}X_n$. 
    The physical and algorithmic errors defined in \cref{eq:observable_alg_error,eq:observable_phy_error} are plotted in purple and red, respectively. 
    We also plot the expectation value error in lighter colors and the absolute value of the expectation value in grey. The expectation value errors are typically smaller than the observable errors. 
    }
    \label{fig:observable_evolution}
\end{figure}

In addition to analyzing the evolution of state error in the Schr\"odinger picture, 
we investigate the evolution of observable error in the Heisenberg picture.
Similar to the state error, the observable error can be decomposed into two components: algorithmic error and physical error. 
The one-step algorithmic error is defined as the operator norm of the difference between the ideally evolved operator and the operator evolved by the product formula, expressed as:
\begin{equation}\label{eq:observable_alg_error}
    \eps_{p}^{\alg,\mathrm{ob}}(O_d)
    := \norm{ U^{\dagger} O_d U - \tilde{U}_p^{\dagger} O_d \tilde{U}_p  }_{\infty},
\end{equation} 
where $O_d$ is the evolved observable at $d$th-step,
$U$ is the exact unitary evolution of one step, and $\pf_p$ is the $p$th-order unitary evolution of one step given by the product formula.
On the other hand, the one-step physical error is defined as the deviation caused by the noise channel, expressed as:
\begin{equation}\label{eq:observable_phy_error}
    \eps_{\gamma}^{\phy,\mathrm{ob}}(O_d): = \norm{ \calE_\gamma^n(O_d) - O_d}_{\infty}.
\end{equation}
where the adjoint channel of the depolarizing channel $\calE_\gamma^n$ is itself.
We plot the error curve for typical observables, including single-site spin $\frac{1}{n}\sum_{i}Z_i$, two-site spin $\frac{2}{n(n-1)}\sum_{i<j}Z_iZ_j$, Loschmidt echo $\op{0^n}$ and string parameter $X_1Z_2,\cdots Z_{n-1}X_n$ with results shown in \cref{fig:observable_evolution}. 
The trends observed in the plotted curves can differ from those of the state error. 
Specifically, the single-site spin observable exhibits an initial growth period before the exponential decay that is similar to the state error curve. 
We also test the expectation value error evaluated with the state $\rho=\op{0^n}$, i.e., 
\begin{equation}
    \eps_{p}^{\alg, \text{val}}(O_d,\rho):=  
    \left| \Tr( \rho \qty( U^{\dagger} O_d U - \tilde{U}_p^{\dagger} O_d \tilde{U}_p  ) )\right|
\end{equation}
and 
\begin{equation}
    \eps_{\gamma}^{\phy, \text{val}}(O_d,\rho): = 
    \left| \Tr( \rho \qty( \calE_\gamma^n(O_d) - O_d ) ) \right|
\end{equation}
with $\rho = \op{0^n}$. 
Though the error in expectation value is implied by the state error $\trnorm{\rho-\rho'}$
through the inequality $\Tr(O(\rho-\rho'))\le\opnorm{O}\trnorm{\rho-\rho'}$,
we can see that the algorithmic errors in expectation values are much (by order) smaller than observable error or state error and exhibit oscillatory behaviors.

\subsection{Noise undermines Trotter simulation}
Considering the effect of error decay, an interesting question is whether noise can reduce the accumulated error.
\cref{fig:worst_input} shows that noise can only improve the accumulated error in some extreme cases (regime of parameters).
We could validate it by taking the partial derivative of the accumulated error of the noisy second-order Trotter
\begin{equation}
    \eps_{2,\gamma}^{\acc} (r)
    = \sum_{d=1}^r 
    C \gamma e^{-c \gamma d} + 
    B \frac{t^3}{r^3} e^{-b \gamma d}
\end{equation}
with respect to $\gamma$, i.e.,
\begin{equation}\label{eq:partial_derivative}
    \partial_\gamma \eps_{2,\gamma}^{\acc} 
    = \sum_{d=1}^r 
    C (1-c\gamma d) e^{-c \gamma d} - 
    Bbd \frac{t^3}{r^3} e^{-b \gamma d}
    < 0
\end{equation}
and find the regime less than 0.

\begin{figure}[!t]
    \centering
    \sidesubfloat[]{\includegraphics[width=0.4\linewidth]{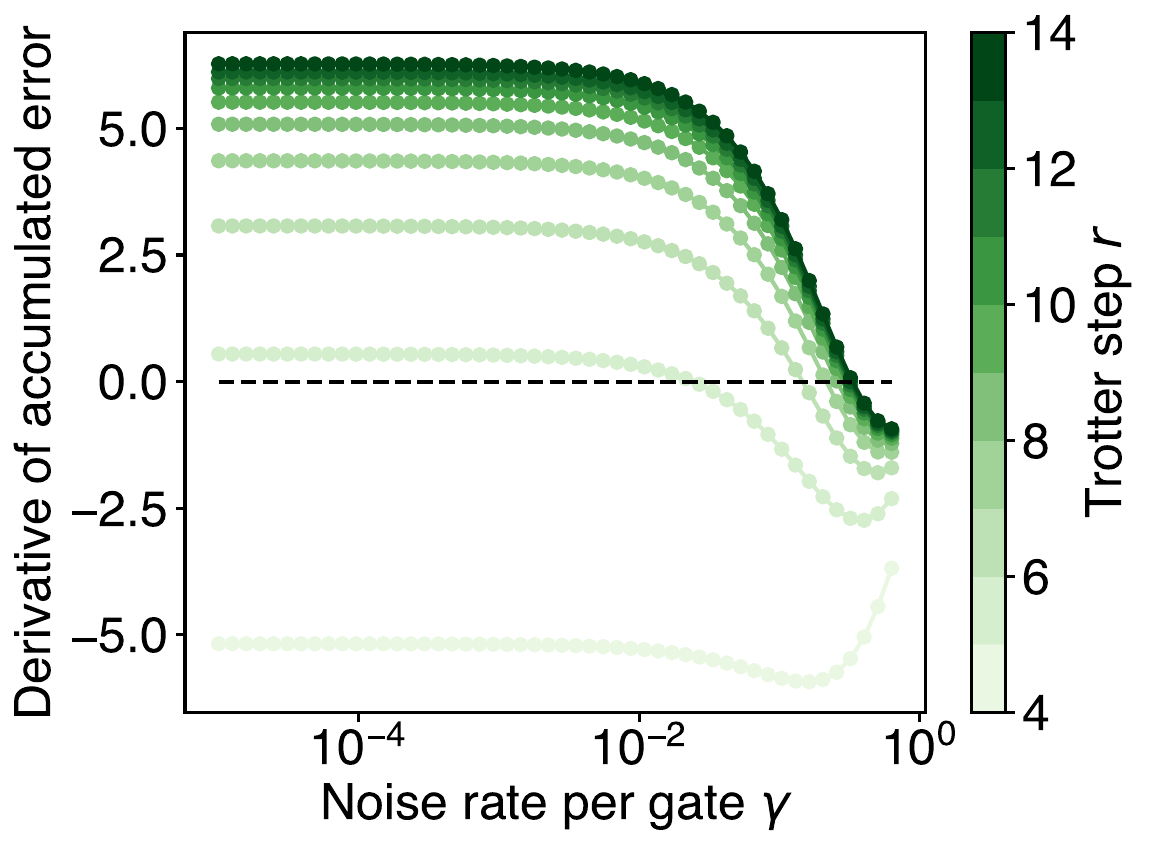}}
    \hspace{1cm}
    \sidesubfloat[]{\includegraphics[width=0.45\linewidth]{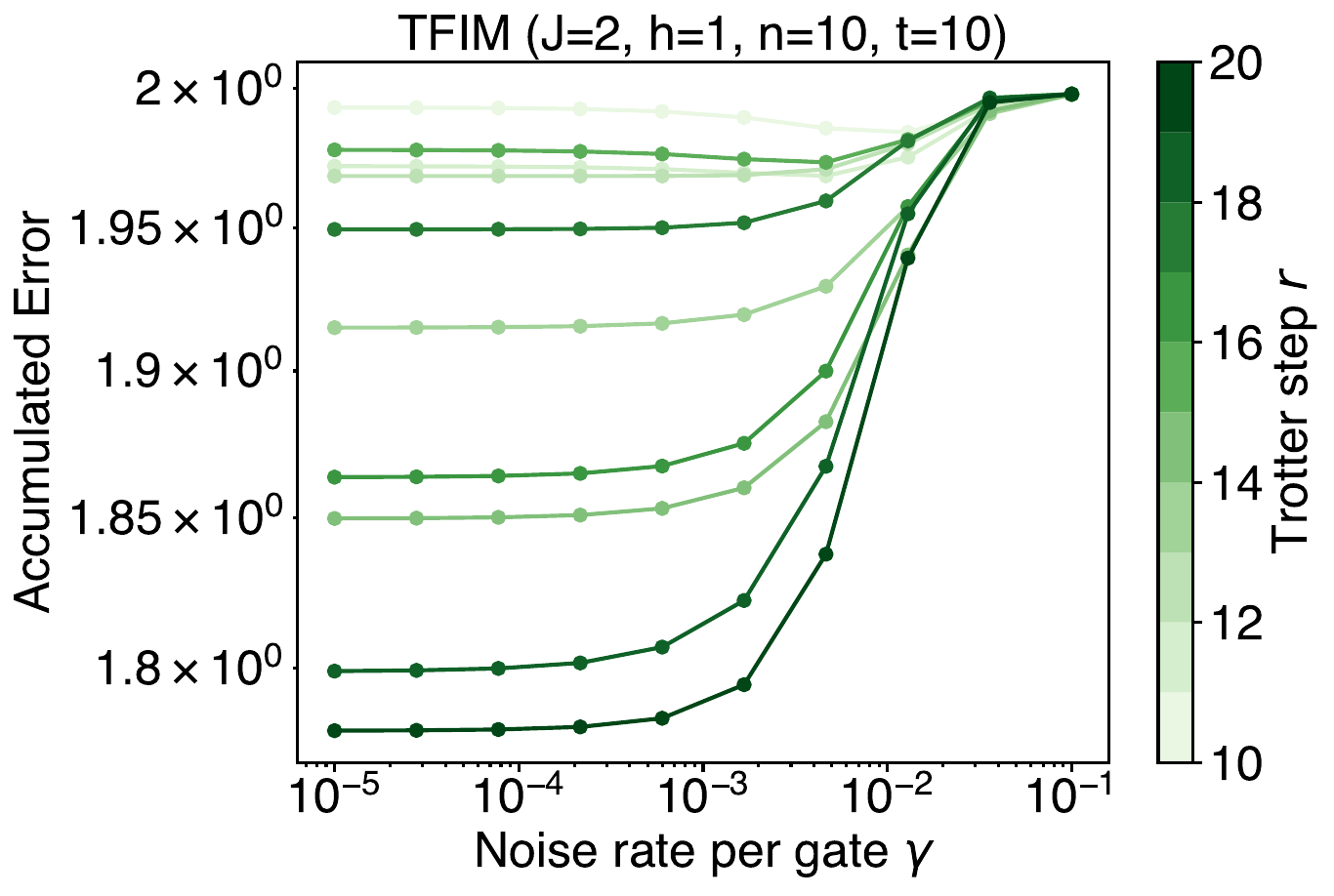}}
    \caption{ 
    (a) The partial derivative of the accumulated error regarding $\gamma$ (empirical formula \cref{eq:partial_derivative}) versus noise rate $\gamma$.
    Different Trotter steps $r$ are indicated by the colors. 
    We take the 10-qubit TFI Hamiltonian with the standard Trotter setup in the main text.
    Only when the Trotter number $r$ is small (light color), the derivative is negative, which implies the noise can reduce the accumulated error.
    (b) The empirical accumulated error V.S. noise rate $\gamma$ with different Trotter steps $r$.
    Only when $r$ is very small, the accumulated error decreases with increasing noise rates, which agrees with the derivative result of (a).
    However, in this regime, the accumulated error has reached the maximal value, which is not practical.
    }
    \label{fig:worst_input}
\end{figure}

\else
\fi

\end{document}